\documentclass[final]{siamltex}

\newcommand{\hide}[1]{}

\usepackage{amsmath,amssymb,epsfig}

\usepackage{times}

\usepackage{color}
\definecolor{Purple}{rgb}{0.5,0,0.5}
\usepackage[dvips,ps2pdf,pagebackref=true,pdffitwindow=true,bookmarks=false]{hyperref}
\hypersetup{pageanchor=true,colorlinks=true,citecolor=blue,urlcolor=Purple}

\newcommand{\nc}{\newcommand} 

\renewcommand{\iff}{{if\mbox{}f\ \mbox{}}}

\nc{\fig}[4]
{
  \begin{figure}[ht!]  
    \centering{\scalebox{#1}{\includegraphics*{./figures/#2.eps}}}
    \caption{#4}
    \label{fig:#3}
  \end{figure}
}

\nc{\nn}{\nonumber} 
\nc{\nit}{\noindent}
\nc{\marginnote}[1]{\marginpar{\tiny #1}}
\nc{\SH}{Schr\"odinger}
\nc{\NLS}{nonlinear Schr\"odinger}
\nc{\ie}{\emph{i.e.\ \mbox{}}}
\nc{\eg}{\emph{e.g.\ \mbox{}}}
\nc{\per}{{l_j}}
\nc{\perx}{q}
\nc{\vol}{{\rm Vol}}

\nc{\D}{\partial} 
\nc{\lapx}{\dfrac{\partial^2}{\partial x^2}}
\nc{\lapy}{\dfrac{\partial^2}{\partial y^2}}
\nc{\lapxy}{\dfrac{\partial^2}{\partial xy}}
\nc{\diff}[2]{\frac{d #1}{d #2}}
\nc{\diffn}[3]{\frac{d^{#3} #1}{d {#2}^{#3}}} 
\nc{\pdiff}[2]{\frac{\partial #1}{\partial #2}} 
\nc{\pdiffn}[3]{\frac{\partial^{#3} #1}{\partial{#2}^{#3}}} 

\def\Xint#1{\mathchoice
  {\XXint\displaystyle\textstyle{#1}}%
  {\XXint\textstyle\scriptstyle{#1}}%
  {\XXint\scriptstyle\scriptscriptstyle{#1}}%
  {\XXint\scriptscriptstyle\scriptscriptstyle{#1}}%
  \!\int}
\def\XXint#1#2#3{{\setbox0=\hbox{$#1{#2#3}{\int}$} \vcenter{\hbox{$#2#3$}}\kern-.5\wd0}}
\def\dashint{\Xint-}
\nc{\Av}{\dashint_\cell}
\nc{\avg}[1]{\mbox{$\left \langle\ \!#1\!\!\ \right \rangle$}}
\nc{\intRd}{\int_{{\mathbb R}^d}}

\nc{\abs}[1] {\lvert #1 \rvert} 
\nc{\norm}[2] {{\lVert #1 \rVert}_{#2}} 
\nc{\normH}[1]{\|#1\|_{H^1}^2}
\nc{\normL}[1]{\|#1\|_2^2}
\nc{\nl}[1]{|#1|^{2\sigma}}
\nc{\nlb}[1]{#1^{2\sigma+1}}
\nc{\Linvd}{L_\delta^{-1}}
\nc{\Linvz}{L_0^{-1}}
\nc{\Linvzs}{L_0^{-2}}

\DeclareMathOperator{\sech}{sech}

\nc{\st}[1]{\stackrel{(\ref{eq:#1})}{=}}
\nc{\stt}[2]{\stackrel{(\ref{eq:#1}),(\ref{eq:#2})}{=}}

\nc{\eps}{\epsilon}
\nc{\veps}{\epsilon}
\nc{\Ue}{U_\eps}
\nc{\koe}{\frac{k}{\veps}} 

\nc{\wrat}{w_{\rm ratio}}
\nc{\aij}{A^{ij}}                  
\nc{\minv}{m_*^{-1}} 
\nc{\sqminv}{m_*^{-\frac{1}{2}}}
\nc{\G}{\gamma_{\rm ef\,\!f}}
\nc{\zetap}{{\zeta_*}}
\nc{\zetas}{{\zeta_{1*}}}
\nc{\Pe}{P_{\rm edge}}
\nc{\slope}{\dfrac{d\cP[u(\cdot;\mu)] }{d\mu}}
\nc{\moot}{\mu_2}
\nc{\mum}{\mu_{\rm min}}

\def\R{\mathbb{R}}

\def\Z{\mathbb{Z}}
\def\N{\mathbb{N}}
\def\T{\mathbb{T}}

\nc{\bk}{{\bf k}}
\nc{\bq}{{\bf q}}
\nc{\br}{{\bf r}}
\nc{\bx}{{\bf x}}
\nc{\by}{{\bf y}}
\nc{\bz}{{\bf z}}
\nc{\bl}{{\bf l}}
\nc{\bnu}{{\bf \nu}}
\nc{\bxc}{{\bf x}_c}
\nc{\bxi}{{\mathbf \xi}}

\nc{\cP}{{\cal P}}
\nc{\cPedge}{{\cal P}_{edge}}
\nc{\cPc}{{\cal P}_{cr}} 
\nc{\cell}{{\cal B}}
\nc{\cDD}{\Lambda} 
\nc{\cF}{{\cal F}}
\nc{\cG}{{\cal G}}
\nc{\cO}{{\cal O}}  
\nc{\cQ}{{\cal Q}}  
\nc{\cR}{{\cal R}}
\nc{\cI}{{\cal I}}
\nc{\cK}{{\cal K}}
\nc{\cL}{{\cal L}} 
\nc{\cM}{{\cal M}}
\nc{\cN}{{\cal N}} 
\nc{\cE}{{\cal E}}
\nc{\cH}{{\cal H}} 
\nc{\cX}{{\cal X}}
\nc{\cZ}{{\cal Z}} 
\nc{\cT}{{\cal T}} 
\nc{\order}{{\cal O}}

\newtheorem{remark}{Remark}[section]

\title{Band-edge solitons, Nonlinear Schr\"odinger / Gross-Pitaevskii Equations and Effective Media}

\author{
B. Ilan\thanks{School of Natural Sciences, 
University of California, Merced, CA; \href{bilan@ucmerced.edu}{bilan@ucmerced.edu} } 
\and M.\ I. Weinstein\thanks{Department of Applied Physics and Applied  Mathematics,  
Columbia University, NY, NY;
\href{miw2103@columbia.edu}{miw2103@columbia.edu} }
 }

\begin{document} 

\maketitle

\begin{abstract}
We consider a class of nonlinear Schr\"odinger / Gross-Pitaevskii
(NLS/GP) equations with periodic potentials, having an even symmetry. 
We construct  ``solitons'', centered about any point of symmetry of
the potential. 
For focusing (attractive) nonlinearities,  these solutions bifurcate 
from the zero state at the lowest band edge frequency,  
into the semi-infinite spectral gap. Our results extend  to
bifurcations into finite spectral gaps, for focusing or defocusing (repulsive) 
nonlinearities under more restrictive hypotheses.

Soliton nonlinear bound states with frequencies near a band edge 
are well-approximated by a slowly decaying solution of a   
{\it homogenized NLS/GP equation}, with constant homogenized 
effective mass tensor and  effective nonlinear coupling coefficient, 
modulated by a Bloch state. 

For the critical NLS equation with a periodic potential, 
{\it e.g.} the cubic two dimensional NLS/GP with a periodic potential,  
our results imply:\\
$\bullet$\ The limiting soliton squared $L^2$ norm, as the 
spectral band edge frequency is approached, is equal to
 $\cP_{edge}=\zetap\times\cP_{cr}$, where  $\cP_{cr}$ denotes 
the minimal mass soliton of the translation invariant critical NLS. 
$\cP_{cr}$ is also known as the Townes critical power for self-focusing of optical beams. \\
$\bullet$\ The constant $\zetap$ 
  is  expressible in terms of the band edge Bloch eigenfunction 
and the determinant  of the {\it  effective mass} tensor.  
f the potential is non-constant, then $0<\zetap<1$ and $\cP_{edge}$  
is strictly less than, $\cP_{cr}$.\\  
The results are confirmed by numerical computation of bound states 
with frequencies near the spectral band edge.
  
Finally, these results have implications for the control of nonlinear waves using periodic structures. 
\end{abstract}

\begin{keywords}
  Multiple Scales,  Lyapunov-Schmidt reduction, nonlinear optics, Bose-Einstein condensates.
\end{keywords}

\begin{AMS} 
35B27; 
35B32; 
35B35; 
35B40. 
\end{AMS}



\section{Introduction and Outline}

"Solitons'' are spatially localized concentrations of energy, which are of great interest in many nonlinear wave systems.
 They arise from a balance of dispersion (or diffraction), which tends to spread energy and (focusing / attractive) nonlinearity which tends to concentrate energy. Although their importance was first recognized in the context of hydrodynamics \cite{Whitham,Ablowitz-Segur}, soliton-like coherent structures are now understood to play a central role  in contexts ranging from optical pulses (temporal solitons) to stationary beams (spatial solitons) of nonlinear optics 
 \cite{Boyd,Newell-Maloney} to soliton matter waves in macroscopic quantum systems \cite{PS:03}.
 Advances in the design of micro- or nano-structured media have greatly enabled the control of optical and  matter waves. 
Thus it is of interest to develop a fundamental understanding of the effect of inhomogeneities in a medium on the dynamics of nonlinear dispersive waves and, in particular, on the dynamics of solitons. See, for example,  \cite{Segev-etal:03} for an experimental investigation of solitons in periodic structures.
\\ 

In this article we consider solitons in non-homogeneous media  governed by a class of nonlinear Schr\"odinger / Gross-Pitaevskii (NLS / GP) equations:
\begin{equation}
  \label{eq:NLS-V} 
  i \D_t\psi\ = -\Delta \psi + V(\bx) \psi  - |\psi|^{2\sigma} \psi~.
\end{equation}
Here $\psi=\psi(\bx,t)$ denotes a complex-valued function of $(\bx,t)\in\R^d\times\R^1,\ d\ge1$.  
The potential $V(\bx)$ is assumed to be a real-valued, smooth, periodic and symmetric about one or more points
\footnote{ 
Periodic potentials are often called ``lattice'' potentials.}
\footnote{The main results of this paper extend to general nonlinearities  of the form 
\mbox{ $f(|\psi|^2)\psi = g\left[\ 1+\cO(|\psi|^2)\ \right]  |\psi|^{2\sigma}\psi$. } }

NLS / GP is a Hamiltonian system, expressible as:
\begin{align}
i\D_t\psi\ &=\ \frac{\delta\ \cH}{\delta\ \psi^*}\nn\\ 
\cH[\psi,\psi^*]\ &=\ \int \nabla\psi\cdot \nabla\psi^*\ +\ V(\bx)\psi\psi^*-\frac{1}{\sigma+1}\psi^{\sigma+1}\ (\psi^*)^{\sigma+1},
\label{ham}
\end{align}
where $\psi^*$ denotes the complex conjugate of $\psi$.
By Noether's Theorem, the invariance $t\mapsto t+t_0$ implies the  time-invariance of $\cH$ for solutions of NLS / GP. Furthermore, the invariance $\psi\mapsto e^{i\theta}\psi$ implies the additional time-invariant quantity:
\begin{equation}
\cP[\psi,\psi^*]\ =\ \int \psi\ \psi^*\ d\bx\ .
\label{power}
\end{equation}
The parameter $\sigma>0$ allows for variation of the strength of the  nonlinearity. In physical systems, we typically have $\sigma=1$. 
Allowing $\sigma$  to vary enables one to quantify the balance between nonlinear effects and dispersive / diffractive effects, which depend on spatial dimensionality, $d$. Local well-posedness in time for the initial value problem for \eqref{eq:NLS-V} with data 
\begin{equation}
\psi(\bx,t=0)\ =\ \psi_0(\bx)
\label{data}\end{equation}
in $\psi_0\in H^1(\R^d)$ (see, for example,  
\cite{Cazenave,SulemSulem})
holds for all $\sigma>0$ for $d=1,2$ and all 
$0<\sigma<2(d-2)^{-1}$ for  $d\ge3$. Global well-posedness for arbitrary data holds for $\sigma<2/d$. For well-posedness for data in spaces of weaker regularity, see \cite{Bourgain-NLS,Tao-NLS}.

The NLS/GP equation, ~(\ref{eq:NLS-V}), with $\sigma=1$ governs the  dynamics of 
 a macroscopic quantum state, a Bose-Einstein condensate, comprised of a large collection of interacting bosons in the mean-field limit \cite{PS:03,ESY:07,C-GFK:08}. The attractive nonlinear potential, $-|\psi|^2$, corresponds to a species of bosons, whose two-particle interactions have a negative scattering length. 
 A second important area of application of NLS/GP is its description of the evolution of the slowly varying envelope of a  stationary and nearly monochromatic laser beam propagating 
through a nonlinear medium \cite{Newell-Maloney,Boyd}. Here, the attractive nonlinear potential is due to the Kerr nonlinear effect; regions of higher electric field intensity have a higher refractive index. In this setting, $t$ denotes the distance along the direction of propagation and $\bx\in\R^2$ the transverse dimensions.
In the quantum mechanical setting the potential, $V(\bx)$,  is determined by  magnetic and optical effects 
that are used to confine a cloud of bosons. In 
 optics, the potential is determined by the spatial variations of the background {\it linear  refractive index} of the medium. 
The functional $\cP$ denotes the optical power or, in the quantum setting, the particle number.

Nonlinear bound states or {\it solitons}
of NLS/GP are solutions of standing wave type:
\begin{equation}
\psi(\bx,t)=e^{-i\mu t}u(\bx,\mu)
\label{sol}\end{equation}
where $\mu$ denotes the frequency (propagation constant in optics, 
chemical potential in quantum many-body theory)
and $u$ is a real-valued  solution of
\begin{align}
  \left(\ -\Delta_\bx\ + V(\bx)\ \right) u(\bx,\mu) \ -\ u^{2\sigma+1}(\bx,\mu) =  \mu\ u(\bx,\mu),\ \ u(\cdot,\mu) \in H^1(\R^d).\nn\\
  &\label{eq:u-nd}
\end{align}
 We shall construct solutions of \eqref{eq:u-nd} with $\mu$ located in a spectral gap of $-\Delta+V$.\\
 
 The properties of solitons in {\it homogeneous} media, $V\equiv0$,  
are reviewed  in detail in  Section~\S\ref{sec:background}. Briefly, for  $\sigma<2/d$ (subcritical nonlinearities) dynamically stable  solitons exist at any prescribed $L^2$ norm (in one  to one correspondence with any $\mu<0$).   In the  critical case and supercritical cases, $\sigma\ge 2/d$, solitons are unstable. \\
 
 We raise several motivating questions and outline our results
  in the next subsection. 
 \begin{enumerate}
 \item[(Q1)]\  {\bf Persistence and stability:}\ What is the effect of a periodic potential on the existence and stability properties of solitons?
 \item[(Q2)]\ {\bf Stabilization:}\ Can  unstable solitons be stabilized by a potential, $V(\bx)$? This question was first addressed in ~\cite{Rose-88} in the context of localized potentials and more recently for more general potentials~\cite{SFIW-08}, {\it e.g.} periodic, quasiperiodic. 
\item[(Q3)]\ {\bf Excitation thresholds / Minimal mass solitons:}\
{\it How does periodic structure effect soliton excitation  thresholds?}\ 
For critical nonlinearity, $\sigma=2/d$, and $V\equiv0$ the soliton
squared $L^2$ norm is independent of $\mu$; $\cP[u(\cdot,\mu)]=\cP_{cr}$.
Thus, {\it  there is an $L^2$ threshold below which there are no  solitons}. 
This $L^2$ {\it excitation threshold} for soliton formation is of
great physical
interest~\cite{Weinstein:99,FKM:97,CP:03,MRS:preprint,BCT:09}. 
In optics it corresponds to the critical {\it power} for self-focusing~\cite{Townes:64,Moll-Fibich-Gaeta:03}. Such solitons are also often called {\it minimal mass solitons}. See Remark \ref{rmk:sol-exc-conj} .
\end{enumerate}


\subsection{Outline of Results}
In order to outline the results of this paper, we begin with very quick review  of the spectral theory of  Schr\"odinger operators, $-\Delta +V$, for $V$ periodic \cite{RS4,Eastham:73,Kuchment-01}.  If $V(\bx)$ is a periodic potential, then the spectrum of $-\Delta+V$ is real, bounded below, tends to positive infinity, is absolutely continuous and consists of the union of closed intervals\ ( {\it spectral bands}). The open intervals separating the spectral bands are called {\it spectral gaps}. One dimensional Schr\"odinger operators with periodic potentials generically have infinitely many gaps. In dimensions $d\ge2$, there are at most finitely many gaps. 

We denote by $E_*$  lowest point in the spectrum, the left endpoint or edge of the first spectral band. $E_*$ is simple and  is the ground state (lowest) eigenvalue of $-\Delta+V$, subject to periodic boundary conditions on the basic period cell of $V$.
The eigenspace associated with $E_*$ is spanned by $w(\bx)$, a non-trivial solution of:
\begin{equation}
\left(-\Delta + V(\bx)\right) w(\bx) = E_* w(\bx),\ \ \ w(\bx)\ {\rm periodic}
\nn\end{equation}
  For the case   $V\equiv0$, $E_*=0$  and we can take $w(\bx)\equiv1$.\\

The present work considers the bifurcation and dynamic stability properties of families of solitons emerging from a spectral band edge. Such edge-bifurcating solitons have a multi-scale character described below.   
Our results include the following:\\ 
\begin{enumerate}
\item {\bf  Theorem \ref{theo:2scale}:}\ Let $\bx_0$ denote any point of symmetry of  $V(\bx)$
\footnote{{\it i.e.} $f(\bx)=f(x_1,\dots,x_d)$ is symmetric (about the origin) if 
 $f(x_1,\dots,x_d)\ =\ f(s_1x_1,\dots,s_d x_d),\ \ s_j=\pm1$.\ \ 
$\bx_0$ is a point of symmetry of $V(\bx)$
   if  $\tilde{V}(\bz)\equiv V(\bx_0+\bz)$ is symmetric. Thus, by translating coordinates, we can arrange for a point of symmetry to be at the origin.} .
   
\nit There is a family 
\begin{equation}
\mu\mapsto u(\bx,\mu\ ) \approx\
 \left(E_*-\mu\right)^{1\over2\sigma}\ F\left(\ \sqrt{E_*-\mu}\ (\bx-\bx_0)\ \right)\  w(\bx),
\nn\end{equation}
which bifurcates from the zero solution at energy $E_*$ into the semi-infinite gap $(-\infty,E_*)$ for $0<E_*-\mu$ sufficiently small.
Here, $\sigma\in\N$ for $d=1,2$ and $\sigma=1$ for $d=3$.

 $F(\by)$ denotes the soliton profile for an effective medium with 
effective mass tensor, $A^{ij}$ given by (\ref{eq:Aij})  and effective nonlinear coupling constant,
$\G$ given in (\ref{eq:geff}), and satisfies the \emph{homogenized soliton equation}:
 \begin{align}
&      -\sum_{i,j=1}^d\ \D_{y_i} \aij \D_{y_j} F(\by)\ - \G\ F^{2\sigma+1}(\by)\ =\ - F(\by)~.
      \label{eq:Feqn-intro}\\
 &     F>0,\ \ \ F\in H^1(\R^d)\nn
\end{align} 
The leading order expansion is constructed via multiple scale
expansion. The error term is studied by decomposition of the corrector
into spectral components near and far from the band edge, and
estimated via  a  Lyapunov-Schmidt strategy; see also
\cite{Busch-etal:06,DPS:09b,DU:09}.  The results can be extended to
solitons near edges of {\it finite} spectral bands for focusing and
defocusing nonlinear potentials under more restrictive hypotheses on
$V$; see Section~\S\ref{section:finite-gaps}. 
A variant of Theorem \ref{theo:2scale} holds in dimension one, in any spectral gap, near a ``positive curvature'' band edge; see Theorem
 \ref {theo:1d2scale}.
\item {\bf Corollary \ref{cor:edgepower-crit},\ part 1\ :\  }\  
Consider the critical cases:  $\sigma=1,\ d=2$ and $\sigma=2,\ d=1$. Near the band edge, {\it i.e.} for $E_*-\mu>0$ and small, we have:
 \begin{equation}
 \cP\left[u(\cdot,\mu)\right]\ =\ \zetap\ \cP_{cr}\ +\ (\mu-E_*) \ \zetas\ +\ \cO\left(\ (\mu-E_*)^2\ \right).
 \label{cP-expand}\end{equation}
  Here, 
 \begin{enumerate}
 \item $\cP_{cr}=\cP[R(\cdot,-1)]$, where $R(\cdot,-1)$ denotes the unique (up to translations)  solution of 
 \begin{equation}
 \Delta R\ -\ R\ + R^{\frac{4}{d}+1}\ =\ 0,\ \ R>0,\ \ R\in H^1\ ;
 \label{R-eqn}\end{equation}
 see  (Q3) above.
 \item 
 \begin{equation}
 \cP_{edge} \equiv\ \lim_{\mu\to E_*}\cP[u(\cdot,\mu)] = \zetap\ \cP_{cr}
\nn\end{equation}
 is given by (\ref{eq:zeta*}) and satisfies the inequality  $0<\zetap<1$, unless $V$ is identically constant. In the latter case, $\zetap=1$.
  \item 
  \begin{equation}
  \zetas\ \equiv\ \left. \frac{d}{d\mu}\right|_{\mu=E_*} \cP[u(\cdot,\mu)]
  \label{zeta1-def}
  \end{equation}
   is given by (\ref{eq:zeta1*}). \\
 For periodic potentials, $V$,  of the form $\delta\ \tilde{V}(x)$,
 where $\delta$ is sufficiently small, we show that   $\zetas>0$; 
 see \eqref{eq:zeta1_delta_d1}. \\
 {\bf Positive slope conjecture:}\ In general, $\zetas>0$. 
  \end{enumerate} 
  Both $\zetap$ and $\zetas$  depend on the edge (periodic) Bloch eigenstate and the 
 Hessian matrix (of $2^{\rm nd}$ partial derivatives) of the band dispersion function, $D^2E_1$,  near $E_*$. The latter is often called the inverse {\it effective mass} tensor.  
  \item {\bf Instability for $\mu$ near the band edge:}\  Consider the critical cases $\sigma=1,\ d=2$ and $\sigma=2,\ d=1$. For $V$ nonzero and $\mu$ near $E_*$, nonlinear bound states are linearly exponentially unstable, provided $\zetas>0$. We conjecture $\zetas>0$, in general, and have verified it for potentials $V=\delta \tilde{V}$, with $\delta$ sufficiently small.
  
  For $V\equiv0$, the linear instability is algebraic, although for the nonlinear dynamics, solutions can blow up in finite time or decay to zero dispersively (diffractively) as $t$ tends to infinity. In contrast, since for $V\ne 0$, for $\mu$ close enough to the spectral band edge, the curve $\mu\mapsto \cP[u(\cdot,\mu)]$ lies \underline{below} the line $\cP=\cP_{cr}$, although a solution with data near a  soliton, $u(\cdot,\mu)$, with frequency near the band edge does not remain nearby in $H^1$, the solution exists globally in time in $H^1$.
   \item {\bf Section~\S\ref{sec:numerics}:\ }\ 
  Numerical computations are used to illustrate the asymptotic results 
  and to study the global behavior.
\end{enumerate}

Figure \ref{fig:power_mu_1d_A1_v2} below summarizes a key consequence
of our results. 

\fig{0.65}{power_mu_1d_A1_v2}{power_mu_1d_A1_v2}{
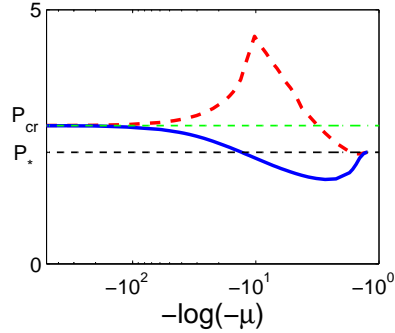
 Plot of power curves: $\cP[u(\cdot,\mu)]$ vs. $\mu$ (using semi-log axis) for the quintic one-dimensional NLS/GP equation, \eqref{eq:bs1dcrit},  with $V_0=10$ and $K=2\pi$ (here $E_*\approx -1.23$). Solid (blue) curve corresponds to power curve for soliton family
centered at a local minimum. Dashed (red) curve corresponds to centering at a  local maximum. Agreement is shown between numerical computations and the analytically obtained value 
for the band edge power (dashed / black horizontal line),
$\cP_{edge}=\lim_{\mu\to E_*}\cP[u(\cdot,\mu)]=\zetap\times\cP_{cr}\approx 2.2$ [Eq.~(\ref{eq:power_edge})].  For $\mu$ large and negative  $\cP(\mu)$ converges to $\cP_{cr}\approx 2.72$ (dashed / green horizontal line), which is 
the critical power of the Townes soliton in translation 
 invariant ($V\equiv$ const) case. 
}

For any nontrivial periodic $V(x)$ the limiting $L^2$ norm at the band edge is {\bf strictly less} than that of the homogeneous medium. The slope of the curve, $\mu\to\cP\left[u(\cdot,\mu)\right]$ is strictly positive.  As  $-\mu=|\mu|$ increases, solitons become increasingly localized in space, and thus depend more and more on the local properties of $V$. The limiting ($|\mu|\to\infty$) squared $L^2$ norm is $\cP_{cr}$. The orbital stability theory, outlined in Section~\S\ref{sec:background} implies that solitons with energies $\mu$ near the band edge (where $\D_\mu\cP\left[u(\cdot,\mu)\right]>0$) are unstable,  while those which are centered and sufficiently concentrated  ($-\mu$ sufficiently large, $\D_\mu\cP\left[u(\cdot,\mu)\right]<0$) about a local minimum of $V$ are stable. It is natural to conjecture that for localized initial conditions with $L^2$ norm strictly less than $\inf_{\mu\le E_*}\cP\left[u(\cdot,\mu)\right]$, solutions to the initial value problem disperse to zero as $t\to\infty$; see the discussion in the proof of part 3 of Theorem \ref{theo:crit-dynamics} and \cite{Weinstein:89}.

\begin{remark}  
  Concerning the dependence of  
  $\mu\mapsto\cP[u(\cdot,\mu)]$ for $\mu$ near the band edge, on the nonlinearity parameter,
  $\sigma$, and dimensionality $d$ (see Theorem
  \ref{theo:edgepower} and Corollary \ref{cor:edgepower-crit}) 
  it  is useful to recall the analogous behavior in the translation invariant case: $V\equiv0$.  In this case, NLS is also invariant under dilation:
  \begin{equation}
    \psi(\bx,t) \ \mapsto \ \lambda^{\frac{1}{\sigma}}\psi(\lambda\bx,\lambda^2t)~.
    \label{dilation}
  \end{equation}
  Let $R(\cdot,\mu)$ denote the positive (unique up to translation), solution of 
  \begin{equation}
    -\Delta R - R^{2\sigma+1} = \mu R,\ \ \  
    \nn\end{equation}
  By uniqueness  
  \begin{equation}
    R(\bx,\mu) \ =\ |\mu|^{\frac{1}{2\sigma}}R(|\mu|^{\frac{1}{2}}\bx,-1) 
    \label{eq:R_scale}
  \end{equation}
  It follows that
  \begin{equation}
    \cP[R(\cdot,\mu)]=\|R(\cdot,\mu)\|^2_2=|\mu|^{\frac{1}{\sigma}-\frac{d}{2}}\ \|R(\cdot,-1)\|^2_2
    \nn\end{equation}
  implying that as $\mu\to E_*$, $\cP[R(\cdot,\mu)]$ tends 
  \begin{enumerate}
  \item[$\bullet$] to $0$, for $\sigma<2/d$
  \item[$\bullet$] to  $\|R(\cdot;-1)\|^2_2$, for $\sigma=2/d$, and 
  \item[$\bullet$] to $+\infty$, for $\sigma>2/d$;
  \end{enumerate}
  see figure~\ref{fig:bifurcate}.
  In one space dimension, the family of solitons is given explicitly by:
  \begin{equation}
    \label{eq:R_1d}
    R(x,\mu) = \left[ \ (\sigma+1)\ |\mu| \ \right]^{\frac{1}{2\sigma}}\sech^{\frac{1}{\sigma}}
    \left( \ \sigma \sqrt{|\mu|}\ x\ \right)~.
  \end{equation}
  In the critical case, $\sigma=2$,
  \begin{equation}
    \cP[R(\cdot,\mu)]  = \frac{\sqrt{3}}{2}\ \int_\R \sech(y)\ dy = \frac{\sqrt{3}}{2}\ \pi\ \sim\  2.7207\ \ ;
    \nn\end{equation}
  see figure \ref{fig:bifurcate}.
  Theorem \ref{theo:edgepower} implies a similar trichotomy of behaviors
  for states bifurcating from the band edge, $E_*$, of a  periodic
  potential. Also, for $\sigma=2/d$, the curves $\mu\mapsto
  \cP[u(\cdot,\mu)]$ in figure \ref{fig:power_mu_1d}, are seen to be
  deformations (for minimum and maximum centered solitons) of  the
  horizontal line $\mu\mapsto \cP_{cr}$ for the case $V\equiv0$.
\end{remark}

\fig{0.6}{bifurcate}{bifurcate}{ 
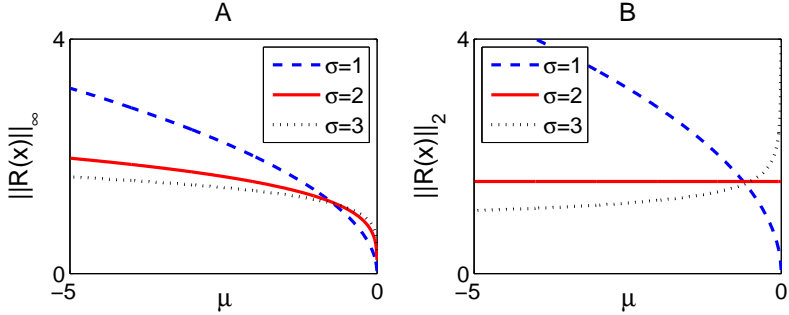
  (A) $L^\infty$ norm and (B) squared $L^2$ norm ($\cP$)
  as functions of frequency, $\mu$, 
  for the ground state solution of Eq.~(\ref{eq:u-nd}) in one dimension ($d=1$)
  with three nonlinear exponents: subcritical ($\sigma d=1<2$), critical ($\sigma d= 2$), and supercritical ($\sigma d=3>2$); see legend. In $L^\infty$, bifurcation appears from a state with zero norm at $E_*=0$.
   In $L^2$, the limiting behavior as $\mu\to 0^-$ depends on $\sigma d$.
}

\nit {\bf Previous work:} 
Formal expansions and numerical approximation of nonlinear bound
states near spectral band edges for periodic and aperiodic structures
and their linearized stability properties  were presented 
in~\cite{Steel-98,BKS:02,Yang-Musslimani-03,Pelinovsky-04,Brazhnyi-04,BMS:04,SM:04-I, SM:04-II,Cheng-05,
AISP:06, Brazhnyi-06, AKP:07,ShiYang-07,Staliunas-07}.
The band edge limit of $\cP$, for case of a 2-dimensional separable
potential was obtained by formal perturbation theory and numerically in \cite{Shi-Wang-Chen-Yang:08}.   
Two-scale convergence methods have been applied to  
rigorously derive homogenized effective equations, 
valid on large but finite time scales, in~\cite{Allaire-05}, for the
linear Schr\"odinger equation, and in~\cite{Sparber-06} for 
the time-dependent NLS/GP,  with two-scale type initial conditions.  
Bifurcation of localized states from the continuous spectrum 
into spectral gaps has been considered  
in~\cite{Kupper-Stuart:90,Kupper-Stuart:92,Heinz-Stuart:92a,Heinz-Stuart:92b,Alama-Li:92,Stuart,Pankov:05}. 
The connection with nonlinear coupled mode equations is explored 
in~\cite{Busch-etal:06,Pel-Sch:07,DPS:09b,DU:09}. 
The Lyapunov-Schmidt strategy applied herein is motivated by these latter approaches.
\\ \\
\nit{\bf Outline:}\  The paper is structured as follows. 
In Section~\S\ref{sec:background} we discuss  background 
for the formulation of our results. We state our main results in  
Section~\S\ref{sec:mainresults}. In Section~\S\ref{sec:homog-expansion} a formal
homogenization / two-scale expansion of solitons with frequencies near
the band edge is derived. The expansion and error estimates are proved
in Section~\S\ref{sec:error-est}. In Section~\S\ref{sec:edgepower} we
derive the consequences of our expansion of  band-edge solitons for
the character of $\cP[u(\cdot,\mu)]$ as $\mu\to E_*$.  
Section~\S\ref{sec:numerics} contains a discussion of
    numerical simulations illustrating our main theorems. 
Section~\S\ref{sec:summary} contains a short summary and discussion.  
The latter sections of the paper are appendices  containing technical results on the effective mass tensor.\\


\subsection{Notation}

\begin{enumerate}
\item We shall write $\eta(\eps)=\cO(\eps^\infty)$  if
$\eta(\eps)=\cO(\eps^q)$ for all $q\ge1$.
\item Fourier transform of $G$:\ \  $\hat{G}(\bk)\ =\ \int e^{-2\pi i\bk\cdot \bx}\ G(\bx)\ d\bx$
\item $\chi(a\le |\bk|\le b)\ =$ characteristic function of the set $\{\bk: a\le |\bk|\le b\}$
  \item $\chi(|\nabla_\by|\le a)\ G\ =\ \int\ e^{2\pi i\bk\cdot \by}\ \chi(|\bk|\le a)\ \hat{G}(\bk)\ d\bk$
  \item $H^s$, Sobolev space of order $s$; $H^s_{even}$ space of even
 $H^s$ functions
 \begin{equation}
 \| f\|^2_{H^s}\ =\ \sum_{|\alpha|\le s}\ \|\D^\alpha f\|_{L^2}^2\ \sim
  \| \hat f\|_{L^{2,s}}^2
  \nn\end{equation}
 \item $H^s_{sym}$, symmetric $H^s$  functions,\\ {\it i.e.} $f\in H^s_{sym}$ if $f\in H^s$ and 
 $f(x_1,\dots,x_d)\ =\ f(s_1x_1,\dots,s_d x_d),\ s_j=\pm1$.
\item $\| f \|^2_{L^{2,s}(D)}\ =\  \int_D\ |f(z)|^2\  (1+|z|^2)^s\ dz$
\item $C^m_\downarrow(\R^d)$, functions in $C^m(\R^d)$ with limit equal to zero as $|\bx|\to\infty$
 \item $\cell$ denotes the fundamental period cell and $\cell^*$ the dual fundamental cell, or Brillouin zone.
\end{enumerate}


\section{Background}\label{sec:background}

\subsection{Solitons and Stability Theory}
\label{sec:stability}

We give a very brief review of the stability theory of 
solitons of NLS/GP, ~(\ref{eq:NLS-V}).
\begin{definition}
\label{def:stable} 
The nonlinear bound state $u(\bx,\mu)$ of  NLS/GP 
is orbitally stable if for all $\varepsilon>0$, there is a $\delta>0$ such that  if the initial condition $\psi_0$  satisfies 
\begin{equation}
 \inf_{\gamma\in [0,2\pi)}  
\|\psi_0 - u(\cdot,\mu) e^{i\gamma} \|_{H^1} < \delta,
\nn\end{equation} then the
corresponding solution, $\psi(\cdot,t)$,  satisfies
\begin{equation}
\inf_{\gamma\in [0,2\pi)}
\|\psi(\cdot,t) - u(\cdot,\mu)e^{i\gamma} \|_{H^1} < \varepsilon,\ 
 {\rm for\ all\ }\ t\ne0. 
\nn\end{equation} 
\end{definition}
\nit This notion of soliton stability for NLS is natural since NLS/GP, for $V$ non-constant is invariant under the group of phase translations, $\psi\mapsto e^{i\theta}\psi$, but not spatial translations.\\ 

A central role in the stability theory is played by the operator
\begin{equation}
    \label{eq:L+}
    L_+ \equiv - \Delta - \mu + V - (2\sigma+1)u^{2\sigma},
  \end{equation}
the real part of the linearization of NLS/GP about $u(\cdot,\mu)$.  
  Let $n_-(L_+)$ denote the number of negative eigenvalues of $L_+$.
  If $u(\bx,\mu)$ is a nonlinear bound state with $\mu<E_*$ 
  (frequency lying in the semi-infinite gap) then 
$n_-(L_+)<\infty$ and the following 
  nonlinear stability theorem holds 
  ~\cite{Weinstein-86,Rose-88,Weinstein:89,GSS:87,SFIW-08}

\begin{theorem}
  \label{theo:stability} 
 \begin{enumerate}
 \item Let $u(\bx,\mu)$ denote a positive soliton 
  solution of NLS/GP with  $\mu$ 
  in the semi-infinite gap $(-\infty,E_*)$.
   The nonlinear bound state,  $\psi(\bx,t) = u(\bx,\mu)e^{- i \mu t}$
   is  orbitally  stable  if 
   the following two conditions hold:
  \subitem (a)
   {\bf Slope} (VK) {\bf condition}:
    \begin{equation}
      \nonumber
      \frac{d}{d\mu}\ \cP[u(\cdot,\mu)]< 0,\ \ \ {\rm and}
    \end{equation}
    \subitem (b)
   {\bf Spectral condition}: $L_+$ has no zero eigenvalues
    and
    \begin{equation}
      \label{eq:spectral}
      n_-(L_+) = 1.
    \end{equation}
\item  If either $\D_\mu\ \cP[u(\cdot,\mu)]> 0$ or $n_-(L_+) \geq 2$ then the 
  soliton is unstable (nonlinearly unstable as well as linearly exponentially unstable).
    \end{enumerate}
\end{theorem}
\begin{remark}
As discussed in \cite{SFIW-08},  the spectral condition can be  associated 
with the suppression of a {\it drift instability} and  the slope condition with the suppression of an energy-concentrating {\it  self-focusing} instability. 
\end{remark}


\subsection{Spectral theory for periodic potentials}\label{sec:FBtheory}

We consider the Schr\"odinger operator $-\Delta+V(\bx)$ 
acting in $L^2(\R^d)$, where $V(\bx)$ is smooth, real-valued potential which is 
periodic. That is,  $V(\bx+\bq)=V(\bx)$ for all $\bx\in\R^d$.  Here $\bq=\{\bq_1,\dots,\bq_d\}$ 
denotes a linearly independent set of vectors in $\R^d$ that
spans (over the integers) a lattice denoted by $\Gamma$.
The set 
\begin{equation}
\cell=\left\{\sum_{j=1}^d\ v_j\bq_j:v_j\in \left[-\frac{1}{2},\frac{1}{2}\right] 
\right\}
\label{fund-cell}
\end{equation}
is called a fundamental period cell.
The first Brillouin zone $\cell^*$ is generated 
by the dual basis $(\br_1,\dots,\br_d)$ given by 
$\br_j\cdot\bq_k= 2\pi\ \delta_{jk}$, \ie
$$
\cell^*=\left\{\sum_{j=1}^d\ v_j\br_j:v_j\in \left[-\frac{1}{2},\frac{1}{2}\right] \right\}
$$
and the dual lattice, $\Gamma^*$, is the integer span of $\cell^*$.

It is useful to review some well-known results of Floquet-Bloch theory~\cite{Kittel,RS4,Eastham:73,Kuchment-01}. 
The spectrum of $-\Delta+V$, denoted $\sigma (-\Delta + V)$, consists of a union of closed intervals called spectral {\it bands}
separated by {\it gaps} (also known as band gaps and photonic band gaps).
The spectral bands are characterized as follows. 

For each $\bk\in\cell^*$ we seek solutions of the linear eigenvalue problem
\begin{equation}
\left(\ -\Delta + V(\bx)\ \right)\ u =\ E\ u
\label{eq:linear-evp}
\end{equation}
of the form 
$ u(\bx;\bk)=e^{i\bk\cdot\bx} p(\bx;\bk)$, where $p(\bx;\bk)$ is
periodic in $\bx$: 
\begin{eqnarray*}
  \left[\ -\left(\nabla+i\bk\right)^2\  
    +\ V(\bx)\ \right]p(\bx;\bk)\ &=& \ E(\bk)p(\bx;\bk)~,
  \\
  p(\bx+\bq_j;\bk)\ &=& \ p(\bx;\bk),\ \ j=1,\dots,d~.  
\end{eqnarray*}
For each $\bk\in\cell^*$ this periodic elliptic boundary value problem 
has a sequence of discrete eigenvalues or  {\it band dispersion functions}
tending to positive infinity:\\ 
$E_1(\bk)\le E_2(\bk)\le\dots\le E_m(\bk)\le\dots$. 
As $\bk$ varies over the Brillouin zone $\cell^*$
each $E_m(\bk)$ sweeps out a closed subinterval of the real axis. 
The spectrum of $-\Delta+V$ acting on $L^2(\R^d)$ is the union of these
subintervals:
\begin{equation}
\sigma(-\Delta + V)\ =\ \cup_{m\ge1}\ \{E_m(\bk):\bk\in\cell^*\}\ \subset [\min_{\cell}V,\infty) 
\nn\end{equation}
and the states $\{u_m(\bx;\bk)=e^{i\bx\cdot\bk}p_m(\bx;\bk)\}$ are complete in the sense that 
\begin{equation}
f\in L^2(\R^d)\ \implies\ f(\bx)=\sum_{m\ge1}\int_{\cell^*}\ \langle u_m(\cdot;\bk),f(\cdot)\rangle_{L^2(\R^d)}\ u_m(\bx;\bk)\ d\bk
\label{completeness}
\end{equation}

We denote the lowest point in the spectrum of Eq.~(\ref{eq:linear-evp})
 and corresponding periodic eigenstate by
\begin{equation}
E_*=E_1(0),\ \ \ \ w(\bx)=p(\bx;\bk=0).
\nn\end{equation}
 $E_*$ is simple. 
   We will often make use of the relation
  \begin{equation}
    L_*w\ =\ 0~, \quad w>0~,\ \ w(\bx+\bq_j)=w(\bq_j)~,
    \label{eq:wdef}
  \end{equation}
  where 
\begin{equation}
L_*\ \equiv\ -\Delta + V-E_*\ \ .
\label{L*def}
\end{equation}
Thus, $w$ is the periodic ground state of $L_*$, $L_*\ge 0$ and $0$ is a simple eigenvalue of $L_*\,$ with kernel spanned by $w$. 
 Note that if $P^\perp$ is the orthogonal projection onto the subspace $\{w\}^\perp$, 
\begin{equation}
P^\perp\ g\ =\ g\ -\ \left\langle \frac{w}{\|w\|}, g\ \right\rangle\ \frac{w}{\|w\|},\ \ \ \langle f,g\rangle\ =\ \int_\cell \overline{f(\bx)}\ g(\bx)\ d\bx.\label{Pperp}\end{equation}
then $L_*^{-1}\ P^\perp$ is bounded 
on the space of $L^2$ periodic functions with fundamental period cell $\cell$.

Finally, note that we may, without loss of generality, restrict to the case where the fundamental period cell is $[-\pi,\pi]^d$. Indeed, if 
 ${\cal B}$ is the fundamental period cell (see \eqref{fund-cell}),  then define the constant matrix $Q$ to be the matrix whose $j^{th}$ column is $(2\pi)^{-1}\bq_j$. Then, under the change of coordinates  $\bx\mapsto\bz=Q\bx$, we have
\begin{align}
&-\nabla_\bx\cdot \nabla_\bx\ +\ V(\bx)\ {\rm acting\ on}\ L^2_{per}(\cell)\ \ {\rm transforms\ to}\nn\\
& -\nabla_\bz\cdot  \alpha\ \nabla_\bz
 + \tilde{V}(\bz) \equiv\ -\ \sum_{i,j=1}^d \alpha_{ij}\ \frac{\D^2}{\D z_i\D z_j}\ +\ \tilde{V}(\bz)\nn\\
& {\rm acting\ on}\ L^2_{per}\left([-\pi,\pi]^d\right)\ 
 {\rm where}\nn\\ 
 &\alpha= \frac{ Q Q^T }{|\det{Q}|},\ \ \tilde{V}(\bz)=V\left(Q^{-1}\bz\right),\  \ \bx=Q^{-1}\bz\ .
 \nn\end{align}


 \section{Main Results}\label{sec:mainresults}

In this section we state our main results on bifurcation of solitons
from the band edge, $E=E_*$, into the semi-infinite gap. \\

\nit{\bf Hypotheses:}\\
(H1) Potential:\ $V(\bx)$ is smooth and periodic with 
$\cell=[-\pi,\pi]^d$ . 
\\
(H2)  Dimension / Nonlinearity
\footnote{The assumption on the nonlinear term can be made less restrictive. However, since to some of our results concerning the higher order character of $\mu\mapsto\cP[u(\cdot,\mu)]$ depends on the construction of a multiple scale expansion to a sufficiently high order, we require a certain degree of smoothness of the nonlinear term in a neighborhood of zero. Note also that the methods and our results extend easily to more general nonlinearities, {\it e.g.}  ${\cal K}[|u|^2]u$ (local or nonlocal).}\ :
\begin{equation}
d=1,2: \sigma\in\N,\ \ \ \ \ d=3: \sigma =1
\nn\end{equation}

\begin{theorem}
  \label{theo:2scale}
  Let $\bx_0$ denote any point of symmetry of $V(\bx)$.
  \begin{enumerate}
  \item For all $\mu$ less than and sufficiently near $E_*$, there is a family of  
      nonlinear bound states of NLS/GP (``solitons''), $u(\cdot,\mu)\in H^s(\R^d),\ s>d/2$, which is centered at $\bx_0$.
    \item These solutions 
    bifurcate from the zero solution at band edge frequency 
    $\mu=E_*$ into the semi-infinite gap. Specifically, 
    this family is given by the two-scale expansion for small $\eps$,
    \begin{eqnarray}
      \label{eq:mu-eps}
      \mu_\eps &=& E_*-\eps^2, \\*[2mm]
      \label{eq:u-eps}
      u_\eps(\bx,\mu_\eps) &=& \eps^{\frac{1}{\sigma}}
      \left[ w(\bx)F\left( \eps(\bx-\bx_0) \right)\ +\ 
      \eps U_1\left(\bx,\eps(\bx-\bx_0)\right)\ 
        \right.
        \\ && \left.\ \ \ \   +\ \eps^2  U_2\left(\bx,\eps(\bx-\bx_0)\right)  \ +\ \eta(\bx;\eps)  \right]~,
      \nn
    \end{eqnarray}
    where $\eta(\bx;\eps)$ satisfies the estimate 
    for any $s>d/2$
    \begin{equation}
    \|\ \eta(\cdot;\eps)\ \|_{H^s}\ \le\ C_s\ \eps^3
    \label{eta-est}\end{equation}
  The terms in the expansion are given as follows:\\
    $w(\bx)$ is the band edge Bloch state [see Eq.~(\ref{eq:wdef})] 
    and $F(\by)$ is the ground state solution of the NLS equation in an effective medium:
    \begin{align}
  &    -\sum_{i,j=1}^d \D_{y_i} \aij \D_{y_j} F(\by)\ - \G\ F^{2\sigma+1}(\by)\ =\ - F(\by)~.
      \label{eq:Feqn}\\
 &     F>0,\ \ \ F\in H^1(\R^d)\nn
    \end{align}
 The matrix $\aij$ 
    is the \emph{inverse effective mass tensor} \cite{Kittel}, 
  expressible  
  in terms of the band dispersion function, $E_1(\bk)$, as
    \begin{equation}
      A^{ij} \ \equiv \ \delta_{ij} - \frac{4\ \langle \D_{x_j}w, L_*^{-1}\D_{x_i}w 
          \rangle}{\langle w,w\rangle}~\ =\  \frac{1}{2}\frac{\D^2E_1}{\D k_i\D k_j}(\bk=0);
      \label{eq:Aij} 
    \end{equation} 
    see Appendix \ref{ap:E''-nd}.
   The \emph{effective nonlinear coupling constant}
is given by
    \begin{equation}
      \G =  \frac{\int_{\cell} w^{2\sigma+2}(\bx)\ d\bx}{\int_{\cell} w^2(\bx)\ d\bx}~.
      \label{eq:geff}
    \end{equation}
    $\aij$ is a symmetric, positive definite constant matrix  
    and its determinant, \emph{the product of inverse effective masses}, denoted by 
    \begin{equation}
      \frac{1}{m_*} = \det (\ \aij \ )\ \le\ 1,
      \label{eq:m*}
    \end{equation} 
    with $m_*=1$ only if $V(\bx)$ is identically constant; see 
    \cite{Kirsch-Simon:87} and Appendix~\ref{ap:D(0)}.
  \item 
    $F(\by)$ is a \emph{rescaled ground state} of the NLS equation (\ref{R-eqn}) as
    \begin{equation}
      F(\by)\ =\ \left(\frac{1}{\G}\right)^{\frac{1}{2\sigma}}\ 
      R(\cDD^{-\frac{1}{2}}S\by,-1)
      \label{eq:FR}
    \end{equation}
    where $S$ is an orthogonal matrix that diagonalizes the effective
    mass tensor, \ie 
    \begin{equation}
      S_{ik}\ A^{kl}\ S_{lj} \ = \cDD_{ij} \ \equiv \  {\rm diag}(\lambda_1,\dots,\lambda_d)~,
      \label{eq:ScDD-def}
    \end{equation}
    where $\lambda_i$ denote the eigenvalues of $\aij\ $.

  \item  
    Combining (\ref{eq:mu-eps}), (\ref{eq:u-eps}) and~(\ref{eq:FR})
    gives, for $E_*-\mu>0$ and small:
    \begin{align}
      u(\bx,\mu)\ &=
      \left( \frac{\Delta\mu}{\G}\right)^{\frac{1}{2\sigma}} 
      \left[\  R\left(\cDD^{-\frac{1}{2}}S\ (\Delta\mu)^{\frac{1}{2}}\ (\bx-\bx_0),-1\right)\ 
        w(\bx)\ +\ \cO(\Delta\mu)^{\frac{1}{2}}\ \right]\nn\\
        \Delta\mu\ &=\ E_*-\mu\label{eq:uDmu}
    \end{align}

  \item The  $\cO(\eps)=\cO\left(\sqrt{\Delta\mu}\right)$ 
    and $\cO(\eps^2)=\cO\left(\Delta\mu\right)$ corrections 
    are given,\\ (using $\by\ =\ \eps\left(\bx-\bx_0\right)$ and summation over repeated indices) by
    \begin{align}
      \cO(\eps):\ \ \ \ \ \  U_1(\bx,\by) &= 2L_*^{-1}\left[ \D_{x_i} w(\bx)\right]\ \D_{y_i}F(\by)~,
      \label{eq:U1}\\
      &\nn\\
      \cO(\eps^2):\ \ \ \ \ \ U_2(\bx;\by) &=    U_{2p}(\bx,\by)+w(\bx)F_{2h}(\by)\nn\\
      U_{2p}(\bx,\by)\ &=\ L_*^{-1}\left[\ \left(\ \delta_{ij} +
          4\D_{x_j}L_*^{-1}\D_{x_i} -\ \aij \right)
        w(\bx)\ \right] \D_{y_i}\D_{y_j}F(\by) \nn \\   
      &\ \ \     \ + \ L_*^{-1}\
      \left[ w^{2\sigma+1}(\bx) -\G\ w(\bx) \right]\ F^{2\sigma+1}(\by),\nn\\
      L_+^A\ F_{2h}(\by)\ &=\ S(\by)~,
    \end{align}
    where  $S(\by)$ is given by
    \begin{align}
      S(\by)\ &=\ \langle w,w\rangle^{-1}
      \left[ \langle w,(\Delta_\by-1)U_{2p}(\cdot,\by)+(2\sigma+1)U_0^{2\sigma}U_{2p}(\cdot,\by)\ \rangle\ \right.\nn\\
      &\left. \ \ \ +\ \ \sigma(2\sigma+1)\ \langle w, U_0^{2\sigma-1}U_1^2(\cdot,\by)\rangle\ +\ 2\langle w,\nabla_\bx\cdot\nabla_\by \tilde{U_3}(\cdot,\by)\rangle\right]\label{eq:Sofydef-2}
    \end{align}
    
  \end{enumerate}
\end{theorem}
Theorem \ref{theo:2scale} is proved in Section~\S\ref{sec:edgepower}.\bigskip

\nit Using expansion~(\ref{eq:mu-eps}) we can derive the
asymptotic behavior for $\cP(\mu)=\cP[u(\cdot,\mu)]$ as $\mu\to E_*$.
\begin{theorem} 
  \label{theo:edgepower}
  \nit Let $\bx_0$ denote a point of symmetry of $V$ and $u(\cdot,\mu)$ a soliton given in Theorem \ref{theo:2scale}.
  \begin{enumerate}
  \item
    For  $\mu$ near the band edge $\cP[u(\cdot,\mu)]$ is given by
    \begin{align}
      & \cP[u(\cdot,\mu)]\nn\\
      &\ =\ \left|\ \mu- E_*\ \right|^{\frac{1}{\sigma}-\frac{d}{2}}\nn\\ 
      &\ \ \ \ \times\ \ \      \left[\ \zetap\  \cP[R(\cdot\ ,-1)] +\ \zetas (\mu-E_*)\ +\ \cO\left(\ (\mu-E_*)^2\ \right)\ \right]
      \label{eq:power_edge}
    \end{align}
    where $\cP[R(\cdot\ ,-1)]$, the optical power of the homogeneous NLS ground state, depends on $\sigma$ and $d$:
  \begin{equation}
  -\Delta R-R^{2\sigma+1}=-R,\ \ \ R>0,\ \ \ \ R\in H^1
  \nn\end{equation}
 and 
    \begin{align}
      \label{eq:zeta*}
      & 0\ <\ \zetap \equiv\ \left(\frac{1}{m_*}\right)^{1\over2}\ 
       \left(\frac{\left( \Av w^2 \right)^{\sigma+1}}
        {\Av w^{2\sigma+2}}\right)^{\frac{1}{\sigma}} 
      \ \le\ 1,\\
      &{\rm where}\ \ \dashint_\cell g  =\ 
      \frac{1}{{\rm vol}(\cell)}\ \int\ g(\bx)\ d\bx, 
      \ \ {\rm and \ the\ slope\ is\ given\ by}\nn\\
      &  \zetas\equiv\ 4\  \sum_{j=1}^d\ \Av \left| L_*^{-1}\left[\D_{x_j}w(\bx)\right]\right|^2\ d\bx\cdot \int_{\R^d}\ \left|\ \D_{y_j}F(\by)\ \right|^2\ d\by\nn\\
      &-\ \dashint_\cell w^2(\bx)\ d\bx\ \int \left(\frac{1}{\sigma}F(\by)+\by\cdot\nabla_\by F(\by)\right)\ S(\by)\ d\by,
      \label{eq:zeta1*}   \end{align}
    where $S(\by)$ is given by Eq.~(\ref{eq:Sofydef-2}).
    Note that to order $\cO(|E_*-\mu|^1)$ the expansion is independent of $\bx_0$, the soliton centering.

 \item {\bf Positive slope for small potentials:}\  Let $V(x)=\delta V_1(x)$, where $|\delta|$ is sufficiently
 small and $V_1(x)$ is a smooth periodic function on $\R$
 with a zero cell average.
Then, in the critical case $\sigma=2$
\begin{eqnarray}
  \label{eq:zeta_delta_d1}
  \zetap &\sim& 1 - 8 \delta^2\,  \Av  \left[  (-\D_{xx})^{-1} V_1 \right]^2 \,dx~, \\
  \label{eq:zeta1_delta_d1} 
  \zetas &\sim& 34\sqrt{3}\ \pi  \delta^2 \Av \left[  (-\D_{xxx})^{-1} V_1 \right]^2  \,dx ~.
\end{eqnarray}
Here, $(-\D_{xx})^{-1}$ and $(-\D_{xxx})^{-1}$ are respectively the second and
third-order integration operators in  $\cell$ acting on the space of
zero average functions to itself.  
Hence, $\zetas>0$ for small potentials.

\item {\bf Positive slope conjecture:}\  $\zetas[V]>0$ if $V$ is  non-constant.

\end{enumerate}
\end{theorem}

Theorem \ref{theo:edgepower} is proved in Section~\S\ref{sec:edgepower}, 
except for part 2, concerning small potentials, which is proved in Appendix 
\ref{sec:smallpotentials}.
\begin{remark}
  Concerning equality in Eq.~\eqref{eq:zeta*}. When $V(\bx)$ is constant
 then so is $w(\bx)$. In that case $E_*=0$, $E_1(\bk)=\bk^2$, 
  and $m_*^{-1}= \det\left\{ 2^{-1}\ D^2_{k_i,k_j}E_1(0)\right\}=1$.
  Therefore $\zetap=1$.
\end{remark}

\begin{remark} 
  That $\zetap\le1$ can be seen by considering
  each factor in the definition \eqref{eq:zeta*}
 separately. First, by H\"older's inequality 
  the quotient in the second factor of (\ref{eq:zeta*}) is bounded one
  with equality holding \iff $w\equiv constant$. 
  Furthermore, $w$ is identically constant if and only if $V\equiv constant$.
  Concerning the first factor in (\ref{eq:zeta*}),  
  by Theorem~\ref{theo:2scale},  
  $0<m_*^{-1}\le 1$ with equality holding if $V\equiv constant$.
  Therefore, $ 0<\zetap\le 1$
  with $\zetap=1$ if and only if  $V\equiv $  constant.
\end{remark}


\nit In the critical case, an immediate consequence of Theorem 
\ref{theo:edgepower}
is the following result for critical nonlinearity ($\sigma=2/d$): 
\begin{corollary}\label{cor:edgepower-crit} Consider the critical case $\sigma=2/d$; by hypotheses (H1)-(H2) this implies either $(d,\sigma)=(1,2)$
 or $(d,\sigma)=(2,1)$.
\begin{enumerate}
   \item 
    As $ \mu\to E_*$ we have
     \begin{equation}
       \label{eq:power_edge-crit}
       \cP[u(\cdot,\mu)]\ =\
       \zetap\ \cPc\ +\ \zetas\ (\mu-E_*)\ +\ \cO\left(\ (\mu-E_*)^2\ \right)\ .
     \end{equation}
     Here, $\cP_{cr}=\cP[R(\cdot,-1)]$.\ \ Since $\zetap < 1$ for any non-constant periodic potential, it follows that
     the limiting power at the band edge is strictly smaller than $\cPc$,
     \begin{equation}
    \cP_{edge} \equiv   \lim_{\mu\to E_*}  \cP[u(\cdot,\mu)]\ =\zetap\ \cPc<\ \cPc\ \ .
       \label{eq:edge-vs-Townes}
     \end{equation}
             \end{enumerate}
  \end{corollary}
  
   \nit  Theorem \ref{theo:edgepower} 
is proved in Section~\S\ref{sec:edgepower}. 
  The band-edge limiting behavior \eqref{eq:edge-vs-Townes} is illustrated in Figure \ref{fig:power_mu_1d_A1_v2}; see also  Figure \ref{fig:power_mu_1d}. 
 
 Concerning the NLS / GP dynamics near solitons, we have the following:
 
\begin{theorem}\label{theo:crit-dynamics}
  Consider the critical case $\sigma=2/d$; by hypotheses (H1)-(H2) this implies either $(d,\sigma)=(1,2)$
 or $(d,\sigma)=(2,1)$. 
 Then, if the positive slope conjecture of Theorem \ref{theo:edgepower} 
 holds, then 
\begin{enumerate}
\item
 \begin{equation}
         \label{eq:edgeslope}
         \left.\ \slope\ \right|_{\mu=E_*} \ > \ 0
       \end{equation}
 and it follows from Theorem \ref{theo:stability} that for $\mu$ such that $E_*-\mu>0$ and sufficiently small, $u(\cdot,\mu)$ is unstable.
 \item In particular, for small periodic potentials, by Theorem \ref{theo:edgepower},  for $\mu$ such that $E_*-\mu>0$ and sufficiently small, $u(\cdot,\mu)$ is unstable.
   \end{enumerate}
\end{theorem}

To complement this information about stability / instability of solitons we remark on $\cP_{cr}$ and $\cP_{edge}$ as they relate
to well-posedness and blow-up / collapse. 
\begin{theorem}\label{thm:globexistV}
 Denote by $R(\bx)$,  the  ground state (``Townes soliton'') for $V(\bx)\equiv 0$. 
 If 
\begin{equation}
    \label{eq:P<Pc}
    \cP[\psi_0] = \int |\psi_0(\bx)|^2\ d\bx <  \int R^2(\bx)\,d\bx\ \equiv\ \cP_{cr}
  \end{equation}
then solutions of NLS/GP \eqref{eq:NLS-V} exist globally in time; no singularity formation / no collapse.  
\end{theorem}
 \begin{remark}\label{rmk:sol-exc-conj}
Recall that in the spatially homogeneous case, $V\equiv0$, if in addition to \eqref{eq:P<Pc} we impose the stronger assumptions: $\psi_0\in H^1$ and $|\bx|\psi_0\in L^2$,  then  $\psi(\bx,t)$ tends to zero as $t\to\infty$ for a range of $p>2$  \cite{Weinstein:89}; see also \cite{KVZ:07} for scattering results in $H^1$. 

$\cP_{cr}$ is thus called a \underline{\it soliton excitation threshold}.  Excitation thresholds also play a role in systems without critical scaling symmetry. See, for example, 
\cite{Weinstein:99,FKM:97} and \cite{CP:03,MRS:preprint,BCT:09}. 

For $V$ non-zero, the picture which emerges from the above 
theorems and numerics (see, for example, figure \ref{fig:power_mu_1d_A1_v2}) is quite different. 
 The minimal mass (minimal power), band edge power and $V\equiv0$ critical mass are related by: 
  \begin{equation}
  \cP_{\rm min}<\cPedge<\cPc\ .
  \nn\end{equation}
  Here, 
  \begin{equation}
  \cP_{min}=\cP[u(\cdot,\mu_{min})]\equiv  \min_{\mu\le E_*}\ \cP[u(\cdot,\mu)]\ ,
  \label{eq:limcP}\end{equation}
  where in \eqref{eq:limcP}: $\mu\mapsto\cP[u(\cdot,\mu)]$ is computed along the family of solitons centered at a local {\it minimum}; see the solid curve in figure \ref{fig:power_mu_1d_A1_v2}.  Along this soliton curve, computations indicate that  $u(\cdot,\mu)>0$ and $n_-(L_+)=1$. By Theorem \ref{theo:stability}
 (applied for $V$ periodic) 
there is an open set of initial data in the phase space $H^1$:\ 
 \begin{equation}
 \{\psi_0\in H^1 :\ \cP_{\rm min}< \cP[\psi_0]  < \cPedge < \cPc \}
 \label{openset}
 \end{equation} 
within which
  there are co-existing  unstable / ``wide'' and stable / ''narrow'' solutions.
  
 \nit There is also an open set in $H^1$
   \begin{equation}
 \{\psi_0\in H^1 :\ \cP_{edge} < \cP[\psi_0] < \cPc \} ,
 \label{openset1}
 \end{equation}
 where the only solitons are  stable  and ``narrow''. The terms wide  and narrow refer, respectively, to solitons with frequencies in an interval near (to the right of $\mu_{min}$) or far (to the left of $\mu_{min}$) the band edge, $E_*$ \cite{SFIW-08,Sivan-NL-08} \\

\nit Finally, we state a \\
{\bf Soliton excitation threshold conjecture} (see also \cite{Weinstein:89, Weinstein:99,KVZ:07} $\cP_{min}$ is an {\it excitation threshold}:
\nit  If $\cP[\psi_0]\ < \ \cP_{\rm min}$,  then $\psi(\bx,t)$ tends
to zero as $t\to\infty$ ( $L^p$, for some range of $p>2$ with the free Schr\"odinger decay-rate if $\psi_0\in H^1$ and sufficiently localized in space) or in $L^2_{loc}$, and scattering holds for $\psi_0\in H^1$.
 \end{remark}

\begin{proof} Proof of Theorem \ref{theo:crit-dynamics}.
  Part 1 follows from part 2 of  Theorem \ref{theo:stability}, where we review results on the stability / instability of solitary waves.
\end{proof}

\begin{proof} Proof of Theorem \ref{thm:globexistV}.
  This follows from an application of the sharp Gagliardo-Nirenberg inequality; see \cite{Weinstein:83,Weinstein:89}. Specifically, for any function $f\in H^1(\R^d)$ we have
  \begin{align}
    &\left(\ 1\ -\ \frac{\|f\|_{L^2} }{\|R\|_{L^2} } \right)^{\frac{4}{d}}\ 
    \int\ |\nabla f|^2\ \le\ \int  \left(\ |\nabla f|^2\ -\
      \frac{1}{1+\frac{2}{d}}\ |f|^{\frac{4}{d}+2}\ \right)
    \ \equiv\ \cH_0[f],\nn\\
    &\label{sharp-sng}
  \end{align}
  where $\cH_0$ denotes the conserved 
  NLS/GP Hamiltonian for $V\equiv0$.
  Estimate \eqref{sharp-sng} was used in \cite{Weinstein:83} to establish, for $V\equiv0$, that if $\psi_0\in H^1$ and $\|\psi_0\|_{L^2}<\|R\|_{L^2}$, then NLS has a global in time $H^1$ bounded solution.
  It was further used in \cite{Weinstein:89} to show that if, in
  addition we assume that $|\bx|\psi_0\in L^2$, then the solution 
  decays to zero in $L^p$, for range of $p>2$ (and therefore in $L^2_{\rm loc}$).
  \end{proof}

\begin{proof}
  To prove Theorem  \ref{thm:globexistV}, note from \eqref{sharp-sng} that
  \begin{equation}
    \left(\ 1\ -\ \left( \frac{\|f\|_{L^2} }{\|R\|_{L^2} } \right)^{\frac{4}{d}}\ \right)\ 
    \int |\nabla f|^2\ \le\ \cH[f]\ -\ \int V\ |f|^2
    \label{sharp-sng-V}\end{equation}
  Applying this inequality to a solution, $\psi(\bx,t)$, of NLS/GP yields
  \begin{equation}
    \left(\ 1\ -\ \left( \frac{\|\psi_0\|_{L^2}}{\|R\|_{L^2} } 
      \right)^\frac{4}{d}\ \right)\ \int |\nabla \psi(\bx,t)|^2\ \le\ \cH[\psi_0] + \|V\|_{L^\infty}\ \int |\psi_0|^2
    \label{psi-bound}\end{equation}
  For initial data, $\psi_0$, in  small $H^1$ neighborhood of a  soliton with frequency near the band edge, we have $\|\psi_0\|_{L^2}<\|R\|_{L^2}$. 
  Estimate \eqref{psi-bound} implies a uniform bound on $\|\psi(\cdot,t)\|_{H^1}$ and therefore global existence (no singularity formation / no collapse).
\end{proof}

 
\subsection{Finite gaps -- results for focusing and defocusing nonlinearities}\label{section:finite-gaps}

In this section we remark on extensions of our results to solitons 
with frequencies in finite gaps (gap solitons). For this, more general, discussion it is convenient to write NLS/GP and its nonlinear bound state equation in the form
\begin{align} 
&  i \D_t\psi\ = -\Delta \psi + V(\bx) \psi  + g|\psi|^{2\sigma} \psi
    \label{eq:NLS-V-g}\\
&  \left(\ -\Delta\ + V\ \right) u \ +g\ u^{2\sigma+1} =  \mu\ u, 
  \label{eq:u-nd-g}
\end{align}
where we have introduced a parameter $g$ to encode
the (i) focusing / attractive ($g=-1$) and the defocusing / repulsive ($g=+1$) cases. 
\\

\nit {\bf Focusing nonlinearity, $g=-1$:}\  Our results of the
previous section applied to solitons with frequencies in the
semi-infinite gap, $\mu<E_*$. The results on bifurcations of solutions
from  the spectral band edge 
can be extended  to the case where  $E_*$ is replaced by $E_{edge}$,
any band edge frequency. Here, we consider the case where the following two conditions hold
 \begin{enumerate}
 \item The space of $\cell$ - periodic solutions $(-\Delta+V)w(\bx)=E_{edge}w(\bx)$  is one-dimensional, spanned by a function $w_{edge}(\bx)$, $E_{edge}$ is attained by the band dispersion function at $\bk=0$.
 \footnote{In dimensions $d\ge2$ band edges may be attained at $0\ne\bk\in\cell^*$; see \cite{DU:09}. In this case, the corresponding solutions are complex-valued and an extension of the present methods we use along the lines of \cite{DU:09} is necessary.}
 \item  The inverse effective mass tensor, $A^{ij}$, is symmetric and positive definite.
 \end{enumerate}
 In this case, we have solitons centered about any point of symmetry of $V(\bx)$, which in analogy to those described in Theorem \ref{theo:2scale}, bifurcate from the left band-edge toward lower frequencies, into the spectral gap
 \begin{equation}
\mu\mapsto u(\bx,\mu\ ) \approx\
 \left(E_{edge}-\mu\right)^{1\over2\sigma}\  w(\bx)\ F\left(\ \sqrt{E_{edge}-\mu}\ (\bx-\bx_0)\ \right),\nn
 \end{equation}
$E_{edge}-\mu>0$  and sufficiently small. Here $F$ satisfies the effective medium nonlinear Schr\"odinger equation \eqref{eq:Feqn}, whose inverse effective mass tensor, $A^{ij}$ is given by equation \eqref{eq:Aij}, with $w$ replaced by $w_{edge}$. Alternatively, this is $(D^2E_{n}(\bk_0))_{ij}$, the Hessian matrix of a Bloch dispersion function, $E_{n}$, where $E_{n}(\bk_0)=E_{edge},\ \ \bk_0\in\cell^*$. 
\bigskip

\nit{\bf Defocusing nonlinearity, $g=+1$:} Here, we consider the case
where the following two conditions hold
\begin{enumerate}
 \item The space of $\cell$- periodic solutions $(-\Delta+V)w(\bx)=E_{edge}w(\bx)$  is one-dimensional, spanned by a function $w_{edge}(\bx)$, $E_{edge}$ is attained by the band dispersion function at $\bk=0$.
 \item  The inverse effective mass tensor, $A^{ij}=-B^{ij}$, is symmetric and {\it negative} definite.
 \end{enumerate}
 In this case, we have solitons centered about any point of symmetry of  $V(\bx)$,  bifurcating  from the right band-edge toward higher frequencies, into the spectral gap.\\
 
 Indeed, if we seek, along the lines of our previous analysis,  soliton-like states with frequency:
 \begin{equation}
 \mu = E_{edge} - \tau\eps^2, 
 \nn\end{equation}
 our analysis near a band edge with negative definite effective mass tensor, $-B^{ij}$,  yields an effective medium soliton equation:
 \begin{equation}
 -\sum_{i,j=1}^d\D_{y_i} B^{ij} \D_{y_j} F - \G  F^{2\sigma+1} = \tau F
\nn\end{equation}
Thus, we can construct localized states for $\tau<0$ 
and  $\mu = E_{edge} + |\tau|\eps^2>E_{edge}$.
\\

Finally, we remark that all hypotheses concerning multiplicity of spectrum and curvature of band dispersion functions are verifiable in one space ($d=1$)  dimension.  Thus we have 

\begin{theorem}\label{theo:1d2scale}
Let $V(x)$ denote a smooth, periodic and even potential.
Consider any finite width, non-empty, spectral gap, $-\infty<a<b<\infty$, of $-\D_x^2+V(x)$.  The band dispersion curvature at $E=a$ is strictly negative and at $E=b$ is strictly positive; see Appendix \ref{ap:FB}.
\begin{enumerate}
\item For focusing nonlinearity,  $g=-1$,  centered about any point of symmetry of $V$, there exists a family of solitons of NLS-GP \eqref{eq:NLS-V-g}, which  bifurcates from the zero solution with frequencies in  the gap {\it less} than $E=b$.
\item For defocusing nonlinearity, $g=+1$, centered about any point of symmetry of $V$, there exists a family of solitons of NLS-GP \eqref{eq:NLS-V-g}, which bifurcates from the zero state with frequencies bifurcating into the gap {\it greater} than $E=a$.
\end{enumerate}
   These bifurcating branches have expansions and properties  analogous to those described in Theorem \ref{theo:2scale} and  Theorem \ref{theo:edgepower}.
\end{theorem}
\nit {\bf N.B.} The results of this subsection indicate extensions to bifurcations into finite  width gaps. In particular, for critical nonlinearities, we are able to analytically characterize the band-edge limit of the squared $L^2$ norm, $\cP$. Note however that the factor, $\zetap$, arising in finite gaps is associated with an excited  Bloch state, {\it i.e.} a state $w_{edge}(\bx)$, which is not a positive ground state of the periodic boundary value problem. Since the estimate $\zetap\le 1$, hinged on the result
\cite{Kirsch-Simon:87}: $(m_*)^{-1}\le1$, which makes use of the ground state property (in particular positivity),  we do not have an estimate on the size of $\zetap$ in finite gap cases.
%


\section{Homogenization / multi-scale expansion }
\label{sec:homog-expansion}
In this section we  derive a formal multiple scale expansion of
solitons bifurcating from the band edge. 
In Section~\S\ref{sec:error-est} we prove an error estimate, 
thus completing the proof of Theorem \ref{theo:2scale} .

{\it Without loss of generality we choose coordinates with $\bx_0=0$.}
 We seek a solution of the bound state equation~(\ref{eq:u-nd}), 
which bifurcates from the zero state at the band edge $\mu=E_*$, 
depending on a ``fast'' spatial scale $ \bx$ and a slow spatial scale 
\begin{equation}
  \by\ =\ \epsilon\left(\bx-\bx_0\right)\ =\ \eps\bx,\  \ \ \ \ \  \eps\ll 1
  \label{eq:slowscale}
\end{equation}
of the form
\begin{subequations}
  \label{eq:u-MS}
  \begin{eqnarray}
    \mu_\eps        &=&E_*+\eps \mu_1 + \eps^2 \moot + \dots \\
    u_\epsilon(\bx)  &=& \eps^{\frac{1}{\sigma}}\Ue(\bx,\by)\\
    \Ue(\bx,\by   ) &=&U_0(\bx,\by) + \eps U_1(\bx,\by) + \eps^2U_2(\bx,\by) + \dots~.
  \end{eqnarray}
\end{subequations}
We also impose periodicity in $\bx$, \ie
\begin{equation}
\Ue(\bx+{\bf q}_j,y)=\Ue(\bx,\by),\ \ j=1,\dots,d \label{eq:Ue-per}
\end{equation}
Rewriting equation~(\ref{eq:u-nd}) by treating $\bx$ and $\by$ 
as {\it independent} variables gives
\begin{eqnarray}
    -\left(\nabla_\bx + \epsilon\nabla_\by \right)^2\Ue
    + V(\bx)\Ue - \eps^2\ \Ue^{2\sigma+1}\ =\ \mu_\eps \Ue
  \nn
\end{eqnarray}
Using the expansion (\ref{eq:u-MS})  and the operator $L_*$  [see Eq.~(\ref{eq:wdef})],
we obtain the following hierarchy of equations to $\cO(\eps^4)$
\begin{eqnarray*}
  \cO(\eps^0): & L_*U_0  =& 0~, \\
  \cO(\eps^1): & L_*U_1  =&
   \left(2\nabla_\bx\cdot\nabla_\by + \mu_1\right)U_0\\
  \cO(\eps^2): & L_*U_2  =& \left(2\nabla_\bx\cdot\nabla_\by+ \mu_1\right)U_1 + (\Delta_\by + \moot) U_0+U_0^{2\sigma+1}\\
  \cO(\eps^3): & L_*U_3  =&\left(2\nabla_\bx\cdot\nabla_\by+ \mu_1\right)U_2 + (\Delta_\by+\moot)U_1+(2\sigma+1)U_0^{2\sigma}U_1
  +\mu_3U_0,\\
  \cO(\eps^4): & L_*U_4  =& \left(2\nabla_\bx\cdot\nabla_\by+ \mu_1\right)U_3 + (\Delta_\by+\moot)U_2\\
   && +(2\sigma+1)U_0^{2\sigma}U_2+(2\sigma+1)\sigma U_0^{2\sigma-1}U_1^2+\mu_4U_0\label{eps-k-rhs}
\end{eqnarray*}
where for each $k\ge5$ we have:
\begin{align}
  \cO(\eps^k): \ \ 
  L_*U_k \ &=\ \mu_k U_0\nn\\
  &+ \  
  \left(2\nabla_\bx\cdot\nabla_\by+\mu_1\right) U_{k-1}(\bx,\by)
  \nn\\
  &  +  (\Delta_\by+\moot)U_{k-2}+ \cF_k[U_j(\bx,\by),\mu_j:1\le j\le k-2] 
\label{eq:eqn-form}
\end{align}
Note that $L_*$ is self-adjoint with a one-dimensional null-space
spanned by $w$.
In addition, $\mu_k$ is determined by a solvability condition of the form:
\begin{equation}
  \mu_k\langle w(\cdot),U_0(\cdot,\by)\rangle + \langle w(
  \cdot) , \tilde{\cF}_k(\cdot,\by)\rangle\ =\ 0,
  \label{eq:solvability}
\end{equation}
obtained by imposing orthogonality of  $w$ to the 
right hand side of  (\ref{eq:eqn-form}). Here, $\tilde{\cF}_k$ denotes expression the sum of the last two lines on the right hand side of \eqref{eps-k-rhs}.
Condition (\ref{eq:solvability}) ensures the existence of a solution to (\ref{eq:eqn-form}) which is periodic
in $\bx$. 

We now implement this procedure at successive orders in $\eps$.
 In particular, we construct the terms $U_j(\bx,\by),\ \ 0\le j\le 4$,
 as these are required in the proof of Theorem~\ref{theo:2scale}.


 \subsection{Solution at each $\cO(\eps^k),\ k=0,1,2,3,4$}
 \nit{\bf $\cO(\eps^0)$ terms:}\ \ 
The $\cO(\eps^0)$ equation is solved by the choice
\begin{equation}
  U_0(\bx,\by) = w(\bx)F(\by),
  \label{eq:U0}
\end{equation}
where $w$ is the periodic Bloch state
associated with the band edge, as defined in Eq.~(\ref{eq:wdef}).\\

  \nit{\bf $\cO(\eps^1)$ terms:}\ \ 
 The $\cO(\eps)$ equation for $U_1$, by (\ref{eq:U0}),  becomes
\begin{equation}
  L_*U_1=2 \nabla_\bx w\cdot\nabla_\by F+\mu_1w\ F
\label{U1mu1-eqn}
\end{equation}
Orthogonality of the right hand side of (\ref{U1mu1-eqn} to $w$ 
 implies $
\mu_1=0$,  from which we obtain (\ref{eq:U1}):
\begin{equation} 
U_1\ =\ 2 L_*^{-1}[ \nabla_\bx w]\cdot \nabla_\by F(\by)\ .
\label{U1-sol}\end{equation}
\bigskip

\begin{remark} 
To be completely systematic, we should add to the right hand side of \eqref{U1-sol} a term of the form $F_{1h}(\by)w(\bx)$, which is in the null space of $L_*$, with $F_{1h}(\by)$ to be determined. 
$F_{1h}(\by)$ is determined via the solvability condition for $U_3$. Symmetry considerations lead to $F_{1h}(\by)\equiv 0$ (see the discussion of $U_3$). We omit inclusion of this term to simplify the presentation. Note however that this degree of freedom
is required at higher order. In particular, see the expression for $U_2(\bx,\by)$ and the role of $F_{2h}(\by)$ in the solving for $U_4(\bx,\by)$.
\end{remark}
\bigskip

  \nit{\bf $\cO(\eps^2)$ terms:}\ \ The $\cO(\eps^2)$ equation for $U_2$, by ~(\ref{eq:U0}) and (\ref{eq:U1}) becomes 
\begin{align}
        L_* U_2 &= w(\bx)\left(\Delta_\by+\moot\right)F(\by)\ +\ 4\ \D_{x_j} \D_{y_j}L_*^{-1}[\D_{x_i}w](\bx)\
       \D_{y_i}F(\by)\nn\\
         &+ w^{2\sigma+1}(\bx)F^{2\sigma+1}(\by)\nn\\
         &= w(\bx)\ \Delta_\by\ F(\by)\ +\ 4\ \nabla_\bx\cdot\nabla_\by\left[\ L_*^{-1}[\nabla_\bx w](\bx)\cdot\nabla_\by F\ \right]\nn\\
         & +\ \mu_2\ w(\bx)F(\by)\ + w^{2\sigma+1}(\bx)F^{2\sigma+1}(\by)
      \label{eq:U2}
    \end{align}
 An equation for $F(\by)$ is obtained by imposing orthogonality of the right hand side of (\ref{eq:U2}) to  $w(\bx)$. It is convenient to formulate the following
 \begin{proposition}\label{prop:homog-op}
 Denote by ${\cal L}_*$ the operator
 \begin{align}
 G(\by)&\mapsto {\cal L}_*[G](\bx,\by)\nn\\
  &=\ w(\bx)\ \Delta_\by\ G(\by) \ +\ 4\ \nabla_\bx\cdot\nabla_\by\left[\ L_*^{-1}[\nabla_\bx w](\bx)\cdot\nabla_\by G(\by)\ \right]\label{cL*def}
 \end{align}
 Then, 
 \begin{equation}
 \langle\ w(\cdot),{\cal L}_*[G](\cdot,\by)\ \rangle\ =\  \D_{y_i} \aij \D_{y_j}\ G(\by)\ \ \times\ \ \langle w,w\rangle\ \ .
  \label{homog-op}\end{equation}
 \end{proposition}

 Imposing orthogonality of the right hand side of (\ref{eq:U2}) to  $w(\bx)$ and applying Proposition \ref{prop:homog-op} yields equation  (\ref{eq:Feqn}) for $F=F(\by,\mu_2)$:
  \begin{equation}
      -\ \D_{y_i} \aij \D_{y_j} F(\by,\mu_2)\ - \G\ F^{2\sigma+1}(\by,\mu_2)\ =\ \moot F(\by,\mu_2)~.
      \label{eq:Feqn-1}
    \end{equation}
{\it Here, we consider only the positive decaying solution of (\ref{eq:Feqn-1}), which by scaling and uniqueness can be expressed as}
\begin{equation}
F(\by,\mu_2)\ =\ |\mu_2|^\frac{1}{2\sigma}\ F(|\mu_2|^{1\over2} \by;-1)
\label{Fmu2scale}
\end{equation}
We can therefore, scale out $|\mu_2|$ and henceforth assume $\mu_2=-1$.\\

    \nit Thus far, we have shown:\\ \\
\nit {\it To leading order, the slowly varying envelope function $F(\by)$
of the nonlinear bound state of NLS/GP is comprised of 
a nonlinear bound state of the NLS equation~(\ref{eq:Feqn}) 
for a homogeneous medium with effective mass tensor $(\aij)^{-1}$ [Eq.~(\ref{eq:Aij})]
and effective nonlinearity $\G$.}\\ \\ 
In Subsection~\S\ref{subsec:Ftownes} we show that  $F(\by)$  is an appropriate scaling of $R(\by)$, the Townes soliton, the ground state associated with an isotropic homogeneous medium.\\ \\

We express the general solution of (\ref{eq:U2}) in the form
\begin{equation}
U_2(\bx,\by)\ =\ U_{2p}(\bx,\by)+U_{2h}(\bx,\by)=U_{2p}(\bx,\by)\ +\ w(\bx)F_{2h}(\by),
\label{U2expnd}\end{equation}
where $U_{2p}$ denotes a particular solution of (\ref{eq:U2}) and 
 $w(\bx)F_{2h}(\by)$ lies in the kernel of $L_*$ (recall $L_*w=0$), with $F_{2h}(\by)$ to be determined.
 
Using  equation (\ref{eq:Feqn-1})  to simplify the right hand side of (\ref{eq:U2}) gives:
 \begin{align}
 U_{2p}(\bx,\by) &=  \sum_{1\le i, j\le d}\ L_*^{-1} \left[\ \left(\ \delta_{ij} + 4\D_{x_j}L_*^{-1}\D_{x_i} \right)
   w(\bx)\ \D_{y_i}\D_{y_j}F(\by)\ \right.
\nn  \\  &  \left.  \ + \ 
    w^{2\sigma+1}(\bx)F^{2\sigma+1}(\by)\ -\  w(\bx)F(\by)\  \right]
\nn  \\*[2mm]  &=   L_*^{-1}\left[\ \left(\ \delta_{ij} +
      4\D_{x_j}L_*^{-1}\D_{x_i} - \aij \right)
    w(\bx)\ \right] \D_{y_i}\D_{y_j}F(\by) \nn \\   
    &    \ + \ L_*^{-1}\
  \left[ w^{2\sigma+1}(\bx) -\G\ w(\bx) \right]\ F^{2\sigma+1}(\by)\label{U2fac-1}\\
  &{\rm or}\nn\\
  U_{2p}(\bx,\by)\ &\equiv\ \sum_{1\le i,j\le d}L_*^{-1} 
  X^{ij}_{2p,1}(\bx)\ \D_{y_i}\D_{y_j}F(\by)\ +\ L_*^{-1}X_{2p,2}(\bx)\ F^{2\sigma+1}(\by),
  \label{U2fac-2}
\end{align} 
with $X^{ij}_{2p,1}$ and $X_{2p,2}$ given by the corresponding expressions in (\ref{U2fac-1}).
To obtain (\ref{U2fac-1}), we use equation (\ref{eq:Feqn-1}) for $F(\by)$, in terms of 
the effective mass tensor, (\ref{eq:Aij}) and effective coupling,
 (\ref{eq:geff}). This is a consequence of the solvability (orthogonality) condition for equation 
  for (\ref{eq:U2}).
 \\ \\
   \nit{\bf $\cO(\eps^3)$ terms:}\ \ Using \eqref{U2expnd}, we obtain the following equation for $U_3(\bx,\by)$: 
   \begin{align}
   L_*U_3\ &=\ 2\ \nabla_\bx\cdot\nabla_\by \left(\ U_{2p}\ +wF_{2h}\ \right)\nn\\
   & +\ \left(\ \Delta_\by-1+(2\sigma+1)F^{2\sigma}w^{2\sigma}\ \right)\ 2L_*^{-1}\D_{x_i}w\ \D_{y_i}F
   \  +\ \mu_3w\ F.\label{U3temp}
    \end{align}
Solvability of (\ref{U3temp}) requires orthogonality of the right hand side to $w$. Since all terms, except the last, on the right hand side of (\ref{U3temp}) are antisymmetric functions of $\bx$, we have $\mu_3=0$. Thus, after substitution of the explicit expression for $U_{2p}$ we have
\begin{align}
 L_*U_3  &=\ 2\D_{x_l}w(\bx)\ \D_{y_l}F_{2h}(\by)\nn\\
   &\ \  +\ 2\D_{x_l}L_*^{-1}X_{2p,1}^{ij}(\bx)\ \D_{y_l}\D_{y_i}\D_{y_j}F(\by)
   \ +\ 2\D_{x_l}L_*^{-1}X_{2p,2}(\bx)\ \D_{y_l}F^{2\sigma+1}(\by)\nn\\
   &\ \ +\ 2L_*^{-1}\D_{x_i}w(\bx)\ \left(\ \Delta_\by\ -1 \right)\D_{y_i}F(\by)\nn\\
   &\ \   +\ 2w^{2\sigma}(\bx)L_*^{-1}\D_{x_i}w(\bx)\ \D_{y_i}F^{2\sigma+1}(\by),
   \label{eq:U3}
   \end{align}
 with summation over repeated indices implied.
 Thus,
 \begin{align}
 U_3\ &=\ 2L_*^{-1}[\nabla_\bx w]\ \cdot\ \nabla_\by F_{2h}\nn\\
 &+ 2L_*^{-1}\left[\ \nabla L_*^{-1}[X_{2p,1}^{ij}]\cdot\nabla_\by\D_{y_i}\D_{y_j}F\ +\ \nabla L_*^{-1}[X_{2p,2}]\cdot\nabla_\by F^{2\sigma+1}\right.
 \nn\\
&\ \ \ \ \ \  +\left. L_*^{-1}[\nabla_\bx w]\cdot(\Delta_\by-1)\nabla_\by F+2w^{2\sigma}L_*^{-1}[\nabla_\bx w]\cdot \nabla_\by F^{2\sigma+1}\ \right]\nn\\
&\equiv\ 2L_*^{-1}[\nabla_\bx w](\bx)\ \cdot\ \nabla_\by F_{2h}\ +\ \tilde{U_3}
  \label{eq:tU3def}\end{align}
 \\ \\
   \nit{\bf $\cO(\eps^4)$ terms:}\ \ For $U_4$ we have
   \begin{align}
 &  L_*U_4  = \left(2\nabla_x\cdot\nabla_y+ \mu_1\right)U_3 + (\Delta_\by-1)U_2+(2\sigma+1)U_0^{2\sigma}U_2\nn\\
 &\ \ \ \ \ \ \ \ \  +\sigma (2\sigma+1)U_0^{2\sigma-1}U_1^2+\mu_4U_0\nn\\
   &\ \ =\ {\cal L}_*[F_{2h}](\bx,\by)\ +\ 
    (2\sigma+1)w^{2\sigma+1}(\bx)F^{2\sigma}\nn\\
&\ \ \ \ \ \ +\   (\Delta_\by-1)U_{2p}\ +\ (2\sigma+1)U_0^{2\sigma}U_{2p}\ +\ 2\nabla_\bx\cdot\nabla_\by\ \tilde{U_3}\nn\\
&\ \ \ \ \ \ +\ \sigma(2\sigma+1) U_0^{2\sigma-1}U_1^2+\mu_4w(\bx)F(\by)
\label{U4eqn}
\end{align}
The operator $ {\cal L}_*[\cdot](\bx,\by)$, appearing in (\ref{U4eqn})
  is defined in Proposition.
\ref{prop:homog-op}. Imposing orthogonality of the right hand side of (\ref{U4eqn}) and applying Proposition \ref{prop:homog-op} gives the following equation for $F_{2h}$: 
\begin{align}
L_+^A\ F_{2h}(\by)\ &=\ \langle w,w\rangle^{-1}
\left[ \langle w,(\Delta_\by-1)U_{2p}(\cdot,\by)+(2\sigma+1)U_0^{2\sigma}U_{2p}(\cdot,\by)\ \rangle\ \right.\nn\\
&\left. \ \ \ +\ \ \sigma(2\sigma+1)\ \langle w, U_0^{2\sigma-1}U_1^2(\cdot,\by)\rangle\ +\ 2\langle w,\nabla_\bx\cdot\nabla_\by \tilde{U_3}(\cdot,\by)\rangle\right]\ +\ \mu_4 F(\by)\nn\\
& \equiv\ S(\by),
\label{Sofydef}
\end{align}
where $L_+^A$ is  the second order linear Schr\"odinger operator:
\begin{equation}
L_+^A\ \equiv\ -\D_{y_i}A^{ij}\D_{y_j}\ +1\ -\ (2\sigma+1)\G F^{2\sigma+1}(\by).
\label{L+Adef}
\end{equation}
We now show that we can take $\mu_4=0$.
Equation (\ref{Sofydef}) can be solved in $L^2(\R^d)$ for $F_{2h}(\by)$  if and only if $S(\by)$ is $L^2-$ orthogonal to the kernel of $L_+^A$. The kernel of 
  $L_+^A$ has dimension $d$ and is generated by translations, {\it i.e.} ${\rm Kernel}(L_+^A)= {\rm span}\{\D_{y_j}F(\by),\ j=1,\dots,d\ \}$ \cite{Weinstein-85,Kwong-89}. Since $F(\by)$ is even, the kernel of $L_+^A$ consists of functions which are antisymmetric  in one coordinate direction. Moreover, it is easy to see that all terms in $S(\by)$ are symmetric and therefore orthogonal to the kernel of $L_+^A$. Thus, we set 
  $\mu_4=0$.


\subsection{ $F(\by)$ is a scaled Townes soliton}\label{subsec:Ftownes}

Thus far, we have constructed 
the formal expansion ~(\ref{eq:mu-eps}) of $(u_\veps,\mu_\veps)$  through $\cO(\eps^2)$.  The proof of its validity, in particular the error estimate (\ref{eta-est}), is given in Section~\S\ref{sec:error-est}.

We conclude this section by relating the {\it effective medium soliton}\ $F(\by)$, which solves the NLS bound state equation with effective media parameters $A^{ij}$ and $\G$, to 
the unique ground state of the uniform-medium NLS equation,
\begin{equation}
  -\Delta R\ -\ R^{2\sigma+1}\ =\ \mu\ R,\ \ R>0,\ \ R\in H^1(\R^d).
  \label{eq:isotropic-nls}
\end{equation}
Let $A=(\aij)$ and $\Lambda\equiv diag(\lambda_1,\dots,\lambda_d)$ 
denote the diagonal matrix, whose diagonal entries are the eigenvalues of $A$.
Let $S$ denote an orthogonal matrix for which 
\begin{equation}
  SAS^T=\Lambda\equiv diag(\lambda_1,\dots,\lambda_d)~.
  \label{eq:Sdef}
\end{equation}
Then, under the change of coordinates 
$\by\mapsto \bz=\Lambda^{-\frac{1}{2}}S\by$,  
$F_1(\bz)=F(\by,-1)$ solves Eq.~(\ref{eq:isotropic-nls}) with $\mu=-1$.
By uniqueness up to translations,
the solution to the isotropic NLS equation~(\ref{eq:isotropic-nls})
is given by
\begin{equation}
  F(\by)\ =\ F_1(\bz)\ =\  R(\Lambda^{-\frac{1}{2}}S\by,-1)
  \ = \ \left(\frac{1}{\G}\right)^{\frac{1}{2\sigma}}\ 
  R(\Lambda^{-\frac{1}{2}}S\by,-1)
  \label{eq:F1def}
\end{equation}
Note that our expansion gives $\mu_\eps=E_*-\eps^2 +\ \cO(\eps^5)$ (recall $\mu_3=\mu_4=0$). As shown in the proof, we can in fact  take  
\begin{equation}
  \Delta\mu\ \equiv\ \mu_\eps=E_*-\eps^2\ .
  \label{eq:mudef}
\end{equation}
Substitution of Eqs.~(\ref{eq:F1def}) and~(\ref{eq:mudef}) into the 
expansion (\ref{eq:u-MS}), and using  (\ref{eq:U0})  and  $\by=\eps\bx$ yields
the leading order expansion of $u(\bx,\mu)$
displayed in Eq.~(\ref{eq:u-eps}). 


\subsection{$\cP\left[u(\cdot,\mu)\right]$ near the band edge}

\label{sec:edgepower}

To prove Theorem  \ref{theo:edgepower} and Corollary \ref{cor:edgepower-crit} we evaluate 
$\int |u_\eps(\bx)|^2\ d\bx$, where $u_\eps$ is given by 
the two-scale expansion plus error term  (\ref{eq:u-eps}),  of Theorem \ref{theo:2scale}. 
We obtain (again recalling the choice of coordinates, so that $\by=\eps\left(\bx-\bx_0\right)=\eps\bx$)
\begin{align}
 &   \int_{\R^d} |u_\eps(\bx)|^2\ d\bx\ =\
     \eps^{\frac{2}{\sigma}}\Biggm[ 
    \int \underbrace{|U_0(\bx,\eps\bx)|^2}_{\alpha_0} 
    + \underbrace{  2\eps U_0(\bx,\eps\bx)U_1(\bx,\eps\bx)}_{\alpha_1}
 \nn \\
&  + \underbrace{\eps^2\left( |U_1(\bx,\eps\bx)|^2+2 U_0(\bx,\eps\bx)
   U_2(\bx,\eps\bx)\right)}_{\alpha_2}\ d\bx\ +\  \cO(\eps^3)~\Biggm]\nn\\ &= \eps^{2\over\sigma}\left[\ \cI_0^\eps\ +\ \cI_1^\eps\ +\ \cI_2^\eps\ +\ \cO(\eps^3)\ \right]
  \label{eq:int-u2}
\end{align}
Each of the three terms on the right-hand side
will be treated below using the following general
averaging method
\begin{lemma}\label{lem:asymp}\  
  Let $p(\bx)$ be periodic on the lattice $\Gamma$ having the fundamental 
  period-cell $\cell$. 
  Let $\cell^*$ denote the dual fundamental cell (first Brillouin zone) 
  which spans the dual lattice $\Gamma^*$.
  Assume that 
  $\sum_{\bk\in\Gamma^*}\  |p_\bk| <\infty$, 
  where $\{p_\bk\}$ denotes the set of Fourier coefficients of $p$.  
  Let $G\in L^1(\R^d)\cap C^\infty(\R^d)$. 
  Then, as $\eps\to 0$
  $$  
  \left|\  \eps^d\int_{\R^d} p(\bx)\ G(\eps\bx) d\bx\ -\ 
    \Av p(\bx)d\bx\cdot \int_{\R^d} G(\by) d\by\ \right|\ =\  
 \cO(\eps^\infty)\times \sum_{\bk\in\Gamma^*} |p_\bk| = \cO(\eps^\infty)~,
  $$
  where the cell average  $\Av$ is defined by $\Av p = \frac{1}{|\cell |}\int_{\cal B} p$.
\end{lemma}

\begin{proof} Proof of Lemma \ref{lem:asymp}.
  $p(\bx)$ has the Fourier representation
\begin{equation}
  p(\bx)\ =\ \sum_{\bk\in\Gamma^*} p_\bk\ e^{ i \bk\cdot \bx},\ \ 
  {\rm where}\ \ p_\bk = \Av e^{- i \bk\cdot \bx} p(\bx) d\bx~.
  \end{equation}
  Therefore,
 \begin{eqnarray*}
   \eps^d\int_{\R^d} p(\bx)\ G(\eps\bx) d\bx\ &=& 
   \eps^d\sum_{\bk\in\Gamma^*}p_\bk \int_{\R^d}e^{ i\bk\cdot \bx}G(\eps\bx)d\bx\\
   &=& p_0\int_{\R^d}G(\by)d\by\ +\ \sum_{0\ne \bk\in\Gamma^*}p_\bk
   \int_{\R^d}e^{ i \frac{\bk}{\eps}\cdot \by}G(\by)d\by\\
   &=& \Av  p(\bx) d\bx\ \int_{\R^d}G(\by)d\by
   \  +\ \sum _{0\ne \bk\in\Gamma^*}p_\bk 
   \overline{\hat{G}\left(\frac{\bk}{2\pi\eps}\right)}\ .
 \end{eqnarray*}
 By smoothness of $G$, for all $q\ge1$ and $\bxi\in\R^d$  there is a 
 positive constant, $r_{G,q}$, such that  
 $|\hat{G}(\bxi)|\le r_{G,q}(1+|\bxi|)^{-q}$. 
 The required estimate of the remainder term 
 follows. This completes the proof of the Lemma.
\vspace{2mm}
\end{proof}

\nit We now proceed with proof of the Corollary
\ref{cor:edgepower-crit}  
by evaluating the terms $\cI_j^\eps,\ j=0,1,2$ in Eq.~(\ref{eq:int-u2}).
 \\
 
 \nit{\bf Claim 1:} \
\begin{equation}
\cI_0^\eps = \eps^{-d}\ 
   \zetap\ \cPc\  +\  \cO(\eps^\infty),\ \ {\rm where}
\label{cI0eps}
 \end{equation}
 \begin{equation}
\zetap\ =\ \left(\frac{ \left(\Av w^2 \right)^{\sigma+1} }{\Av w^{2\sigma+2} }\right)^{\frac{1}{\sigma}}
  \ m_*^{-\frac{1}{2}}.
  \label{zeta*def}
  \end{equation}
 
  \begin{proof} 
By Eq.~(\ref{eq:U0})  and Lemma \ref{lem:asymp} one has 
$$  
\cI_0^\eps\ = \ \int \alpha_0\ d\bx \ = \ 
  \int_{\R^d}\ w^2(\bx)\ F^2(\eps\bx)\ d\bx 
   \ =\ \eps^{-d}\ \Av w^2(\bx)\ d\bx\ \int_{\R^d} F^2(\by)\ d\by\ +\ \cO(\eps^\infty).
$$
Using expression~(\ref{eq:FR}) for $F(\by)$ as a scaling of $R(\by;-1)$ 
we get
\begin{eqnarray*}
  \int_{\R^d} F^2(\by)\ d\by &=&
  (\G)^{-\frac{1}{\sigma}}\ m_*^{-\frac{1}{2}}\ \int_{\R^d} R^2(\by;-1)\ d\by
  \\ &=&
  \left(\frac{ \Av w^2  }{\Av w^{2\sigma+2} }\right)^{\frac{1}{\sigma}}
  \ m_*^{-\frac{1}{2}}\ \cPc~,
\end{eqnarray*}
  \end{proof}

\nit {\bf Claim 2:} $\cI_1^\eps\ =\ \cO(\eps^\infty)$.

\begin{proof} 
We proceed similarly by using (\ref{eq:U1}) and Lemma \ref{lem:asymp}. We obtain
\begin{eqnarray*}
  \cI_1^\eps &=& \int \alpha_1\ d\bx  \ = \
  2\eps \ \sum_{j=1}^d\ \int_{\R^d}\  w(\bx)\ F(\eps\bx)\cdot
  2L_*^{-1}\ [ \D_{x_j} w](\bx)\D_{y_j}F(\eps\bx)\ d\bx\\
  &=& 4\eps^{1-d} \Av w(\bx) L_*^{-1}\ ( \D_{x_j} w)(\bx)\ d\bx 
  \int F(\by)\D_{y_j}F(\by)\ d\by\ +\ \eps\ \cO(\eps^\infty)
  \ = \ \cO(\eps^\infty)~,
\end{eqnarray*}
since  $\int_{\R^d} F(\by)\D_{y_i}F(\by)\ d\by=0$.
  \end{proof}

 Finally, we turn to $\cI_2^\eps$.\\ \\
 
 \nit {\bf Claim 3:}
\begin{align}
 \cI_2^\eps\ =\ \eps^{-d}\cdot\  \eps^2\ &\left[\  4\sum_{j=1}^d\ \Av \left| L_*^{-1}\left[\D_{x_j}w(\bx)\right]\right|^2\ d\bx\ \int_{\R^d}\ \left|\ \D_{y_j}F(\by)\ \right|^2\ d\by\right.\nn\\
&\left.+\ \dashint_\cell w^2(\bx) d\bx\ \int \D_\Omega F(\by)\ S(\by)\ d\by\ \right]+\ \cO(\eps^\infty).\nn\\
& \label{I2eps}
 \end{align}
 where $S(\by)$ is explicitly displayed in (\ref{Sofydef}).\\ 

 \begin{proof}
   \begin{align}
     \cI_2^\eps\ = \int \alpha_2\ d\bx &= \eps^2\int 
     \left[\ |U_1(\bx,\eps\bx)|^2  \ + \ 
       2 U_0(\bx,\eps\bx) \left(U_{2p}(\bx,\eps\bx)\ +\ w(\bx)F_{2h}(\eps\bx)\ \right) \ \right] \ d\bx~\nn\\
     &=\ \cI_{2,a}\ +\ \cI_{2,b}\ +\ \cI_{2,c}\ \ .\nn\end{align}
   and, by Lemma \ref{lem:asymp},\ $\cI_{2,a}^\eps$ is given by
   \begin{align}
     \cI_{2,a}^\eps \ &\equiv \  \eps^2\int |U_1(\bx,\eps\bx)|^2 \ d\bx \ =\ 
     4\ \eps^2\ \int\ \left| L_*^{-1}\left[\D_{x_j}w(\bx)\right]\ \D_{y_j}F(\eps\bx)\ \right|^2\ d\bx\nn\\
     &\ =\ \eps^{2-d}\ \left(\    4\sum_{j=1}^d\ \Av \left| L_*^{-1}\left[\D_{x_j}w(\bx)\right]\right|^2\ d\bx\cdot \int_{\R^d}\ \left|\ \D_{y_j}F(\by)\ \right|^2\ d\by\ +\ \cO(\eps^\infty)\ \right)\nn
   \end{align}
   Concerning $\cI_{2,b}$, we assert the following:\\
   \begin{equation} 
     \cI_{2,b}^\eps\ = \ 2\eps^2 \intRd w(\bx) F(\eps \bx)  U_{2p}(\bx,\eps\bx) \ d\bx\ =\  
     \cO(\eps^\infty).\nn\\
     \label{I2b-claim}\end{equation}
 \end{proof}

 \begin{proof}
   To prove (\ref{I2b-claim}) we note that $U_2$ [Eq.~(\ref{U2fac-2})] is of a sum of terms
   that have the factored form
   \begin{equation}
     U_{2p}(\bx,\eps\bx) = \sum_j\ G_j(\eps\bx)\cdot  L_*^{-1} P^\perp g_j(\bx)~.
     \label{factored}\end{equation}
   Here, $P^\perp$ denote the projection onto the orthogonal complement 
   of $w$ in $L_{periodic}^2(\cell)$; see (\ref{Pperp}).
   Substitution of (\ref{factored}) gives
   \begin{equation}
     \cI_{2,b}^\eps\ = \ 
     2\eps^2 \ \sum_j\intRd w(\bx)\ L_*^{-1} P^\perp g_j(\bx)\cdot F(\eps \bx)  G_j(\eps\bx) \ d\bx\nn\end{equation}
   which by Lemma \ref{lem:asymp} implies
   \begin{equation}
     \cI_{2,b}^\eps\ = \ \sum_j\Av w(\bx)\ L_*^{-1} P^\perp g_j(\bx) d\bx\ \int F(\by) G_j(\by) d\by\ +\ \cO(\eps^\infty).
     \nn\end{equation}
   Since $P^\perp$ commutes with functions of $L_*$
   and $P^\perp w=0$ we have
   \begin{align}
     \langle w(\bx), L_*^{-1} P^\perp g_j\rangle =
     \langle w,  L_*^{-1} P^\perp\ P^\perp g_j\rangle\ 
     = \langle P^\perp w,  L_*^{-1} \ P^\perp g_j\rangle = 0.
     \nn\end{align}
   It remains to calculate $\cI_{2,c}$.
   \begin{align}
     \cI_{2,c}\ &=\ 2 \eps^2\ \int U_0(\bx,\eps\bx) U_{2h}(\bx,\eps\bx)\ d\bx\nn\\
     &=\ 2 \eps^2\ \int w(\bx)F(\eps\bx)\cdot w(\bx)F_{2h}(\eps\bx)\ d\bx\nn\\
     &=\ 2 \eps^{2-d}\dashint_\cell w^2\cdot \int F(\by)\ F_{2h}(\by)\ d\by\ +\ \cO(\eps^\infty),\nn\\
     &=\  2 \eps^{2-d}\dashint_\cell w^2\cdot \int F(\by)\ \left(L_+^A\right)^{-1}S(\by)\ d\by\ +\ \cO(\eps^\infty)\nn\\
     &=\  2 \eps^{2-d}\dashint_\cell w^2\cdot \int \left(L_+^A\right)^{-1} F(\by)\ S(\by)\ d\by\ +\ \cO(\eps^\infty)\nn\\
     &=\ - \eps^{2-d}\dashint_\cell w^2\cdot  \left(\frac{1}{\sigma}F(\by)+\by\cdot\nabla_\by F(\by)\right) \
     \ S(\by)\ d\by\ +\ \cO(\eps^\infty)
   \end{align}
Here we have used the relation 
\begin{align}
\left(L_+^A\right)^{-1} F= \left.\D_{\mu_2} F(\cdot,\mu_2)\right|_{\mu_2=-1}\ =\ -\frac{1}{2}\left(\frac{1}{\sigma}F(\by)+\by\cdot\nabla_\by F(\by)\right)\nn\\ 
&\nn\end{align} which follows from differentiation of the equation for $F=F(\bx;\mu_2)$  with respect to $\mu_2$; see \eqref{eq:Feqn-1} and \eqref{Fmu2scale}.
 \\ \\
 Therefore, summing up the terms we have
 \begin{align}
 &   \int_{\R^d} |u_\eps(\bx)|^2\ d\bx\ =\
     \eps^{2\over\sigma}\left[\ \cI_0^\eps\ +\ \cI_1^\eps\ +\ \cI_2^\eps\ +\ \cO(\eps^3)\ \right]\nn\\
 & =\   (\eps^2 )^{\frac{1}{\sigma}-\frac{d}{2}} \zetap\ \cPc\nn\\
 & +
   (\eps^2)^{\frac{1}{\sigma}-\frac{d}{2}+1}\ \left[\ 4
    \sum_{j=1}^n\Av\ \left| L_*^{-1}\left[\D_{x_j}w(\bx)\right]\right|^2\ d\bx\ \int_{\R^d}\ \left|\ \D_{y_j}F(\by)\ \right|^2\ d\by\ \right.\nn\\
    &\ \ \ \ \ \ \ \ \ -\ \   \left. \dashint_\cell w^2\ d\bx
    \int \left(\frac{1}{\sigma}F(\by)+\by\cdot\nabla_\by F(\by)\right)\ S(\by)\ d\by\ \right]\ +\ \cO(\eps^\infty)\nn\\
    & =\ (\eps^2)^{\frac{1}{\sigma}-\frac{d}{2}}\ 
    \left(\ \zeta_*\cPc\ + \eps^2\zeta_{1*}\ +\cO(\eps^\infty)\ \right)
    \nn
\end{align}
 Recall 
$S(\by)$ is displayed in (\ref{Sofydef}).
 \end{proof}

This concludes the proof of Theorem \ref{theo:edgepower}.



\section{The error estimate \eqref{eta-est} 
and conclusion of the proof of Theorem \ref{theo:2scale}}
\label{sec:error-est}

In this section we prove Theorem \ref{theo:2scale}. 
For ease of presentation, we focus on the cubic ($\sigma=1$) one-dimensional case ($d=1$):
\begin{equation}
\left(\ -\D_x^2 + V(x)\ \right)u\ -\ u^3\ =\ \mu\ u
\nn\end{equation} 
The proof carries over to the more  setting in the statement of Theorem \ref{theo:2scale}.\ After the proof, we indicate the modifications required for the proof to go through in general dimension $d=1,2,3$; see Remark
 \ref{rmk:extensions} below. 

We shall construct a solution $(u,E)=(u_\veps,\mu_\veps)$, 
using the  formal multiple scale expansion of  Section~\S\ref{sec:homog-expansion}
:  
\begin{align}
u_\veps\ &=\ \veps\ U^\veps(x)\ =\ \veps\ \left[\ \sum_{k=0}^4\veps^k\ U_k(x,y)
 +\ \veps^3\ U_5^\veps(x)\ \right]\label{u-eps}\\
\mu_\veps\ &=\ E_*\ - \veps^2 \label{u-mu-eps}
\end{align}
The expansion includes an error term, $\veps^3U_5^\veps(x)$, which must be estimated. The equation for $\veps^3U_5^\veps$ is:
\begin{align}
&\left(\ -\D_x^2+V(x)-3\veps^2U_0^2(x,\veps x)\ -\ E_*+\veps^2\ \right)\  U_5^\veps(x)\nn\\
&=\ \ \veps^2\ R^\veps[\ U_j;0\le j\le4, U_5^\veps(x)\ ]\nn\\
&\equiv\ \  \veps^2\ R_0^\veps[U_j;0\le j\le4] +\ 
 \veps^3\ R_2^\veps[U_j;0\le j\le4] \ U_5^\veps\nn\\
 &\ \ \ \  \ \  +\ \veps^5\ R_2^\veps[U_j;0\le j\le4] \left(\ U_5^\veps\ \right)^2\ +\ \veps^{8}\ \left(\ U_5^\veps\ \right)^3,
\label{U5-corrector}\end{align}
where $R_k^\veps[U_j;0\le j\le4]$ denotes the coefficient of the $k^{th}$ power of $U_5^\veps$, and is a polynomial in the previously constructed functions $U_j,\ j=0,1,2,3,4$.

 The scaling of the error term in (\ref{u-eps}) is motivated as follows.
  Formally, the correction to the leading order sum in (\ref{u-eps})
   will be of order $\veps^5$. In our analysis, we find that the frequency components of the corrector (to the truncated multiple scale expansion near the band edge are of order $\cO(\eps^3)$ .  Therefore anticipate this result in \eqref{u-mu-eps}.
    We will in fact show that for $s>d/2$, $\|U_5^\veps\|_{H^s}$ is bounded uniformly in $\veps$. This implies the error bound of Theorem \ref{theo:2scale}.
    
 In particular, 
\begin{equation}
R_0^\veps\ =\ 2\D_x\D_y U_4+(\D_y^2-1)U_3+3U_0^2+6U_0U_1U_2+U_1^3\ +\ \cO(\veps)\ .
 \nn\end{equation}
 Our goal is to estimate $U_5^\veps$ and to do this we employ the spectral (Floquet - Bloch) decomposition of the operator $-\Delta+V(x)$. 
 

\subsection{Floquet-Bloch Theory and the Bloch transform}
See the references \cite{Eastham:73,RS4,Kuchment-01} for basic results on 
the spectral theory of  operators with periodic coefficients.

Assume $V(x+2\pi)=V(x)$. For each $k\in\T=[-\frac{1}{2},\frac{1}{2}]$
 we seek solutions of the eigenvalue equation  for the operator $\left(-\D_x^2+V(x)\right)$ of the form:
 \begin{equation}
 u(x;k)=e^{ikx}p(x;k),\ \ \ p(x+2\pi;k)=p(x;k),\ \ x\in\R
 \nn\end{equation}
 This yields the periodic elliptic eigenvalue problem for $p(x;k)$
 \begin{equation}
 \left(-(\D_x+ik)^2+V(x)\right)\ p(x;k)\ =\ E\ p(x;k),\ \ \ \ p(x+2\pi;k)=p(x;k)
 \nn\end{equation}
 For each $k\in\T$ the spectrum is discrete give rise to eigenpairs
  $(E_m(k),p_m(x;k))_{m\ge1}$, and a complete orthonormal set $\{p_m(x;k)\}$
   in $L^2_{per}$ with respect to the inner product:
   \begin{equation}
   \langle f,g\rangle_{L^2_{per}}\ =\ \int_0^{2\pi}\overline{f(x)}g(x)\ dx.
   \nn\end{equation}
  $(E_n(k),u_n(x;k)),\ {n\ge1}, \ k\in[-1/2,1/2]$ are solutions of the eigenvalue problem:
  \begin{align}
&\left(-\D_x^2+V(x)\right) u_n(x;k)\ =\ E_n(k)\ u_n(x;k)\nn\\
& u_n(x+2\pi;k)= e^{2\pi ik}u_n(x;k),\ x\in\R,
 \nn\end{align}
 where $k\mapsto E_m(k)$ sweeps out the m$^{th}$ spectral band, and yield a complete set of states in $L^2(\R)$; see \eqref{completeness}.
 
\nit{\bf N.B} In this section we  assume that $w(x)$ is normalized, $\langle w,w\rangle=1$.
 Thus, $w(x)$ is the unique normalized ground state of the periodic boundary value problem and 
 \begin{equation}
 (p_1(x;0),E_1(0))\ =\   (p_1(x;0),E_1(0))\ =\ (w(x),E_*)
 \nn\end{equation}

 Furthermore, for each $k\in\T$, the set $\{p_n(x;k)\}$ is an orthonormal  set  in $L_{\rm per}^2(\ [0,2\pi)\ )$.

Introduce the Gelfand - Bloch transform $(\cT\phi)(x;k)=\tilde{\phi}(x;k)$, and its inverse $\cT^{-1}$:
\begin{align}
(\cT\phi)(x;k)&=\tilde{\phi}(x;k)=\sum_{m\in\Z^d}e^{im\cdot x}\hat{\phi}(k+m)\label{cT}\\
(\cT^{-1}\tilde\phi)(x)&=\int_{[-\frac{1}{2},\frac{1}{2}]^d}e^{ik\cdot x}\tilde{\phi}(x;k)dk,
\nn\end{align}
where $\hat\phi(k)$ denotes the Fourier transform of $\phi(x)$. Clearly we have
\begin{equation}\
\tilde\phi(x+2\pi;k)=\tilde\phi(x;k),\ \ {\rm and}\ \ \tilde\phi(x;k+1)=e^{-ix}\ \tilde\phi(x;k)
\nn\end{equation}
One can check that
\begin{equation} \cT\ \cT^{-1}\ =\ {\rm Identity\ on}\ L^2(\R).\ \nn\end{equation}

Another important property of $\cT$ is that it commutes with multiplication by a periodic function:
\begin{equation}
f(x+2\pi )=f(x)\ \implies\ (\cT fg)(x;k)=f(x)\ (\cT g)(x)
\nn\end{equation}
Since $\tilde\phi(x;k)$ is $2\pi-$ periodic in $x$, we have
\begin{equation}
\tilde\phi(x;k)\ =\ \sum_{m\ge1}\langle p_m(\cdot;k),\tilde\phi(\cdot;k)\rangle p_m(x;k)\label{tphi-exp}\end{equation}
We  conclude this subsection with some basic definitions and results
 required below; see, for example, \cite{DU:09} and references cited therein.
\begin{theorem}\label{theo:cXs}
\begin{enumerate}
\item There exist positive constants $c_1, c_2$  band dispersion functions $E_n(k),\ \ n\ge1$ satisfy the bounds
\footnote{In dimension $d$, $n^2$ is replaced by $n^{\frac{2}{d}}$.}
\begin{equation}
c_1\ n^2\ \le\ E_n(k)\ \le\ c_2\ n^2,\ \ |k|\le 1/2\ ;
\label{omega-n-bound}\end{equation}
see \cite{CH1,Hor4}.
\item The mapping
\begin{equation}
\phi(x)\ \mapsto\ \left(\ \left\langle\ \tilde\phi(\cdot,k)\ ,\ p_n(\cdot,k)\ \right\rangle\ \right)_{n\ge1}\ \equiv\ \left(\ \tilde\phi_n(k)\ \right)_{n\ge1}\nn\end{equation}
is an isomorphism of $H^s(\R^1)$ with $\cX^s=L^2(\T^1;l^{2,s})$, with norm:
\begin{align}
& \left\|\ \left(\ \tilde\phi_n(k)\ \right)_{n\ge1}\ \right\|^2_{\cX^s}\ \equiv\ 
\left\|\ \left(\ \left\langle\ \tilde\phi(\cdot,k)\ ,\ p_n(\cdot,k)\ \right\rangle  \right)_{n\ge1}\ \right\|_{\cX^s}^2\nn\\
& =\ 
\int_{\T}\ dk\ \sum_{n\ge1}\  (1+|n|^2)^s\ \left|\ \left\langle\ \tilde\phi(\cdot,k)\ ,\ p_n(\cdot,k)\ \right\rangle\ \right|^2\   \nn\\
&
\label{cXs-def}
\end{align}
\item Moreover, there exist positive constants $C_1, C_2$, such that we have the norm equivalence
\begin{equation}
C_1\ \|\phi\|_{H^s}\ \le\ \left\|\ \left\langle\ \tilde\phi(\cdot,k)\ ,\ p_n(\cdot,k)\ \right\rangle_{n\ge1}\ \right\|_{\cX^s}\ \le\ C_2\ \|\phi\|_{H^s}
\label{norm-equiv}
\end{equation}
\item Assume $\phi, \psi\in H^s(\R^d)$. \\
(i)\  If $s>q+d/2$, then $\phi\in C^q_{\downarrow}(\R^d)$, the space of $C^q$ functions, $f$, with $|\D^\alpha f(x)|\to0$ as $x\to\infty$, $|\alpha|\le q$.\\
(ii)\ If $s>d/2$ then $H^s$ is an algebra, {\it i.e.} $\phi\psi\in H^s$ and $\|\phi\psi\|_{H^s}\le C\ \|\phi\|_{H^s}\ \|\psi\|_{H^s}$.
\end{enumerate}
\end{theorem}

\begin{remark}
  The bounds (\ref{omega-n-bound}) are well known; see \cite{CH1,Hor4}. 
  To prove the isomorphism, recall the operator 
  $L_*= -\Delta+V-E_*\ge0$; see (\ref{L*def} ). 
  Standard elliptic theory implies 
  $\phi\mapsto\| L_*^{s\over2}\phi\|_{L^2}$ defines a norm 
  equivalent to the $H^s$ norm. Furthermore, by (\ref{tphi-exp})
  \begin{align}
    \|\phi\|_{H^s}^2\sim\| (I+L_*)^{s\over2}\phi\|_{L^2}^2\ &=\ \left\|\ \int_{[-{1\over2},{1\over2}]}\ e^{ik\cdot}\sum_{n\ge1}\tilde\phi_n(k)\left(1+ E_*-E_n (k)\right)^s\ p_j(\cdot,k)\ \right\|_{L^2}^2\nn\\
    &=\sum_{n\ge1}\ \int_{[-{1\over2},{1\over2}]}\ |\tilde\phi_n(k)|^2\  
    |1+ E_*-E_n(k)|^{s}\ dk\nn\\
    &\sim\ \sum_{n\ge1}\ \int_{[-{1\over2},{1\over2}]}\ |\tilde\phi_n(k)|^2\ (1+|n|^2)^{s}\ dk\nn\\
    &\equiv \ \left\| \left(\ \tilde\phi_n(k)\ \right)_{n\ge1}\ \right\|_{\cX^s}^2.
    \label{isonorm-pf}
  \end{align} 
\end{remark}


\subsection{Corrector equation and localization in Bloch variables}
In this subsection we express the  equation for the corrector
\begin{equation} \Psi^\veps(x)\ \equiv\ U_5^\veps(x)
\label{Psidef}
\end{equation}
in Floquet-Bloch variables and, in particular, decompose this equation
into spectral components near and away from the band edge $E_*$;  see, for example, 
\cite{Busch-etal:06,DPS:09b,DU:09}.

 Applying the Bloch transform, $\cT$, to 
 (\ref{U5-corrector}) we obtain an equation for $\tilde\Psi^\eps(x;k)=(\cT\Psi^\veps)(x;k)$:
 \begin{align}
&\left[\ -\left(\D_x+ik\right)^2+V(x) -E_*+\veps^2\right]
(\cT{\Psi}^\veps)(x;k)\ \nn\\ 
& \qquad\qquad  -\ 3\veps^2w^2(x)
\cT\left[F^2(\veps\cdot){\Psi}^\veps(\cdot)\right](x;k)\nn\\
&\qquad\qquad\ =\  \veps^2\ \left(\cT R^\veps\right)(x;k),
\label{TPsi-eqn}\end{align}
where $\veps^2\ R^\veps$ is defined in  (\ref{U5-corrector}). 
Here, we have used that $U_0(x,y)=w(x)F(y)$.

Now $\tilde\Psi^\veps(x;k)$ is periodic in $x$. Therefore, 
\begin{equation}
\tilde\Psi^\veps(x;k)=\sum_{m=1}^\infty\ \tilde\Psi^\veps_m(k)\  p_m(x;k),\ \ \ \tilde\Psi^\veps_m(k)\equiv\langle p_m(\cdot;k)\ ,\ \tilde\Psi^\veps(\cdot;k)\ \rangle
\nn\end{equation}
We introduce  a decomposition of $\Psi^\veps$   into 
spectral components
 near the band edge $E_1(k=0)=E_*$ ({\it low frequencies}) and spectral components away from $E_*$ ( {\it high frequencies}) as follows. Let  $1_A$ denote the characteristic function for the set $A$ and define
 \begin{equation}
 \chi\left(a\le k\le b\right)\equiv 1_{\{k: a\le k\le b\}}\ \ \ .
 \nn\end{equation}

\nit Express $\tilde\Psi^\veps(x,k)$ as
 \begin{align}
&\tilde\Psi^\veps(x;k)= 
\underbrace{ \chi\left(|k|\le\veps^r\right)\Psi^\veps_1(k)\ p_1(x;k)}_{\tilde\Psi^\veps_{low}(x;k)}\nn\\
& +\ \underbrace{\chi\left(\veps^r\le|k|\le2^{-1}\right)\Psi^\veps_1(k)\ p_1(x;k)\ +\ \sum_{m\ge2}\ \tilde\Psi^\veps_m(k)\  p_m(x;k)}_{\tilde\Psi^\veps_{high}(x;k)},
\label{hl-decomp}\end{align}
 where $r$ is chosen to satisfy
 \begin{equation}
 0\ <\ r\ <\ 1\ .\label{r-restrict}
 \end{equation}
 Using the inverse Bloch transform, we obtain
 \begin{align}
 \Psi^\veps(x)\ &=\ \cT^{-1}\Psi^\veps_{low}(x;\cdot)
\ +\  \cT^{-1}\Psi^\veps_{high}(x;\cdot)\nn\\
\ & =\ \Psi^\veps_{low}(x)\ +\ \Psi^\veps_{high}(x)
 \nn\end{align}
 Taking the inner product of (\ref{TPsi-eqn}) with $p_j(\cdot;k)$, we obtain:
\begin{align}
&\left[\ E_j(k)-E_*+\veps^2\ \right]\ \tilde\Psi^\veps_j(k)
\ -\ 3\veps^2\left\langle p_j(\cdot;k)\ ,\  w^2(\cdot)
\cT\left[F^2(\veps\cdot){\Psi}^\veps(\cdot;k)\right](\cdot;k)\right\rangle\nn\\
&=\ \veps^2\left\langle p_j(\cdot,k)\ ,\ (\cT R^\veps)(\cdot,k)
\right\rangle\ \equiv\ \veps^2\ (\cT R^\veps)_j(k),\ \ \ j\ge1\nn\\
&\label{band-j}\end{align}
The system (\ref{band-j}) can be viewed as two coupled systems 
for the low and high frequencies:
\begin{align}
&\tilde\Psi^\veps_{1,low}(k)\equiv \chi\left(|k|\le\veps^r\right)\Psi^\veps_1(k)
\ \ \ {\rm and}\nn\\
& \tilde\Psi^\veps_{high}(k)\ \equiv\ \left(\   \chi\left(\veps^r\le|k|\le2^{-1}\right)\tilde\Psi^\veps(k)
\ ,\ \{\tilde\Psi^\veps_j(k)\}_{j\ge2}\ \right)\nn\\
\label{lowhigh}
\end{align}
{\bf Low frequency components:}
\begin{align}
&\left[\ E_1(k)-E_*+\veps^2\ \right]\ \tilde\Psi^\veps_{1,low}(k)\nn\\
&\ \ \ \ \ \ \ \ \ \ -\ 3\veps^2\  \chi\left(|k|\le\veps^r\right)\ 
\left\langle p_1(\cdot;k)\ ,\  w^2(\cdot)
\cT\left[F^2(\veps\cdot)\Psi_{1,low}^\veps(\cdot)\right](\cdot;k)\right\rangle\nn\\
&=\  3\veps^2\  \chi\left(|k|\le\veps^r\right)\ 
\left\langle p_1(\cdot;k)\ ,\  w^2(\cdot)
\cT\left[F^2(\veps\cdot)\Psi_{high}^\veps(\cdot)\right](\cdot;k)\right\rangle\ +\ \veps^2\ \tilde R^\veps_{1,low},\nn\\
\label{low-eqn}\end{align}
\\ 
{\bf High frequency components:}
\begin{align}
&\left[\ E_1(k)-E_*+\veps^2\ \right]\   \chi\left(\veps^r\le|k|\le2^{-1}\right)\tilde\Psi_1^\veps(k)\nn\\
&\ =\ 3\veps^2\   \chi\left(\veps^r\le|k|\le2^{-1}\right)\ \left\langle p_1(\cdot;k)\ ,\  w^2(\cdot)
\cT\left[F^2(\veps\cdot){\Psi}^\veps(\cdot)\right](\cdot)\right\rangle\
 +\ \veps^2 \tilde R^\veps_{1,high}\nn\\
&\label{high-eqns-a}\\ 
&\left[\ E_j(k)-E_*+\veps^2\ \right]\ \tilde\Psi^\veps_j(k)\nn\\
&\ =\ 3\veps^2\left\langle p_j(\cdot;k)\ ,\  w^2(\cdot)
\cT\left[F^2(\veps\cdot){\Psi}^\veps(\cdot)\right](\cdot;k)\right\rangle\
 +\ \veps^2\ \tilde R^\veps_{j,high},\ \ \ j\ge2\nn\\
\label{high-eqns-b}\end{align}
Here, $\tilde R^\eps_{1,low}$ and $\tilde R^\eps_{high}\ =\ (\tilde R^\eps_{j,high})_{j\ge1}$ are given by
\begin{align}
\tilde R^\eps_{1,low}\ &=\ \chi(|k|\le\veps^r)\ \langle p_1(\cdot,k)\ ,\ (\cT R^\veps)(\cdot,k)\rangle\label{tRlow}\\
\tilde R^\eps_{1,high}\ &=\ \chi(\veps^r\le|k|\le2^{-1})\ \langle p_1(\cdot,k)\ ,\ (\cT R^\veps)(\cdot,k)\rangle\label{tRhigh-1}\\
\tilde R^\eps_{j,high}\ &=\  \langle p_j(\cdot,k)\ ,\ (\cT R^\veps)(\cdot,k)\rangle,\ j\ge2,
\label{tRhigh-j}
\end{align}
where $R^\veps$ is defined in (\ref{U5-corrector}).
We study the system for $\tilde\Psi_{1,low}^\veps(k),\ \  \tilde\Psi_{high}^\veps(k)$ using the following \\ \\
{\bf Lyapunov-Schmidt reduction strategy:}
\begin{enumerate}
\item Using the implicit function theorem, solve the infinite system of high frequency component equations for   $\tilde\Psi^\veps_{high}$ as a functional of $\tilde\Psi^\veps_{1,low}$:\ \ $\tilde\Psi^\veps_{high}= \tilde\Psi^\veps_{high}\left[\tilde\Psi^\veps_{1,low}\right]$ with an appropriate bound on this mapping.
\item Substitute $\tilde\Psi^\veps_{high}= \tilde\Psi^\veps_{high}\left[\tilde\Psi^\veps_{1,low}\right]$ into (\ref{low-eqn}) to obtain a closed equation for  the low frequency components, which is solved via fixed point iteration.
\end{enumerate}

\nit We now embark on implementing this strategy. Our first step is to
rewrite the low frequency equation (\ref{low-eqn}) in appropriately
rescaled variables. 


\subsection{Closure and rescaled of  low frequency equation for $\tilde\Psi^\veps_{1,low}$} 

With a view toward obtaining a closed equation for $\tilde\Psi^\veps_{1,low}$, we begin with several observations.
\begin{enumerate}
\item From our formal multi-scale construction, we expect $\Psi^\veps_{1,low}(x)\sim \Phi(\veps x)\ w(x)$. This motivates the following
\\ \\ 
{\bf  Ansatz:}\ Seek the low frequency components in the form
\begin{equation}
\tilde\Psi_{low}^\veps\left(k\right)\ =\   \chi\left(|k|\le\veps^r\right)\ \frac{1}{\veps}\ \hat{\Phi}\left(\koe\right),
\label{Phi-def} 
\end{equation}
Thus,
\begin{equation}
\tilde\Psi_{low}^\veps(x;k)\ =\ \chi\left(|k|\le\veps^r\right)\ \frac{1}{\veps}\ \hat{\Phi}\left(\koe\right)\ p_1(x;k).\nn\end{equation}
Using the definition of $\cT^{-1}$ and that\\ $p_1(x;k)=p_1(x;0)+\cO(k)\ =\ w(x)+\cO(\veps^r),\ \ |k|\le\veps^r$ we have:
\begin{equation}
\Psi_{low}^\veps(x)\ =\ \Phi(\veps x)\ w(x)\ +\ 
 \cO(\veps^r).
\nn\end{equation}
\item Note that for $|k|<\veps^r$,  
\begin{equation}
 E_1(k)-E_*-\frac{1}{2}\D_k^2 E_1(0)k^2\ =\   \frac{1}{6} \D_k^3E_1(\tilde k)\  \tilde{k}^3,\ \  0\le \tilde{k}\le\veps^r
\nn\end{equation}
Thus,  
\begin{align}
\left(E_1(k)-E_*-\veps^2\right)\tilde\Psi_{1,low}^\veps(k) &=\ 
 \left(\frac{1}{2}\D_k^2 E_1(0)k^2-\veps^2\right) \chi(|k|\le \veps^r)\frac{1}{\veps}\hat{\Phi}\left( \koe \right)\nn\\
& + \cO\left(\ \|\D_k^3E_1\|_\infty \veps^{3r}\  \frac{1}{\veps}\chi(|k|\le \veps^r)\hat{\Phi}\left( \koe \right)\ \right)
\label{exp-EmE*}
\end{align}
\nit This and the  Ansatz (\ref{Phi-def})  suggest the scaling
\begin{equation}
\kappa \equiv \frac{k}{\veps}
\label{kappa-def}
\end{equation}
In this scaling \eqref{exp-EmE*} becomes 
\begin{align}
&\left(E_1(k)-E_*-\veps^2\right)\tilde\Psi_{1,low}^\veps(k)\nn\\
& =\ 
\veps^2\left(\ \frac{1}{2}\D^2 E_1(0)\kappa^2-1\ \right)
 \chi(|\kappa|\le \veps^{r-1})\frac{1}{\veps}\hat{\Phi}(\kappa)\nn\\
 &\ \ \ \ \ \ +\ \cO\left(\ \eps^{3r}\ \chi(|\kappa|\le \veps^{r-1})
  \frac{1}{\veps} \hat{\Phi}(\kappa)\ \right)
\label{exp-EmE*-A}
\end{align}
\item  Consider the last term on the left hand side of (\ref{low-eqn}).
 We have 
 \begin{align}
 &\ \ \ \ \ \ -\ 3\veps^2\  \chi\left(|k|\le\veps^r\right)\ 
\left\langle p_1(\cdot;k)\ ,\  w^2(\cdot)
\cT\left[F^2(\veps\cdot)\Psi_{1,low}^\veps(\cdot)\right](\cdot;k)\right\rangle\nn\\
&=\ -\ 3\veps^2\  \chi\left(|k|\le\veps^r\right)\ \left\langle p_1(\cdot;k)\ ,\  w^2(\cdot)
\cT\left[F^2(\veps\cdot)\ \chi\left(|\nabla_x|\le\veps^r\right)\Phi(\veps\cdot)\ w(\cdot)\right](\cdot;k)\right\rangle\ +\ \cO(\veps^{r+2})\nn\\
&=-\ 3\veps^2\  \chi\left(|k|\le\veps^r\right)\ \left\langle p_1(\cdot;k)\ ,\  w^3(\cdot)
\cT\left[F^2(\veps\cdot)\ \chi\left(|\nabla_x|\le\veps^r\right)\Phi(\veps\cdot)\right](\cdot;k)\right\rangle\ +\ \cO(\veps^{r+2})\nn\\
&=-\ 3\veps^2\ \int w^4\ dx\cdot
  \chi\left(|k|\le\veps^r\right)\ \left[F^2(\veps\cdot)\ \chi\left(|\nabla_x|\le\veps^r\right)\Phi(\veps\cdot)\right]\hat{ }\ (k) +\ \cO(\veps^{r+2})\nn\\
&=-\ 3\veps^2\ \int w^4\ dx\cdot
   \chi\left(|k|\le\veps^r\right)\ \frac{1}{\veps} \left[F^2\ \chi\left(|\nabla_y|\le\veps^{r-1}\right)\Phi\right]\hat{ }\ \left(\koe\right)\ +\ \cO(\veps^{r+2})\nn\\
 &=-\ 3\veps^2\ \int w^4\ dx\cdot
 \chi\left(|\kappa|\le\veps^{r-1}\right)\ \frac{1}{\veps} \left[F^2\ \chi\left(|\nabla_y|\le\veps^{r-1}\right)\Phi\right]\hat{ }\ \left(\kappa\right)\ +\ \cO(\veps^{r+2})
&\nn\end{align}
\end{enumerate}
Use of (\ref{Phi-def}) and (\ref{kappa-def}) in (\ref{low-eqn}) yields the following:\\ \\
{\bf Closed / Rescaled low frequency equation for $\hat\Phi(\kappa)$:} 
\begin{align}
&\left(\ \frac{1}{2}E_1''(0)\kappa^2+1\right)\ 
\chi\left(|\kappa|\le\veps^{r-1}\right)\ \hat{\Phi}(\kappa)\ 
 -\ 3\ \G\ \left[\ F^2\ \chi\left(|\nabla_y|\le\veps^{r-1}\right)\   \Phi \ \right]\hat{ }\ (\kappa)\nn\\
 & =\  \chi\left(|\kappa|\le\veps^{r-1}\right)\left[\ \hat{R}^\veps_{\rm rescaled\ low}\left(\ \kappa;\Phi,\Psi^\veps_{high}\ \right) + \cO\left(\| \D_k^3E_1\|_\infty\ \veps^{3r}\ |\ \hat\Phi(\kappa)|\ \right) \right] 
 \label{Rescaled-low-eqn}
 \end{align} 
 Here, $E_1:[-1/2,1/2]\to\left[E_*, E_1(\frac{1}{2})\right],\ \ \ k\mapsto E_1(k)$ denotes the band dispersion function for the first spectral band and 
$\G$ is given by the expression in (\ref{eq:geff}).  As we have assumed $p_1(x,0)=w(x)$ to be normalized, $\G=\int_0^{2\pi}w^4$.\\

\nit We summarize the arguments of this subsection, which lead to the system we'll study:
\begin{proposition}\label{prop:Rescaled-System}
 The coupled system consisting equation  (\ref{Rescaled-low-eqn}) for the rescaled low frequency components, $\hat\Phi(\kappa),\ \ |\kappa|\le\veps^{r-1}\ \ (|k|\le\veps^r)$, and the {\it high frequency equations}, (\ref{high-eqns-a})  and (\ref{high-eqns-b}) is equivalent to the original system.
 \end{proposition} 
 

\subsection{Proof by Lyapunov-Schmidt reduction}

\nit We estimate the right hand sides of the high frequency equations
 (\ref{high-eqns-a}) and  (\ref{high-eqns-b}).\\ 
\begin{proposition}\label{rhsides} 
Let $s>d/2$. For some  positive constants $C_1$ and $ C_2$ we have\\ 
\begin{align}
&\left\|\ 3\veps^2\left\langle p_j(\cdot;\star)\ ,\  w^2(\cdot)
\cT\left[F^2(\veps\cdot){\Psi}^\veps(\cdot)\right](\cdot;\star)\right\rangle\ \right\|_{\cX^s}\ \le\ C_1\ \veps^2\ \|\tilde\Psi^\veps\|_{\cX^s}\nn\\
&\left\|\ \tilde R^\eps_{j,high}\ \right\|_{\cX^s} \le\ C_2\left(\ \cO(\veps^{\infty})\ +\ \veps^3\|\tilde\Psi^\veps\|_{\cX^s}\ +\ \veps^5\|\tilde\Psi^\veps\|^2_{\cX^s}\ +\ \veps^{8}\ \|\tilde\Psi^\veps\|^3_{\cX^s}\ \right).
\nn\end{align}
\end{proposition}
\begin{proposition}\label{prop:Psi-highofPsi-low}
  The system for $ \tilde\Psi^\veps_{high}$ can be solved 
  in terms of\\
   $\tilde\Psi^\veps_{1,low}=\chi(|k|\le\veps^r)\tilde\Psi^\veps_{1,low}$ and we have the estimate:
  \begin{align}
&\left\|\ \tilde\Psi^\veps_{high}\left[ \tilde\Psi^\veps_{1,low}\right]\ \right\|_{\cX^s}\nn\\
& \le\ C\ \left(\  \cO(\veps^\infty)\ +\ \veps^{3-2r}\| \chi(|k|\le\veps^r)\tilde\Psi^\veps_{1,low}\ \|_{\cX^s}\ +\ \veps^{5-2r}\ \| \chi(|k|\le\veps^r)\tilde\Psi^\veps_{1,low}\ \|^2_{\cX^s}\right.\nn\\
 &\left.\ +\ \veps^{8-2r}\ \| \chi(|k|\le\veps^r)\tilde\Psi^\veps_{1,low}\ \|^3_{\cX^s}\ \right).\ \ \nn\\
 &
 \label{Psi-highofPsi-low}
\end{align}
\end{proposition}

\begin{proof}
  Consider the system for $\tilde\Psi^\veps_{high}$, (\ref{high-eqns-a}-\ref{high-eqns-b}). The result follows from direct estimation using:
  \begin{align}
    &\left|\ E_j(k)-E_*+\veps^2\ \right|\ge c>0,\ \  j\ge2,\nn\\&
    \left|\ E_1(k)-E_*+\veps^2\ \right|\ge \veps^{2r},\ \ 
    \veps^r\le |k|\le 1/2\end{align}
  and applying the implicit function theorem. 
\end{proof}

 Use of the Proposition \ref{prop:Psi-highofPsi-low} in the rescaled low frequency equation we obtain
  \begin{proposition}\label{prop:close-low-freq}
    \begin{align}
      &\left(\ \frac{1}{2}E_1''(0)\kappa^2+1\right)\ \chi\left(|\kappa|\le\veps^{r-1}\right)\hat{\Phi}(\kappa)\ 
      -\ 3\ \G\ \left[F^2\ \chi\left(|\nabla_y|\le\veps^{r-1}\right)\Phi\ \right]\hat\ (\kappa)\ =\ \hat H^\veps\nn\\
      &{\rm where}\ \hat H^\veps=\ \chi\left(|\kappa|\le\veps^{r-1}\right)\hat H^\veps\ {\rm \ and\ satisfies\ the\ bound}\nn\\
      & \|\ \hat H^\veps\ \|_{L^{2,s}} \le \ \| \hat\cG_{{\rm sym}}\ \|_{ L^{2,s}}\  +\ \cO\left(\| \D^3E_1\|_\infty\ \veps^{3r}\ \|\chi\left(|\kappa|\le\veps^{r-1}\right)\hat\Phi\|_{L^{2,s}}\ \right)\nn\\
      & \ \ \ \ +\ \eps^\sigma\ \left(
        \|\ \chi\left(|\kappa|\le\veps^{r-1}\right)\ \hat\Phi\|_{L^{2,s}}\ +\ \|\ \chi\left(|\kappa|\le\veps^{r-1}\right)\ \hat\Phi\|_{L^{2,s}}^3\ \right)
      \label{closed-Rescaled-low-eqn}
 \end{align} 
 where $\sigma>0$ and $\cG_{sym}\in H^s_{sym}$.\\ \\
  Equivalently, we have
 \begin{align}
&\left(\ - \D_y\ A\ \D_y-\ 3\ \G F^2\ \ +1\ \right)\ \chi\left(|\nabla_y|\le\veps^{r-1}\right)\Phi\ 
  =\ \chi\left(|\nabla_y|\le\veps^{r-1}\right)\ H^\veps,\nn\\
& \label{closedPhi}\end{align}
where $A=\frac{1}{2}E_1''(0)$ and
\begin{align}
& \|\ H^\veps\ \|_{H^s} \le \ \| \cG_{{\rm sym}}\ \|_{H^s}\  +\ \cO\left(\| \D^3E_1\|_\infty\ \veps^{3r}\ \|\chi\left(|\nabla_y|\le\veps^{r-1}\right)\Phi\|_{H^s}\ \right)\nn\\
& \ \ \ \ +\ \eps^\sigma\ \left(
 \|\ \chi\left(|\nabla_y|\le\veps^{r-1}\right)\ \Phi\|_{H^s}\ +\ \|\ \chi\left(|\nabla_y|\le\veps^{r-1}\right)\ \Phi\|_{H^s}^3\ \right)
 \nn\end{align}
 \end{proposition}
\nit We now complete the proof. Denote by $L_+^A$ the operator
 \begin{align}
&L_+^A\ \equiv\ -\ \ \D_y\ A\ \D_y -3\G\ F^2(y)\ +1 \label{L+def}
\end{align}
and 
\begin{align}
&\chi_\veps\ =\  \chi\left(|\nabla_y|\le\veps^{r-1}\right),\ \ \ 
\overline{\chi_\veps}\ =\  1-\chi_\veps\ = \chi\left(|\nabla_y|\ge\veps^{r-1}\right)\label{chi-veps}
\end{align}
We recall (see (\ref{r-restrict})\ )
\begin{equation}
0<r<1\ .\nn\end{equation}
Equation (\ref{closedPhi}) for $\Phi$ can be written as
\begin{equation}
\chi_\veps\ L_+^A\ \chi_\veps\ \Phi\ =\ \chi_\veps H^\veps[\Phi]
\label{symPhieqn}\end{equation}
Since $F$ is chosen to be centered at local extremum of the {\it symmetric potential}, $V(\bx)$, we have that the mapping 
\begin{equation}
\Phi\mapsto \chi_\veps\ H^\veps[\Phi]
\nn\end{equation}
maps $H^s_{even}$ to itself. 

We claim that for some $\veps_0$, if $\veps\le\veps_0$, then the operator $\chi_\veps\ L_+^A\ \chi_\veps\ :\ H_{sym}^{s+2}\to H_{sym}^s$ has an inverse with norm bound which depends only on $\veps_0$. Thus, for $0\le\veps<\veps_1$, we can reformulate
\eqref{symPhieqn} as
\begin{equation}
\Phi\ =\ \left(\ \chi_\veps\ L_+^A\ \chi_\veps\ \right)^{-1}\ \chi_\veps \ H^\veps[\Phi]\nn\end{equation}
and show by fixed point iteration that for some $0<\veps_1\le\veps_0$, sufficiently small, equation (\ref{symPhieqn}) has a unique $H^{s+2}$ solution, which is bounded uniformly for $\veps\le\veps_1$. This then implies Theorem \ref{theo:2scale}

Therefore, the proof boils down to establishing the invertibility of 
 $\chi_\veps\ L_+^A\ \chi_\veps\ :\ H_{sym}^{s+2}\to H_{sym}^s$.
  We first prove that $ L_+^A\  :\ H_{sym}^{s+2}\to H_{sym}^s$
  has a bounded inverse.
  
Now the the operator $L_+^A$, acting in $L^2(\R^1)$ has one-dimensional kernel, spanned by the function $\D_yF$. To see this,
differentiate the equation for $F(y)$:
\begin{equation}
-\ \D_y\ A\ \D_y F\ +\ F\ -\G\ F^3=0
\nn\end{equation}
and obtain $L_+^AF'=0$. Moreover, ${\rm Ker}(L_+^A)={\rm span}\{F'\}$, since the eigenvalues of a Sturm-Liouville operator are simple. 
 Since $F'$ is odd, $L_+^A$ is an invertible and bounded map from  $H^{s+2}_{even}(\R^1)$ to $H^{s}_{even}(\R^1)$.
 
 Finally, we turn to the invertibility of $\chi_\veps\ L_+^A\ \chi_\veps\ :\ H_{sym}^{s+2}\to H_{sym}^s$ for $\veps$ sufficiently small. We begin by expressing $\chi_\veps\ L_+^A\ \chi_\veps$ as a perturbation of $L_+^A$:
 \begin{align}
 &\chi_\veps\ L_+^A\ \chi_\veps\ =\ L_+^A\ +\ Q_\veps\nn\\
 &Q_\veps\ =\ -(\overline{\chi_\veps}L^A_+\ +\ L^A_+\overline{\chi_\veps})\ +\ 
  \overline{\chi_\veps}\ L^A_+\ \overline{\chi_\veps}
  \nn\end{align}
  Therefore it suffices to prove that
  \begin{equation}
  L_+^A+Q_\veps\ =\ L_+^A\ \left(I\ +\ (L_+^A)^{-1}\ Q_\veps\ \right)
  \nn\end{equation}
  has a bounded inverse defined on $H_{sym}^s$. A bounded inverse
  \begin{equation}
  \left(\ L_+^A+Q_\veps\ \right)^{-1}\ =\ 
   \left(I\ +\ (L_+^A)^{-1}\ Q_\veps\ \right)^{-1}\ \left(\ L_+^A\ \right)^{-1}
   \nn\end{equation}
  exists provided the norm of $(L_+^A)^{-1}\ Q_\veps$ can be made smaller than one, by choosing $\veps$ sufficiently small. 
    \begin{equation}
  (L_+^A)^{-1}\ Q_\veps\ =\ -(L_+^A)^{-1}\ \overline{\chi_\veps}\ L^A_+\chi_\eps
\   -\  \overline{\chi_\veps}
\nn\end{equation}
Concerning the second term, the mapping $f\mapsto\overline{\chi_\veps}f$  maps $H^s$ to $H^s$. If $s>0$, the operator norm tends to zero as $\veps\to0$, by explicit  calculation using the Fourier transform.

 Finally, consider the mapping $f\mapsto (L_+^A)^{-1}\ \overline{\chi_\veps}\ L^A_+f$.  We prove that this mapping is bounded from $H^k$
  to $H^{k-\delta}$, for any $\delta>0$. We see this as follows. Denote by $\langle y\rangle=(1+|y|^2)^{1\over2}$ and therefore the operator $\langle D\rangle^a$ is defined by 
  \begin{equation}
  \langle D\rangle^a\ f\ =\ \int e^{ik\cdot x}\ \langle k\rangle^a\ \hat{f}(k)\ dk
  \nn\end{equation}
  Now, for any $a>0$, we write
  \begin{align}
  &  (L_+^A)^{-1}\ \overline{\chi_\veps}\ L^A_+\ =\ (L_+^A)^{-1}\ \langle D\rangle^a\ \cdot\ 
    \langle D\rangle^{-a}\ \overline{\chi}_\veps\ \cdot\  L^A_+\
   \nn\end{align}
   and estimate the norm as follows:
   \begin{align}
& \|(L_+^A)^{-1}\ \overline{\chi_\veps}\ L^A_+\|_{H^{s-a}\leftarrow H^s}\nn\\
& \le \| (L_+^A)^{-1}\ \langle D\rangle^a \|_{H^{s-a}\leftarrow H^{s-2}}\
   \|  \langle D\rangle^{-a}\ 
   \overline{\chi}_\veps \|_{H^{s-2}\leftarrow H^{s-2} }\cdot 
   \| L^A_+\|_{H^{s-2}\leftarrow H^s}
   \nn\end{align}
   Note that the first and third factors are 
    bounded independently of $\veps$. We claim that
     $\|  \langle D\rangle^{-a}\ 
   \overline{\chi}_\veps \|_{H^{s-2}\leftarrow H^{s-2} }\to0$ as $\veps\to0$. To see this, calculate as follows:
    \begin{align}
&   \| \langle D\rangle^{-a}\ 
   \overline{\chi}_\veps\ f\|_{H^\tau}^2\ =\ 
   \int  \langle \kappa \rangle^{-2a} 1_{\{|\kappa|\le\veps^{r-1}\}}\ 
     \langle \kappa \rangle^\tau\ |\hat{f}(\kappa)|^2\ d\kappa \nn\\
 & \le\ \veps^{2a(1-r)}\ \| f\|_{H^\tau}^2
 \nn\end{align}
 Thus for any $a>0$, we have  $\|  \langle D\rangle^{-a}\ 
   \overline{\chi}_\veps \|_{H^{\tau}\leftarrow H^{\tau} }\to0$
 This completes the proof of Theorem \ref{theo:2scale}.

  \begin{remark}\label{rmk:extensions}
   {\bf General spatial dimensions $d\ge1$:}\ 
 The proof given readily extends
  to general dimension $d\ge1$. 
   One works in spaces $H^s(\R^d), \ s>d/2$.
    $\cX^s$ is constructed taking into account the behavior of the dispersion functions, $E_n(k)$, in dimension $d$.
    As before, we choose $F(\by)=F\left(\eps(\bx-\bx_0)\right)$, with $\bx_0$ a point of symmetry of $V(\bx)$. The kernel of 
  $L_+^A=-\ \D_{y_i}A^{ij}\D_{y_j}+F-3F^2(y)$ has dimension $d$ and is generated by translations, {\it i.e.} ${\rm Kernel}(L_+^A)= {\rm span}\{\D_{y_j}F(y),\ j=1,\dots,d\ \}$ \cite{Weinstein-85,Kwong-89}. Since ${\rm Kernel}(L_+^A)$ is orthogonal to $H^2_{even}(\R^d)$,
   $L_+^A$ is invertible mapping from $H^{s+2}(\R^n)$ to $H^{s}$.
   \end{remark}


\section{Numerical computations in the semi-infinite gap}
\label{sec:numerics}

Our analytical results apply to solitons with frequencies in a
spectral gap, which are also sufficiently close to a spectral band edge.
In this section we present the results of numerical computations 
corroborating the rigorous asymptotic results  near 
the spectral band edge, but also illustrating their approximate validity further away from the band edge, well into the spectral gap. The details of the numerical methods are discussed in Section~\S\ref{sec:numeric_details}.

The particular rigorous asymptotic results we explore numerically in detail are:
\begin{enumerate}
\item the asymptotic structure of soliton's lying near the edge
of the spectral gap (Theorem \ref{theo:2scale}):
\begin{equation}
u(x,\mu) \approx (E_*-\mu)^{1\over\sigma}\ w(\bx)\ F\left(\veps(\bx-\bx_0)\right),
\nn\end{equation}
where $\bx_0$ is a local extremum of $V(\bx)$.
\item The asymptotic behavior of the soliton (nonlinear bound state ) power, $\cP(\mu)$, along minima- and maxima-centered solitons
 as $\mu$ approaches $E_*$; see Theorem \ref{theo:edgepower}  and  Corollary \ref{cor:edgepower-crit}. In particular, in the critical case $\sigma=2/d$ we have
 \begin{equation}
 \cP(\mu)\approx \zetap\ \cPc
 \label{crit-edge}
 \end{equation}
 \end{enumerate}

We focus on the one-dimensional  NLS/GP
 equation  \eqref{eq:NLS-V}
with critical nonlinearity and periodic potential governing $\psi(x,t)$
 and nonlinear bound states: $\psi(x,t)=e^{-i\mu t}u(x,\mu)$.\\ \\
 $d=1,\sigma=2$ :
\begin{eqnarray}
  \label{eq:nls1dcrit}
  i\D_t\psi &=& -\D_x^2\psi + V_0\cos^2(2\pi x)  \psi - |\psi|^4 \psi\\
  \mu u    &=& -\D_x^2 u+ V_0\cos^2(2\pi x) u- | u |^4 u
  \label{eq:bs1dcrit}
\end{eqnarray} 
$V_0$ is the variation or  contrast of the potential.

We have observed similar results to those presented below for the two-dimensional critical NLS/GP with periodic potential:
$d=2,\sigma=1$:
\begin{eqnarray}
  \label{eq:nls2dcrit}
  i \D_t\psi_t &=& - \left(\D_x^2 + \D_y^2\right)\psi 
  + \frac{V_0}{2}\left[\cos^2(2\pi x)+\cos^2(2\pi y)\right] \psi -  
  |\psi|^2 \psi~.  \\   
  \mu u &=& - \left(\D_x^2 + \D_y^2\right)u
  + \frac{V_0}{2}\left[\ \cos^2(2\pi x)+\cos^2(2\pi y)\ \right] u-  |u |^2 u~
  \label{eq:bs2dcrit}
\end{eqnarray}

Theorems \ref{theo:2scale} and \ref{theo:edgepower}, and Corollary \ref{cor:edgepower-crit} apply to \eqref{eq:bs1dcrit}
 with states centered at a minimum: $x_0=.25$ or maximum:  $x_0=0$. These results also apply to \eqref{eq:bs2dcrit}
  with states centered at a minimum: $\bx_0=(.25,.25)$,
  at a maximum: $\bx_0=(0,0)$, and at a saddle point: $\bx_0=(.25,0)$ or $\bx_0=(0,.25)$.
   

\subsection{Soliton profiles: asymptotic theory and computation}
\label{sec:profiles}

Figures~\ref{fig:transition} and~\ref{fig:profiles_1d}
display nonlinear bound state profiles of the one-dimensional 
NLS/GP equation \eqref{eq:NLS-V} for values of $\mu$ in the semi-infinite gap of the Schr\"odinger operator: 
$\D_x^2 + V_0\cos(K x)$, {\it i.e.}
$$
\mu\in (-\infty, E_*),\ \ \mu<E_*=E_*(V_0,K)
$$
both near and  far from the band edge.

Plots A1, A2 and A3 in Figure~\ref{fig:transition} display the case of solitons,
centered at local  minima of the potential
 with, from left to right, frequency $\mu$ approaching $E_*$, 
 at distances $E_*-\mu=10,\ 1$ and $0.01$, respectively. 
 Plots  B1, B2, and B3 in Figure~\ref{fig:transition} correspond
 to the case of solitons centered at local maxima of the potential.
\\

We first note that the figures show the expected trend toward increased localization 
 as $\mu<0$ is decreased. For $\mu$ large and negative
 the solitons centered at maximum or minima
  approach a scaled $V\equiv 0$ soliton; see \eqref{eq:R_1d}
   and 
  
 Our main analytical results apply to solitons whose frequencies lie near the band edge, although numerical studies  indicate their approximate validity some distance away from the band edge.

Theorem \ref{theo:2scale} implies that  nonlinear bound states are, to leading in order in the distance to the spectral band edge, 
 a product of a linear Bloch state with band edge energy
  and a soliton in an effective homogeneous medium; see
 \eqref{eq:u-eps} and \eqref{eq:R_1d}:
\begin{align}
  \label{eq:wF}
  & u(x,\mu) \approx w(x)\ F(y)\nn\\
  &F(y) = \left(\frac{E_*-\mu}{\G}\right)^{1/4}\ \sech^{1\over2}
   \left(\ 2\sqrt{m_*(E_*-\mu)}\ (x - x_0)\ \right).  
\end{align}
The centering of the soliton is $x_0$,  a point of symmetry of the potential, $V(\bx)$.

The maximum of the Bloch modes $w(x)$  
(normalized to be positive and with unit mass)
occurs at the minimum of the potential.
However, depending on the centering point $x_0$,
$F(y)$ has a maximum (minimum) at the potential 
maximum (minimum).
Thus, bound states centered on a potential minimum
are approximately a product of functions that peak
at the same values of $x$ yielding a more peaked
bound state; compare the top and bottom panels of
Fig.~\ref{fig:transition} with $ E_*-\mu =10$.

Figure~\ref{fig:transition} (A2,B2) shows that for  $E_*-\mu=1$ 
the bound states have discernible oscillations about a positive envelope, reflecting the solutions leading order behavior \eqref{eq:wF}. These oscillations can be understood as a result of the 
``underlying'' Bloch modes in Eq.~(\ref{eq:wF}).
Here as above, the asymptotic theory appears to capture the 
structure of the bound states even when $\mu$ is not very close
to the band edge. 
\\

We note as well, for the soliton centered at the potential's local maximum,  a transition in the profile from {\it  single-humped} to
{\it double-humped } (having a dimple at $x=0$) as $\mu$ decreases through $\mu=\mu_\#$, the value at which 
$\cP[u(\cdot,\mu)]$, {\it along the branch of solitons, centered at a local maximum of $V$}, achieves its maximum; 
see figure~\ref{fig:power_mu_1d}. A related observation is made in  \cite{AISP:06}.
\\ 

Comparing Figs.~\ref{fig:transition}\ (A3) and ~\ref{fig:transition}\ (B3)
shows that near the band edge ($E_*-\mu=0.01$)
there is almost no visible difference between
the bound states centered at potential minima and those centered
at potential maxima. This is clear
from Eq.~(\ref{eq:wF}), since in this regime $F(y)\sim$ decays only  on a length scale much larger than the period of $V(x)$
 and thus, for both maxima- and minima- centered solitons, 
  $u(x)\sim$ constant\ $\times w(x)$.
\\

Figure~\ref{fig:profiles_1d} shows a direct comparison between
the asymptotic theory, i.e., the leading order solution
Eq.~(\ref{eq:u-eps}) near the band edge, and the ``actual''
bound state profiles, computed by  solving the 
bound state differential equation \eqref{eq:bs1dcrit} with high accuracy.

\fig{0.75}{transition}{transition}{ 
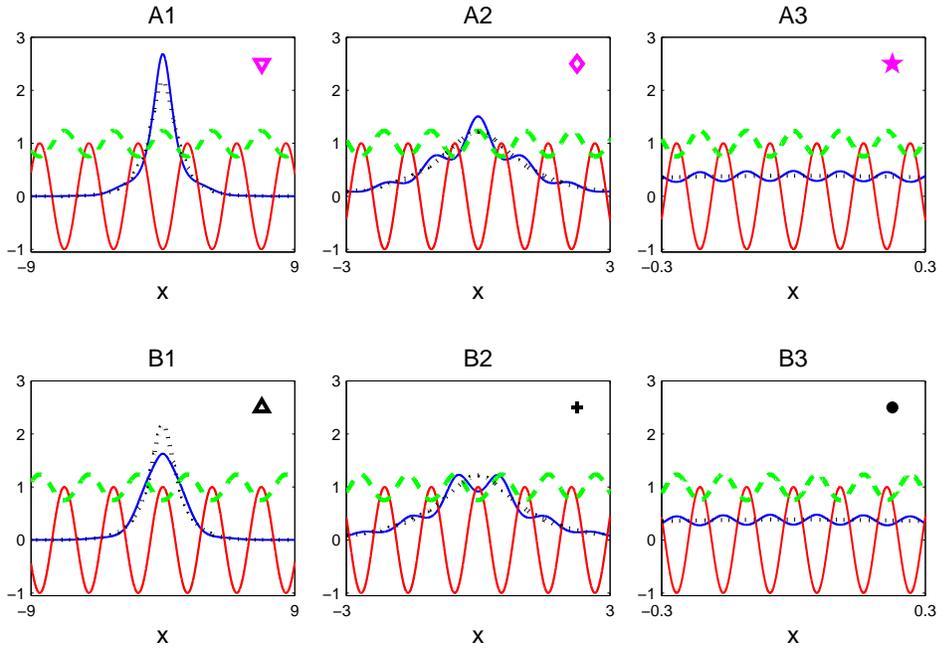
Leading order asymptotic profiles obtained via Theorem \ref{theo:2scale} (blue, solid) centered on the potential minimum
(top panel) and potential maximum (bottom panel) for
(A1, B1) $E_*-\mu=10$ , (A2, B2) $E_*-\mu=1$,
and (A3, B3) $E_*-\mu=0.01$.
Also shown are the  scaled (for plotting purposes) potential $V(x)/V_0$ (red, solid), 
Bloch wave $w(x)$ (green, dashes),  
and rescaled homogeneous ground state  
$F(y)=F\left(\eps(x-x_0)\right)$ [Eq.~(\ref{eq:wF}), black, dots)].
For clarity only a small portion of the domain is shown and 
the $x$-axes are zoomed in as $E_*-\mu$ decreases.
Geometric shapes correspond to those depicted in
Figs.~\ref{fig:profiles_1d} and~\ref{fig:power_mu_1d}. 
}

\fig{0.75}{profiles_1d}{profiles_1d}{
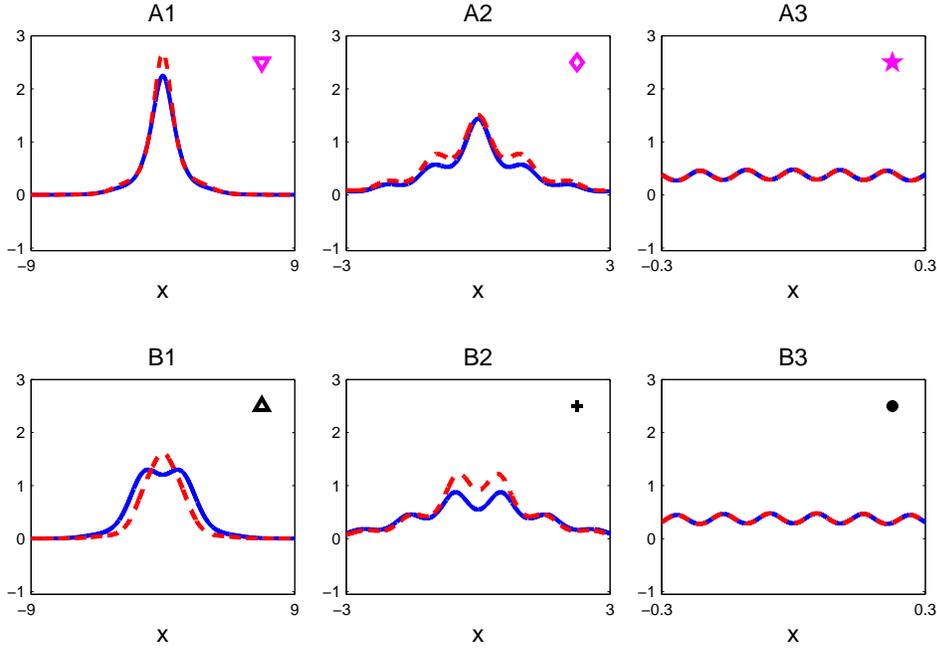
Bound-state profiles computed using the Renormalization method
[Eq.~(\ref{eq:u-nd}), blue, solid] compared with the leading order  asymptotic 
theory [Eq.~(\ref{eq:u-eps}), red, dashes] 
for the same parameters as in Fig.~\ref{fig:transition}.
}


\subsection{Effective mass and the power curve $\mu\mapsto\cP[u(\cdot,\mu)]$}
\label{sec:power_curve}

Each plot in Figure~\ref{fig:power_mu_1d}  shows two curves of 
bound-state power, $\cP[u(\cdot,\mu)]$, in the semi-infinite gap as a function of $\mu$ 
for  \eqref{eq:bs1dcrit}
 with $V_0=10, K=2\pi$. The solid (blue) curve corresponds to the variation of $\cP[u(\cdot,\mu)]$ for the family of solitons centered at the potential's local minimum and the dashed (red) curve for the family centered at  the potential's local maximum.

We observe the following:
\begin{enumerate}
\item
  Panel (A1) of Figure~\ref{fig:power_mu_1d} shows the variation of $\cP$
  over a wide range of $\mu$     in the semi-infinite gap. 
  As $\mu\to -\infty$, i.e. in the semi-classical limit
  the power approaches the homogeneous ($V\equiv 0$) power of the
  ground state,  $\cPc=\frac{\sqrt{3}}{2}\pi \approx 2.72\,$.
\item
 Panels (A2) and (A3) of  Figure~\ref{fig:power_mu_1d} 
 show, as predicted by Corollary \ref{cor:edgepower-crit},  that  
  as $\mu\to E_*$ the $\cP[u(\cdot,\mu)]$  approaches the value $\zetap\cPc \approx 2.2\,$ strictly less than $\cPc\approx 2.72$. Here,  $\zetap\approx 2.2/2.72\approx 0.8\,$.
   This is true for solitons centered at either  potential minima or potential maxima; see [Fig.~\ref{fig:power_mu_1d} (A3)].
\item Although the asymptotic behavior of the $\cP$ for maxima and minima- centered bound states is the same, 
  across most of the gap bound states centered on lattice minima
  (resp. maxima) have   power below (respectively,  above) $\cPc\approx 2.72$.
\item
 Panels (A1) and (A2) of Figures~\ref{fig:power_mu_1d} show a
 transition in the slope of the minima- and maxima- centered power
 curves at the {\it same} value $\mu\equiv \mu_\#$. 
As discussed in Section~\S\ref{sec:stability}, the transition in slope of $\mu\mapsto\cP[\mu]$ along the power curve for minima-centered solitons signals a transition from the unstable (positive slope) to stable (negative slope) regime.
 Maxima centered solitons, as discussed, are unstable and the  transition in slope signals a change in the number of unstable eigenvalues of the linearized problem \cite{Jones:88,Grillakis:90}.
\end{enumerate}

\fig{0.75}{power_mu_1d}{power_mu_1d}{
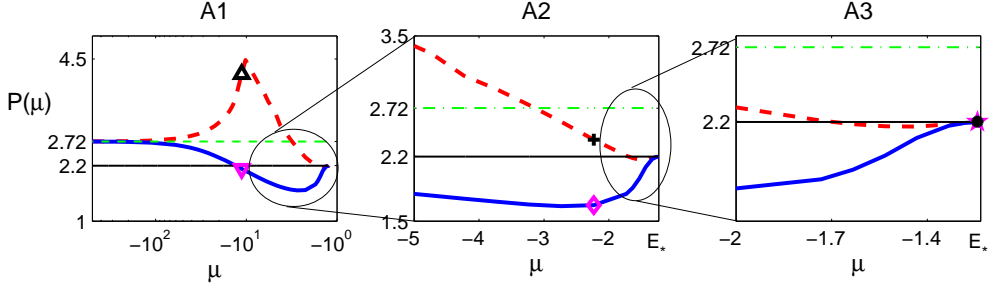
Power$-\mu$ plot for Eq.~(\ref{eq:bs1dcrit}) with $V_0=10$ and $K=2\pi$
for bound states centered on a maximum (red, dashes) and minimum (blue, solid) of the potential.
(A1): Wide view of the semi-infinite gap (semi-log $\mu$-axis).
(A2) and (A3): Zooming in near the band edge.
The asymptotically computed value at the  band edge, 
$\Pe=\zetap\times\cPc$ [Eq.~(\ref{eq:power_edge}), black, solid line)
and critical power $\cPc$ for the homogeneous (translation
invariant)  equation (dots, black) are shown as well. 
Geometric shapes correspond to the cases whose bound states are
depicted in Fig.~\ref{fig:transition}. 
}


\subsection{Numerical methods}
\label{sec:numeric_details}

The computation of the asymptotic bound states and the 
``actual'' bound states are carried out using Matlab and Octave.

{\bf Computation of the Bloch mode, inverse effective mass and $\zetap$}.
The Bloch mode at the band edge, $w$, is computed 
using an eigenvalue solver within a single lattice cell
(see \cite{Sivan-NL-08}[Appendix] on using Matlab's eigenvalue solver).
For convenience we normalize the Bloch mode to have unit mass, i.e. 
$\int w^2=1\,$.
The inverse effective mass tensor $A^{ij}$ is computed
by employing Matlab's linear system solver for $L_*^{-1}$. 
It is then straightforward to compute 
the inverse effective mass (curvature) $\minv$, coupling constant
$\G$, and the band-edge power factor $\zetap$ using 
Eqs.~(\ref{eq:m*}) and (\ref{eq:zeta*}).

{\bf Computation of the bound state at the band edge}.
The asymptotic bound state is comprised of the Bloch mode
and the rescaled homogeneous solution. The Bloch mode is
obtained by periodically extending $w$ from one lattice cell
to the domain over which the bound state is computed --
typically several hundred lattice cells.
The rescaled homogeneous ground state, $F$, and its power  $\cPc$
are computed in 1D using the explicit solution, i.e. Eq.~(\ref{eq:R_1d} )
with the rescaling in Eq.~(\ref{eq:wF}). Finally, the asymptotic bound state is obtained by shifting 
$F$ to be centered at point of symmetry of the potential
and taking its product with the Bloch mode.

{\bf Computation of the ``actual'' bound states}.
The  bound states of Eq.~(\ref{eq:u-nd})  are computed using 
Renormalization method~\cite{Ablowitz-Musslimani:05}.
This method is based on fixed-point iterations coupled to
an algebraic condition, whose role is to constraint the solution
to a suitable integral identity consistent with the bound-state
(otherwise, the iterative solution would diverge).
The convergence is monitored by the $L_\infty$ norm
of successive iterations and by relative change of the Renormalization constant.
For example, for a 1D computation with $|\Omega|=0.01$ 
the domain size is a few hundred lattice cells.
We use $2^{16}$ grid points to well-resolve the oscillations on the 
scale of the potential period. 
The computation of the bound state is considered to have converged 
when the difference between successive iterations satisfies
$\|u^{n+1}(\bx) - u^n(\bx) \|_\infty < 10^{-8}$. This typically happens 
within fewer than 100 iterations (a few minutes).

The Renormalization method needs to be seeded with an initial guess.
Deep inside the gap the Renormalization method converges when 
seeded by a Gaussian (or sech) profile. 
On the other hand, near the band edge the method diverges when seeded by 
a Gaussian or sech, which are apparently too far from the
basin of attraction of the bound state.
We overcome this difficulty by seeding the Renormalization method 
with the asymptotic solution. 


\section{Summary and discussion}\label{sec:summary}

In this paper we have studied the bifurcation of small  amplitude
 ($H^s(\R^d), s>d/2$)  nonlinear bound states (solitary waves or ``solitons'') of the nonlinear Schr\"odinger  / Gross-Pitaevskii equation with a periodic and symmetric potential.  Our results provide insight into  questions (Q1-Q3) of the introduction.
  We now briefly summarize our results, with reference to  (Q1-Q3).\bigskip
 
Concerning (Q1):
  \begin{enumerate}
 \item A family of bifurcating solitons (spatially localized standing wave states) can be constructed centered at any point of symmetry, $\bx_0$, of $V(\bx)$.
 \item Solitons with frequencies near a spectral band edge have a two-scale structure: $u_\veps(\bx)\approx \veps^{1\over\sigma}\ F\left(\veps(\bx-\bx_0)\right)\ w(\bx)$,  where $\veps^2=|E_*-\mu|$ is the distance of the frequency to the spectral band edge.
 \end{enumerate}
 
 Concerning (Q2):
 \begin{enumerate}
 \item We prove, in general, that the limit of the soliton power, along any family of solitons centered at a point of symmetry of $V(\bx)$, is strictly less than the power of the Townes soliton: 
  \begin{equation}
  \lim_{\mu\to E_*} \cP[u(\cdot,\mu)]\ =\ \zetap\ \cP_{cr}\ <\ \cP_{cr}.
  \label{edge-lim}
  \end{equation}
Note: This limit is independent of the centering of the soliton, $\bx_0$.
  \item We prove a high order expansion, which is necessary to capture information about the slope of the curve, $\mu\mapsto \cP[u(\cdot,\mu)]$, near the band edge. Encoded in the slope of this curve is information on nonlinear dynamic stability. We conjecture that for critical nonlinearities ($\sigma=2/d$), the curve has positive slope near the band edge therefore solitons with frequencies near the band edge are unstable. We have verified this analytically for low contrast potentials and numerically for a range of potentials, without a smallness constraint on the contrast.
  \item Our analytical results concerning the multiple scale structure of solitons of NLS / GP and the curve $\mu\mapsto \cP[u(\cdot,\mu)]$ are corroborated through careful numerical experiments.
    \end{enumerate}
    
  Concerning (Q3), see Remark \ref{rmk:sol-exc-conj}. 
 In particular,  see figure \ref{fig:power_mu_1d_A1_v2} and the {\it soliton excitation threshold conjecture}.\bigskip
    
  We conclude this section with a  discussion the emergent parameter, $\zetap$, appearing  in \eqref{edge-lim}.
From equations ~(\ref{eq:geff}) and~(\ref{eq:zeta*}) we have
\begin{equation}
     \zetap  \ =\ \left( \frac{1}{m_*}\right)^{\frac{1}{2}}
     \left(\frac{\left(\dashint_{\cell} w^2\right)^{\sigma+1}}{\dashint_{\cell} w^{2\sigma+2}}
     \right)^{\frac{1}{\sigma}}  =\ \left( \frac{1}{m_*}\right)^{\frac{1}{2}}\ \left(\frac{1}{\G}\right)^{1\over\sigma}\ \frac{1}{{\rm vol}({\cal B})}\ .
\label{zeta-num}
\end{equation}
Here $m_*$ denotes the determinant of the {\it effective mass} tensor, $w(\bx)$ the $\cell$- periodic Bloch (band edge) state, $\G$ the effective nonlinear coupling and ${\rm vol}(\cell)$, the volume of the fundamental periodic cell.\\ 

 A matter of  practical / experimental interest is that the parameters $\G$ and $m_*$  are tunable via appropriate design of periodic structure, $V(\bx)$,  therefore making it possible to manipulate the power curve, $\cP$ vs. $\mu$.  
Figure~\ref{fig:mass_zeta}  displays  $\sqminv$, 
and  $\zetap$ as functions of the potential contrast $V_0$ in 1D and 2D.
All three quantities are bounded between 0 and 1 and decrease
monotonically with $V_0$. In particular, this means that
$\Pe/\cPc$ decreases with $V_0$. This decrease is ``faster''
in 1D than in 2D, at least for $V_0<40$.

\fig{0.65}{mass_zeta}{mass_zeta}{
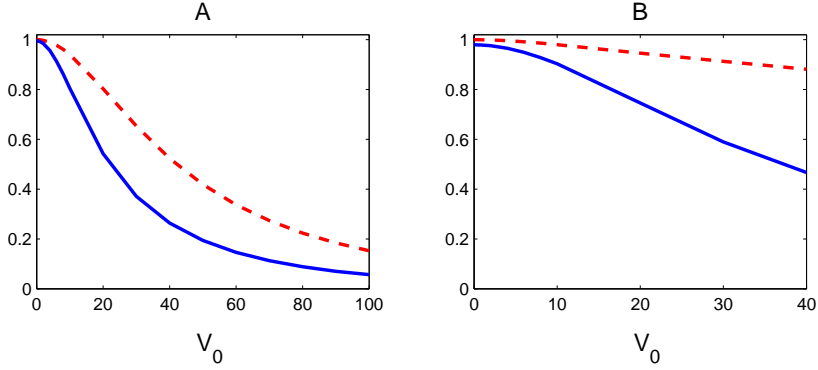
Power factor $\zetap$ (blue, solid) and $\sqminv$ (red, dashes),
as functions of the potential contrast $V_0$ in 
(A) 1D with $\sigma=2$ [Eq.~(\ref{eq:bs1dcrit})  with $K=2\pi$]
and (B) 2D with $\sigma=1$ [Eq.~(\ref{eq:bs2dcrit}) with $K_x=K_y=2\pi$].
}


\newpage


\section*{Acknowledgments}
MIW was supported in part by US -National Science Foundation  Grants
DMS-04-12305 and DMS-07-07850. The authors thank 
Y. Sivan for discussions concerning this work and for comments on the manuscript.
We would like to thank B. Altschuler, T. Dohnal, M. Hoefer, A. Millis,
P. Kuchment and  C.\ W. Wong for informative discussions.
{\bf We also thank the reviewers  for a very careful reading of this work
and detailed comments.}


\appendix

\section{Effective mass tensor}
\label{ap:E''-nd}

In this section we prove equation  \eqref{eq:Aij}, which relates
 $D^2E_1(0)$, the Hessian matrix of the band dispersion function $E_1$ to the matrix 
$\aij$ arising in the multiple scale analysis

Denote by  $e^{i \bk\cdot\bx}\phi(\bx;\bk)$ the Bloch state, 
associated with $E_1(\bk):\bk\in\cell^*\to \R $;
see Section~\S\ref{sec:FBtheory}. Here,  $\phi(\bx;\bk)$  is periodic. 
Thus, 
\begin{eqnarray}
  \label{eq:phi-nd}
   & \left(-\Delta - 2 i \bk\cdot\nabla + |\bk|^2 + V\right) \phi = E_1(\bk) \phi~\\
   & \phi(\bx+\bq_j;\bk)=\phi(\bx;\bk),\ \ j=1,\dots,d\nn
\end{eqnarray}

At the bandgap edge one has
\begin{equation}
  \label{eq:k=0-nd}
  E_1(0) = E_*~, \quad \phi(\bx;0)\equiv w(\bx)~,
\end{equation}
where $w(\bx)$ is the ground state of $-\Delta+V$, subject to 
periodic boundary conditions on $\cell$. 

We denote $f_{,k_j}\equiv \D_{k_j}f$ and $f_{,r}\equiv \D_{x_r}f$. 
Differentiation of Eq.~(\ref{eq:phi-nd}) with respect to $k_i$ gives
\begin{eqnarray}
  \label{eq:phi_k-nd}
  \lefteqn{
    \left(-\Delta - 2 i k_i \D_{x_i} + |\bk|^2 + V\right) \phi_{k_i} =
    } \\*[2mm] &&
  E_{1,k_i} \phi + E_1\phi_{,k_i} + 2 i \phi_{,x_i} - 2k_i \phi~,
  \nonumber
\end{eqnarray}
At $\bk=0$, Eqs.(\ref{eq:k=0-nd}) and~(\ref{eq:phi_k-nd})  yield
\begin{equation}
  \label{eq:phi_k0_1-nd}
  L_* \phi_{,k_i}(\bx;0) = E_{*,k_i} w + 2 i w_{,i}~.
\end{equation}
Removing secular terms from Eq.~(\ref{eq:phi_k0-nd}) leads to
\begin{equation}
  \label{eq:E'-nd}  
  E_{*,k_i}=0.
\end{equation}
It follows from Eqs.~(\ref{eq:phi_k0_1-nd}) and~(\ref{eq:E'-nd}) that 
\begin{equation}
  \label{eq:phi_k0-nd}
   \phi_{,k_i}(\bx;0) = 2 i L_*^{-1}w_{,i}~.
\end{equation}
Differentiating Eq.~(\ref{eq:phi_k-nd}) with respect to $k_j$, setting
$\bk=0$ and using Eq.~(\ref{eq:k=0-nd}) we arrive at
\begin{eqnarray*}
\lefteqn{
L_*  \phi_{,k_ik_j}(\bx;0) = 
\left(E_{*,k_ik_j}-2\delta_{ij}\right) w 
+ \big [E_{*,k_i}\phi_{,k_j}(\bx;0) + 
} \\*[2mm] &&
E_{*,k_j}\phi_{,k_i}(\bx;0) \big]
+ 2 i \left [\phi_{,x_ik_j}(\bx;0)+\phi_{,x_jk_i}(\bx;0)\right]~.  
\end{eqnarray*}
Using Eqs.~(\ref{eq:E'-nd}) and~(\ref{eq:phi_k0-nd}) gives
$$    
L_* \phi_{,k_ik_j}(\bx;0) = 
\left(E_{*,k_ik_j}-2\delta_{ij}\right) w 
- 4\left( \D_{x_i}L_*^{-1}\D_{x_j}w +  \D_{x_j}L_*^{-1}\D_{x_i}w\right)~.
$$
Removing the secular growth requires that the inner product of the 
with $w(\bx)$ vanish, \ie  
\begin{eqnarray*}
\lefteqn{
   \left(E_{*,k_ik_j}-2\delta_{ij}\right) \avg{w,w} 
- 4\avg{\D_{x_i}L_*^{-1}\D_{x_j}w,w} 
} \\*[2mm] &&
- 4\avg{\D_{x_j}L_*^{-1}\D_{x_i}w,w} 
= 0~.
\end{eqnarray*}
Using the fact that $L_*^{-1}$ is self-adjoint, the last two terms are
equal to each other. Therefore, using integration by parts leads to
$$ 
\left.\frac{1}{2}\ \frac{\D^2E_1}{\D k_i \D k_j} \right|_{\bk=0} \ =\
\ \delta_{ij}\  -\ 4\frac{\avg{L_*^{-1}\D_{x_j}w,\D_{x_i}w}}{\avg{w,w} }
\ \equiv \  \aij.
$$
This proves the relation (\ref{eq:Aij}).
$\square$


\section{Bound on determinant of effective mass tensor}
\label{ap:D(0)}

\begin{proposition}
$m_*^{-1}\ \equiv\ \det\left(\ 2^{-1}\ E_{1,k_ik_j}(0)\ \right)\ \le\ 1$, with  $m_*=1$ only if $V(\bx)$ is constant.
\end{proposition}

\nit For the proof we use  $m_*>0$; see \cite{Kirsch-Simon:87}.
Recall that
\begin{equation}
  \label{eq:B}
  \frac{1}{2}E_{*,k_ik_j} = \ \delta_{ij} - B_{ij},  
  \quad 
  B_{ij} \equiv \dfrac{4\avg{L_*^{-1}w_{,j},w_{,i}}}{\avg{w,w}}~.  
\end{equation}
We claim that
  $B_{ij}$ is positive definite.
  To see this, first recall that $L_*\ge0$ with one dimensional $L^2(\T^d)$ kernel  spanned by $w$. Clearly, $w\perp M\equiv span\left\{w_{,i}\ :\ i=1,\dots, d\right\}$ and therefore $B$ is well-defined.

Let ${\bf v}=(v_1,\dots,v_d)\in\R^d$ be arbitrary.
Then
\begin{equation}
{\bf v} \cdot B {\bf v} =  \langle\ L_*^{-1} {\bf v}\cdot\nabla{w} , {\bf v}\cdot\nabla{w}\ \rangle\ \ge  \lambda_2^{-1}\ \|\ {\bf v}\cdot\nabla{w}\ \|^2\ge 
\ C\  \|{\bf v}\|^2
\nn\end{equation}
where $\lambda_2>0$ denotes the second eigenvalue of $L_*$ acting
on $L^2(\T^d)$.

The matrix $E_{*,k_ik_j}$ is positive definite \cite{Kirsch-Simon:87}.
Therefore, $E_{*,k_ik_j}$ can be diagonalized by a unitary transformation
$p_{ij}$ such that 
$$
 p_{ri} E_{*,k_ik_j} p_{lj} = 2\left(1 - \beta_r\right)\delta_{rl}~,
$$
where $\lambda_i~(i=1\dots d)$ are the eigenvalues of $B_{ij}$.
It follows that
\begin{equation}
  \label{eq:det}
  \minv = \det\left(\frac{E_{*,k_ik_j}}{2}\right) 
  =\ \Pi_{i=1}^d\left(1-\beta_i\right),  
\end{equation}
where $\beta_i>0,\ i=1,\dots, d$.
In order to bound $\minv$ from above we will show\\
$ \beta_i\in(0,1),\ \ i=1,\dots, d$ and therefore $m_*^{-1}\le 1$.

We argue by continuity. 
Consider the one-parameter family of potentials $V(\bx;\theta)\equiv \theta V(\bx)$,
where $\theta\in [0,1]$ and the associated self-adjoint operator $L_*^\theta$ and matrix $B_{ij}(\theta)$.
Since $L_*^\theta$ is self-adjoint and, $w(\bx;\theta)$, its ground state   is simple, there are $d$ continuous functions $\theta\to\beta_i(\theta),\ \  i=1,\dots, d$, defining the eigenvalues of $B_{ij}(\theta)$. For $\theta=0$  (homogeneous medium), 
 $E_*=0$, $E_1(\bk) =\bk^2$, and  $E_{*,k_ik_j}=2\delta_{ij}$. 
Therefore $B_{ij}=0$ and $ \beta_j(0)=0,\ \ j=1,\dots, d$.
In this case (and only in this case!) $\minv=\det{I}=1$.
Next consider $\theta=1$, \ie  the original problem.
We claim that  $\beta_i(1)<1,\ i=1,\dots, d$. Otherwise, at some value of $\theta=\theta_*>0$ an eigenvalue of $B(\theta_*)$ would attain the value  one. This would contradict the positive definiteness
of $E_{*,k_ik_j}^\theta$.


\section{Effective mass for $d=1$ and the Floquet-Hill discriminant}
\label{ap:FB}
In one space dimension the endpoints of the spectral bands are obtained by studying the periodic and anti-periodic eigenvalue problems \cite{Eastham:73}. Very briefly, for each $E$ one constructs a $2\times2$ fundamental solution matrix, $M(x;E)$,
and considers the values of $E$ for which $M(q;E)$ has an eigenvalue $+1$ or $-1$, corresponding to periodic or antiperiodic eigenvalues. This is equivalent to $\Delta(E)=\pm2$, where 
\begin{equation}
\Delta(E)\equiv {\rm trace}\left(M(q;E\right)\ \equiv\ 2\cos [k(E)\perx]~.
\label{f-discr}
\end{equation}
 is the Floquet discriminant.

The band edge,   $E=E_*$, corresponds to $k=0$, at which we have  $\Delta(E_*)=2$.
Expanding Eq.~(\ref{f-discr}) in Taylor series around $k=0$ and $E=E_*$ gives
\begin{equation}
 \Delta(E_*) + \Delta'(E_*)(E-E_*) + \cO[(E-E_*)^2] 
= 2\left(1 - \frac{k^2\perx^2}{2}\right)  + O(k^4)~.
\nn\end{equation}
Using $\Delta(E_*)=2$ and solving for the second term on the LHS gives 
to leading order
\begin{equation}
 E - E_* = -\dfrac{k^2\perx^2}{\Delta'(E_*)}\ +\ \cO(k^4)
\nn\end{equation}
which yields the relation
\begin{equation}
  m_*^{-1} =E^{''}_1(0) = -\dfrac{2\perx^2}{\Delta'(E_*)}~.
\nn\end{equation}
Since $\Delta'(E_*)<0$, $m_*>0$. More generally, we have that $E_j''(0)>0$ at the {\it left} edge of each band and $E_j''(0)<0$ at the {\it right} edge of each band.
$\square$


\section{Power and slope for small potentials}\label{sec:smallpotentials}

In this section we use a regular perturbation expansion to
derive  the power and slope constants near the band edge, i.e.  $\zetap$ and
$\zetas$,  assuming a small potential. Such an expansion can be made rigorous by an argument based on the implicit function theorem.
The derivation is comprised of preliminary calculations in any dimension
of the Bloch function , an inverse linear operator,  and the inverse
effective mass tensor and coupling constant. 
To simplify notation, subsequent calculations are carried out 
explicitly in the critical case $(d=1,  \sigma=2)$.

\begin{remark}
  In the derivation below $\delta$ is assumed to be a small constant 
  independently of the of $\eps$. The calculations are done to order 
  $\cO(\eps^2\delta^m)$ for suitable $m$.
  For convenience, the $\eps^2$ is suppressed from $\cO(\cdot)$.
\end{remark}

Let $V(\bx)=\delta V_1(x)$, where $|\delta|\ll 1$. 
W.l.o.g. we assume that
\mbox{$\avg{V_1} \equiv \int_\cell V_1(\bx)\,d\bx  = 0~.$}
Let  $w_\delta \equiv  w_*(\bx;\delta),\ E_*=E_*(\delta)$ be the ground state eigenpair  of 
\begin{equation}
  \label{eq:L_delta}
  L_\delta w_*(\bx;\delta) = 0~,   \quad 
  L_\delta \equiv L_0   + \delta V_1 - E_*(\delta)~, \quad
  L_0 \equiv -\Delta~.
\end{equation}
with $\cell$ periodic boundary conditions.

We expand $w_\delta(\bx)$ and $E_*(\delta)$ in a Taylor expansion in $\delta$:
$$
  w_\delta    \equiv   w_0  + \delta w_1  + \delta^2 w_2 + \dots~,  \quad
  E_*(\delta)  =           \delta E_1 + \delta^2 E_2 + \dots~,
$$
where $w_k\equiv w_k(\bx)$ and we set $E_0=0$ since we are interested
in bifurcation from the lowest band edge. 
The first three terms in the hierarchy are
\begin{eqnarray}
  \label{eq:delta0}
  O(\delta^0): && L_0 w_0 \ =\ 0~,  \\
  \label{eq:delta1}
  O(\delta^1): && L_0 w_1 \ =\ (E_1 - V_1)w_0~,  \\
  \label{eq:delta2}
  O(\delta^2): && L_0w_2 \ =\ E_2w_0 + (E_1-V_1)w_1~.
\end{eqnarray}
Corresponding to the lowest band edge
Eq.~(\ref{eq:delta0}) admits a  constant solution
$w_0(\bx)=const$ that spans the kernel of $L_0$.
W.l.o.g. we may choose this constant such that $\mbox{\avg{w_0(\bx)}=1}$.
In order to remove secular growth the non-homogeneous terms in Eqs.~(\ref{eq:delta1})
and~(\ref{eq:delta2}) must be orthogonal to $w_0(\bx)$.
Therefore, their cell-average must vanish.
Removing secular terms at $O(\delta^1)$ gives
$$
  E_1 = \avg{V_1} = 0~, \quad w_1 =  -\Linvz V_1~.
$$
Substituting the above results into Eq.~(\ref{eq:delta2}) 
and removing secular growth yields
$$
   E_2 = - \avg{V_1 \Linvz V_1}~, \quad
   w_2 =  \{ \Linvz ( V_1 \Linvz V_1) \}~,
$$
where  the curly-bracket is a projection symbol defined as
$$
   \{ \ f \ \} \equiv  f - \avg{\ f \ }~.
$$
Summarizing the above results gives
\begin{eqnarray}
   \label{eq:L_delta_2}
   L_\delta              &=& L_0   + \delta V_1  + O (\delta^2)~,\\
   \label{eq:w_delta_2}
     w_\delta  &=&  1 - \delta \Linvz V_1  +
   \delta^2  \{ \Linvz ( V_1 \Linvz V_1 \} +\ O(\delta^3)~, \\
   \label{eq:E_delta_2}
   E_*(\delta)        &=&  - \delta^2 \avg{V_1 \Linvz V_1} \ +\ O(\delta^4 )~.
\end{eqnarray}

We now approximate the inverse operator $\Linvd$. 
It is expedient to make the following definitions.
\begin{definition}
We denote the domains of $\Linvz$ and $\Linvd$ as
\begin{eqnarray*}
  \label{eq:K0}
  K_0 &\equiv& \{ f(\bx)| \bx\in \cell, f \mbox{ is periodic in }
  \cell,  f \in {\rm Ker}^{\bot}  L_0 \}~, \\
  \label{eq:M_delta}
  K_\delta  &\equiv&  \{ f(\bx)| \bx\in \cell, f \mbox{ is periodic in }
  \cell,  f \in {\rm Ker}^{\bot} L_\delta  \}~,
\end{eqnarray*}
respectively. We denote the projection operator into $K_\delta$ as
\begin{equation}
  \label{eq:P_delta}
  P_\delta F_\delta  \equiv  F_\delta  
  - \dfrac{ \avg{ F_\delta ,  w_\delta } } { \avg{ w_\delta, w_\delta } }w_\delta~.
\end{equation}
\end{definition}
Using the above definitions we obtain
\begin{lemma}
  \label{lem:L_delta_inv}
  Let  $F_\delta = F_0 + \delta  F_1  +  O(\delta^2)\,$.
  Then  
  $$
  \Linvd P_\delta F_\delta  = \Linvz \{F_0\}  
  + \delta \left[  \Linvz \{F_1\}  + \avg{F_0} \Linvzs V_1 
    - \Linvz \{ V_1 \Linvz \{F_0\} \}   \right]  + O(\delta^2)~.
  $$
\end{lemma}

\begin{proof}
  The proof of Lemma~\ref{lem:L_delta_inv} follows directly
  from expanding the projection operator~(\ref{eq:P_delta})
  in powers of $\delta$, using Eq.~(\ref{eq:w_delta_2}), 
  $\Linvz : K_0\to K_0$, and $\avg{\Linvz f} = 0$.  
\end{proof}

In the derivations of $\zetap$ and $\zeta_{*1}$,  in each and every
case that $\Linvd$ is applied,  one has $F_0=0$ and $\avg{F_1}=0$. 
Hence,  we shall use
\begin{corollary}
It follows from Lemma~\ref{lem:L_delta_inv} that
\begin{equation}
  \label{eq:L_inv_delta}
  F_0 = 0,~ \avg{F_1}=0 ~\Longrightarrow ~
  \Linvd F_\delta = \delta \Linvz F_1 + O(\delta^2)~.
\end{equation}
\end{corollary}
In addition, the following approximations of the expressions related to the inverse effective mass tensor 
and effective nonlinear coupling constant are used in the approximation of
$\zetas$,
\begin{eqnarray}
  \label{eq:X1_delta}
  \lefteqn{
    X_1^{ij}(\bx)  = \left( \delta^{ij} +
      4\D_{x_j}L_*^{-1}\D_{x_i} -\ \aij \right) w(\bx)
  } \\ \nn &&
  \st{Aij} 4 \left( \D_{x_j}L_*^{-1}\D_{x_i}  +
    \frac{ \avg{ \D_{x_j}w, L_*^{-1}\D_{x_i}w  } } { \avg{w,w} }  \right) w(\bx)
  \\ \nn && 
  ~\st{P_delta}~  P_\delta \left  ( 4   \D_{x_i}\ \Linvd  \D_{x_j}  w_\delta \right )
  ~\stt{w_delta_2}{L_inv_delta}~ -4\delta \D_{x_i} \D_{x_j} \Linvzs V_1 + O(\delta^2)~,
\end{eqnarray}
\begin{equation}
  \label{eq:geff_delta}
  \G \st{geff}    \frac{\int_{\cell} w^{2\sigma+2}\
    d\bx}{\int_{\cell} w^2\ d\bx}   \st{w_delta_2}  1 + O(\delta^2)~,
\end{equation}
and therefore
\begin{equation}
  \label{eq:X2_delta}
  X_2(\bx)  = w^{2\sigma+1}- \G \, w  
  \stt{w_delta_2}{geff_delta}  2\sigma\delta \Linvz V_1 + O(\delta^2)~.
\end{equation}

\vspace{4mm}
We proceed to calculate the power and its slope.
For simplicity we consider the critical case $(d=1, \sigma=2)$.

{\bf Calculation of the power constant $\zetap\ $.}

Expanding, using Eq.~(\ref{eq:w_delta_2})  with $(d=1,
\sigma=2)$, gives
\begin{eqnarray*}
  w_\delta^2 &=& 1 - 2 \delta \Linvz V_1 + 2\delta^2 \{ \Linvz ( V_1  \Linvz V_1 \} 
   + \delta^2 (\Linvz V_1)^2 \ + \ O(\delta^3)~, \\
  w_\delta^6 &=&   1 - 6 \delta \Linvz V_1 + 6 \delta^2 \{ \Linvz ( V_1  \Linvz V_1 \} 
   + 15 \delta^2 (\Linvz V_1)^2 \ + \ O(\delta^3)~.
\end{eqnarray*}
When integrating these functions the contributions of the second and
third terms vanish, as they are in $K_0$.
Therefore, the first factor in $\zetap$ can be approximated by
\begin{equation}
  \label{eq:ratio1}
\sqrt{ \frac{\left( \Av w^2 \right)^3}  {\Av w^6} }  \sim 
\sqrt{ \frac {\Av \left( 1 + 3\delta^2 \Linvz V_1 \right) }  
                  { \Av \left( 1 + 15\delta^2 \Linvz V_1 \right)    } }  \sim 
1 - 6\delta^2 \Av \left( \Linvz V_1 \right)^2~.
\end{equation}
Similarly, the inverse effective mass (Gaussian curvature) is
\begin{eqnarray*}
m_*^{-1} &=& 1 - 4\dfrac{\avg{\D_x w, L_*^{-1} \D_x w} }{\avg{w,w}}
\st{w_delta_2} 1 - 4 \Av (-\delta \D_x \Linvz V_1) (-\delta \D_x
\Linvzs V_1)  + O(\delta^3) \\
&\sim& 1 - 4\delta^2 \Av \left( \D_x L_0^{-3/2} V_1 \right)^2~,
\end{eqnarray*}
where in the last step we used the self-adjointness of $\Linvz$.
The operator in the above integral can be simplified as
$
  \D_x L_0^{-3/2}   \equiv  \Linvz~.
$
Thus,  the second factor in $\zetap$ can be approximated with
$$
\dfrac{1}{\sqrt{m_*}} \sim 1 - 2\delta^2 \Av \left( \Linvz V_1 \right)^2~.
$$
Combining with Eq.~(\ref{eq:ratio1}) yields Eq.~(\ref{eq:zeta_delta_d1}).

\vspace{4mm}
{\bf Calculation of the slope constant $\zetas\ $.}

In one dimension Eq.~(\ref{eq:zeta1*}) reduces to 
\begin{equation}
  \label{eq:zeta1d1}
  \zetas = \int_\R \Av |U_1(x,y)|^2\,dxdy 
  + \Av w^2\,dx \int_\R S(y) \left. \D_\Omega F(y;\Omega)\right|_{\Omega=-1} \,dy~.
\end{equation}
As we shall see, $U_1=O(\delta)$ and $S(y)=O(\delta^2)$.
Therefore, both integral terms are $O(\delta^2)$.
It follows from Eqs.~(\ref{eq:Feqn}), (\ref{eq:X1_delta}), and~(\ref{eq:geff_delta})
that to leading order in $\delta$, $F(y)$ is the Townes mode, i.e. 
\begin{equation}
  \label{eq:F_delta}
  F(\by;\Omega,\delta) = R(\by;\Omega)  + O(\delta^2)~.
\end{equation}
Furthermore, we shall use Eq.~(\ref{eq:R_1d}) with $\sigma=2$ 
to explicitly evaluate the $y$-integrals.

The first integral term~(\ref{eq:zeta1d1}) depends on
\begin{eqnarray}
  \nn
    U_1(x,y) &\st{U1}& 2\Linvd \D_x w \D_y F
    \st{w_delta_2}  - 2\Linvd (\delta \D_x \Linvz V_1)  R_y  +
    O(\delta^2)
    \\
    &\st{L_inv_delta}& -2\delta  \D_x \Linvzs V_1  R_y + O(\delta^2)~.
    \label{eq:U1_delta}
\end{eqnarray}
Therefore,
\begin{eqnarray*}
\lefteqn{
  \int_\R \Av |U_1(x,y)|^2\,dxdy = 
  4\delta^2\Av(\D_x \Linvzs V_1)^2\,dx
  \int_\R R_y^2\,dy + O(\delta^3) 
} \\ &&
  \st{R_1d}
   \sqrt{3}\ \pi \delta^2\Av(\D_x \Linvzs V_1)^2\,dx+ O(\delta^3)~.
\end{eqnarray*}
The operator in the above integral can be simplified  as
$\D_x \Linvzs  = (-\D_{xxx})^{-1}$.
Hence,
\begin{equation}
  \label{eq:zeta1-U1-term}
   \int_\R\Av|U_1(x,y)|^2\,dxdy = 
   \sqrt{3}\ \pi \delta^2\Av \left[  (-\D_{xxx})^{-1} V_1 \right]^2 \,dx+ O(\delta^3)~.
\end{equation}

For the second integral term in~(\ref{eq:zeta1d1}) we need to calculate $S(y)$. 
We  use
\begin{equation}
  \label{eq:U0_delta}
   U_0(x,y) \st{U0}  w(x) F(y) = R(y) + O(\delta^2).
\end{equation}
The first two terms in Eq.~(\ref{eq:Sofydef-2}) are negligible. 
This follows from 
\begin{lemma}
  \label{lem:U_2p}
  $
  \avg { U_{2p} , w } = 0~.
  $
\end{lemma}

\begin{proof}
  Equation~(\ref{U2fac-2})  shows that the $x$-dependence of $U_{2p}$
  is of the form $L_*^{-1} X_k(x)$ for suitable $X_k(x), k=1,2$.
  As $\Linvd$ is into the orthogonal space to $w(x)$, the Lemma follows. 
\end{proof}

That the first term in Eq.~(\ref{eq:Sofydef-2}) is zero follows
immediately from Lemma~\ref{lem:U_2p}. 
The second term in Eq.~(\ref{eq:Sofydef-2}) has an additional $U_0^{2\sigma}$.
However  in light of Eq.~(\ref{eq:U0_delta}) the $x-$dependence of 
$U_0^{2\sigma}$ is constant to leading order. Therefore, 
the second term in $S(y)$ is $O(\delta^3)$.
It remains to calculate the two last terms in $S(y)$. 
Note that the  coefficient preceding the square brackets in Eq.~(\ref{eq:Sofydef-2}) 
cancels with the $w$-integral in Eq.~(\ref{eq:zeta1d1}).
This leaves (assuming $d=1, \sigma=2$)
\begin{equation}
  \label{eq:S_delta}
  S(y) = 10\avg{ w ,  U_0^3U_1^2}  
  \ - \ 2\avg{  \D_x w , \D_y  \tilde{U_3} } \ + \ O(\delta^2)~,  
\end{equation}
where 
\begin{eqnarray*}
  \lefteqn{
    \tilde U_3 \st{tU3def}
    2\Linvz \biggm[\, \D_x \Linvz X_1 \D_{yyy} F 
    \ +\  \D_x \Linvz X_2 \D_y  F^5
   }  \nn\\ &&   \hspace{0.2in} +
   \Linvz \D_x w (\D_{yy} - 1)\D_y  F 
   \ + \ 2 w^4\Linvz \D_x w \D_y  F^5\ \biggm ]
\end{eqnarray*}
and in one-dimension $X_1 \equiv  X_1^{ij}$.
We denote by $S_k$ the various terms in $S(y)$ when $\tilde U_3$
is explicitly inserted into it, and by $I_k$ their corresponding 
contributions to Eq.~(\ref{eq:zeta1d1}).

Using Eqs.~(\ref{eq:w_delta_2}), (\ref{eq:U1_delta}) and~(\ref{eq:U0_delta}),
the first term in Eq.~(\ref{eq:S_delta}) is 
$$ 
 S_1 \equiv 10\avg{ w ,  U_0^3U_1^2}  = 10\avg{  U_1^2 }  + O(\delta^3)~.
$$
Therefore, its contribution to the slope is 10 times the first term in
Eq.~(\ref{eq:zeta1d1}), i.e., 
\begin{equation}
  \label{eq:I1}
  I_1  \st{zeta1-U1-term} 10\sqrt{3}\ \pi \delta^2
  \Av \left[  (-\D_{xxx})^{-1} V_1 \right]^2 \,dx+ O(\delta^3)~.
\end{equation}

The first term arising from substituting $\tilde U_3$ into $S(y)$ is
\begin{eqnarray*}
  S_2 &\equiv& -2\avg{  \D_x w , \D_y  2\Linvz   \D_x\Linvz X_1 \D_{yyyy}F } 
  \\* &\stt{w_delta_2}{X1_delta}&
  2 \avg{ -\delta\Linvz \D_x V_1 ,  2\Linvz   \D_x \Linvz
    \left(-4\delta  \D_{xx} \Linvzs V_1\right) 
  } \D_{yyyy}F  + O(\delta^3)
  \\* &\st{F_delta}&
  -16 \delta^2  \Av \left (  \D_{xx}  L_0^{-5/2} V_1 \right)^2 \,dx \, \D_{yyyy}R + O(\delta^3)~,
\end{eqnarray*}
where in the last step we used the skew self-adjointness of $\D_x$
and the self-adjointness and positivity of $\Linvz$.
Simplifying the operator in the above integral we 
obtain\footnote{Note that the Fourier representations of 
$\D_x \Linvzs$ and $\D_{xx}  L_0^{-5/2}$ are the same only in $d=1$.}
$$
 S_2 =  -16 \delta^2  \Av \left[  (-\D_{xxx})^{-1} V_1 \right]^2  \,dx \,
  \D_{yyyy}R  + O(\delta^3)~.
$$
Substituting $S_2$ into Eq.~(\ref{eq:zeta1d1}) gives
$$
  I_2  = -16 \delta^2 \Av \left[  (-\D_{xxx})^{-1} V_1 \right]^2  \,dx \,
  \int_\R \D_{yyyy}R \left. \D_\Omega R(y;\Omega) \right|_{\Omega=-1} \,dy + O(\delta^3)~.
$$
The following explicit integral can be obtained from Eq.~(\ref{eq:R_1d}) 
$$
  \int_\R \D_{yyyy}R \left. \D_\Omega R(y;\Omega) \right|_{\Omega=-1} \,dy 
  = -\dfrac{11\sqrt{3}\ \pi}{16}~.
$$
Using this gives
\begin{equation}
  \label{eq:I2}
  I_2  =  11\sqrt{3}\ \pi \delta^2 \Av \left[  (-\D_{xxx})^{-1} V_1 \right]^2  \,dx
  \ +  \ O(\delta^3)~.
\end{equation}
Similar calculations can be carried out for $I_3, I_4$ and $I_5$ using 
the explicit integral
$$
\int_\R \D_{yy}R^5 \left. \D_\Omega R(y;\Omega) \right|_{\Omega=-1} \,dy 
  = \dfrac{13\sqrt{3}\ \pi}{16}~.
$$
Thus, to $O(\delta^3)$, we get 
$$   I_3  =  13\sqrt{3}\ \pi \delta^2 \Av \left[  (-\D_{xxx})^{-1} V_1   \right]^2  \,dx~
$$
and $I_5 = - I_4 = \frac{1}{4} I_3$.
Summing the contributions from $I_1 \dots I_5$ gives
$$
\zetas = 34\sqrt{3}\ \pi \delta^2 
\Av \left[  (-\D_{xxx})^{-1} V_1 \right]^2  \,dx  \ +  \ O(\delta^3)>0~.
$$
This concludes the proof of Corollary \ref{cor:edgepower-crit}.

\bibliographystyle{siam}

\end{document}